\documentclass[pra, aps, 10pt, twocolumn ,superscriptaddress]{revtex4-2}
\usepackage{amsfonts}
\usepackage{amsmath}
\usepackage{amssymb}
\usepackage{mathtools}
\usepackage{graphicx}
\usepackage{caption}
\usepackage{subcaption}
\usepackage{float}
\usepackage{physics}
\usepackage{bbold}
\usepackage{tabularx}
\usepackage{multirow}
\usepackage{pbox}
\usepackage[absolute]{textpos}
\usepackage[dvipsnames]{xcolor}
\usepackage[colorlinks=true,linkcolor=blue,citecolor=red,plainpages=false,pdfpagelabels]
{hyperref}
\usepackage{newtxtext,newtxmath}%
\providecommand{\U}[1]{\protect\rule{.1in}{.1in}}
\newcommand{\cmmnt}[1]{}
\newtheorem{theorem}{Theorem}

\newtheorem{corollary}{Corollary}

\newtheorem{definition}{Definition}

\newtheorem{lemma}{Lemma}

\newtheorem{proposition}{Proposition}
\newtheorem{remark}{Remark}

\newenvironment{proof}[1][Proof]{\noindent\textbf{#1.} }{\ \rule{0.5em}{0.5em}}
\allowdisplaybreaks

\begin{document}
\preprint{ }
\title[ ]{Unextendible entanglement of quantum channels}
\author{Vishal Singh}
\affiliation{School of Applied and Engineering Physics, Cornell University, Ithaca, New York 14850, USA}
\affiliation{Hearne Institute for Theoretical Physics, Department of Physics and Astronomy,
and Center for Computation and Technology, Louisiana State University, Baton
Rouge, Louisiana 70803, USA}
\author{Mark M.~Wilde}
\affiliation{School of Electrical and Computer Engineering, Cornell University, Ithaca, New York 14850, USA}
\affiliation{Hearne Institute for Theoretical Physics, Department of Physics and Astronomy,
and Center for Computation and Technology, Louisiana State University, Baton
Rouge, Louisiana 70803, USA}

\begin{abstract}
Quantum communication relies on the existence of high quality quantum channels to exchange information. In practice, however, all communication links are affected by noise from the environment. Here we investigate the ability of quantum channels to perform quantum communication tasks by restricting the participants to use only
local operations and one-way classical communication (one-way LOCC) along with the available quantum
channel. In particular, a channel can be used to distill a highly entangled state between two parties, which further enables quantum or private communication. In this work, we invoke the framework of superchannels  to study the distillation of a resourceful quantum state, such as a maximally entangled state or a private state, using multiple instances of a point-to-point quantum channel. We use the idea of $k$-extendibility to obtain a semidefinite relaxation
of the set of one-way LOCC superchannels and define a class of entanglement measures for quantum channels that
decrease monotonically under such superchannels; therefore these measures, dubbed collectively the ``unextendible entanglement of a
channel'', yield upper bounds on several communication-theoretic quantities of interest in the regimes of resource distillation and zero error. We then generalize the formalism of $k$-extendibility to bipartite superchannels, thus obtaining functions that are monotone under two-extendible superchannels. This allows us to analyze probabilistic distillation of ebits or secret key bits from a bipartite state when using a resourceful quantum channel. Moreover, we propose semidefinite programs  to evaluate several of these quantities, providing a computationally feasible
method of comparison between quantum channels for resource distillation.
\end{abstract}

\date{\today}
\startpage{1}
\endpage{10}
\maketitle
\tableofcontents

\section{Introduction}

Quantum communication technologies revolve around the transmission of quantum data between two spatially separated parties. The promise of a future quantum internet~\cite{10.1145/1039111.1039118,Kimble_2008,science.aam9288} relies on our ability to exchange quantum data and generate highly entangled states~\cite{Horodecki_2009} shared across distant locations. However, it is challenging to realize an ideal quantum communication link between two distant parties, and in practice, we only have a noisy channel to transmit quantum data. Thus, it is crucial to understand our ability to perform quantum communication tasks over a noisy channel in order to recognize the advantages of a realizable quantum network over existing classical networks.

Finding limitations on the rate of communication imposed by the underlying noisy channel has been a topic of interest in both classical~\cite{6773024, el_gamal_kim_2011}  and quantum information theory~\cite{hayashi_2017-book,wilde_2017,watrous_2018,Holevo_2019,khatri2020principles}. The quantum capacity of a channel is equal to the largest rate at which qubits can be reliably transmitted over asymptotically many uses of the channel, such that the error vanishes in this limit. The laws of quantum mechanics also allow for unconditionally secure communication~\cite{BB84,PhysRevLett.67.661,Horodecki_2005, K_Horodecki_2009}, unlike classical networks, which often rely on the computational power of the adversary to implement privacy. This has sparked interest in understanding  limitations on private communication over quantum channels.

While the standard  definitions of quantum capacities allow for arbitrarily small errors in communication that vanish in the asymptotic limit of many channel uses, zero-error communication~\cite{1056798} is a special case in which the quantum channel is required to perform a communication task exactly (i.e., without error). While a noisy channel is incapable of doing so on its own, one can use error-correction protocols along with the channel to reduce the error probability in transmitting data, albeit at the cost of a reduced rate of communication. In some cases, one can use such protocols to reduce the error probability to zero.

Any error-correction protocol can be mathematically described using the language of superchannels~\cite{Chiribella_2008,LW15,  Gour_2019}, a linear map that transforms one quantum channel into another. In doing so, one needs to restrict the allowed superchannels to mimic the physical reality of the protocols. A natural restriction is that the two participants, Alice and Bob, can only perform local quantum operations. In our setting, we also allow Alice to send classical data to Bob. This restricts the set of allowed superchannels to the set of one-way LOCC (local operations and classical communication) superchannels, as done, e.g., in~\cite{LW15,Berta_2021,Holdsworth_2023}. Allowing one-way LOCC is motivated not only by its power, as is evident from the teleportation~\cite{PhysRevLett.70.1895} and super-dense coding~\cite{PhysRevLett.69.2881} protocols, but also by the low cost of classical communication nowadays.

In this paper, as a companion to our recent findings in~\cite{SW24}, we use the concept of unextendibility of a quantum channel~\cite{Kaur_2019, Kaur_2021} in order to quantify its capability for quantum communication. The presence of quantum correlations imposes a fundamental restriction on the extendibility of any quantum resource. This property has been studied in the context of the entanglement content of quantum states~\cite{Werner1989, PhysRevLett.88.187904,Doherty_2014}. Similar notions of $k$-extendibility have been explored for quantum channels~\cite{Kaur_2019, Kaur_2021, Berta_2021}. The unextendibility of a resourceful quantum channel arises from its inability to broadcast the same quantum data to multiple parties and is closely related to the no-cloning theorem~\cite{Park1970, Wootters1982}. Unextendibility has also been studied in some resource-theoretic frameworks~\cite{Kaur_2021,WWW19}, and it has been used to obtain tight bounds on information processing quantities such as one-way distillable entanglement~\cite{Bennett_1996, PhysRevA.59.1025, Pan_2001, Rozp_dek_2018} and distillable secret key~\cite{BB84, Horodecki_2005, K_Horodecki_2009}. Extendible channels can be understood as a relaxation of the set of one-way LOCC channels~\cite{Kaur_2021}, which is one of the foundations upon which the present work builds.  

We also define a family of entanglement measures for a quantum channel based on its unextendibility, which we call the unextendible entanglement of quantum channels. The definition of these measures is motivated from the unextendible entanglement for quantum states, previously  defined in~\cite{WWW19}. We use this entanglement measure to give bounds on multiple quantities of interest in quantum information processing tasks in the regimes of zero error and probabilistic distillation. Additionally, several of our bounds can be computed via semidefinite programs. In what follows, we briefly discuss our contributions in more detail.

First, we give an upper bound on the exact one-way distillable key of a quantum channel, which is roughly defined as the maximum rate of distilling exact secret bits using the channel along with one-way LOCC assistance. We investigate one-shot as well as asymptotic protocols for distributing bipartite private quantum states, from which a secret key can be realized. These upper bounds are practically relevant because quantum key distribution~\cite{PhysRevLett.67.661, BB84, 10.1145/382780.382781} is one of the major advantages of quantum networks over classical networks, as this method ensures unconditional private communication between two parties, based on the laws of quantum mechanics. 

Next we give an upper bound on the exact one-way distillable entanglement of a quantum channel, which is roughly defined as the maximum rate of distilling exact Bell states (ebits) using the channel with one-way LOCC assistance. This is a task of utmost importance to establish an ideal quantum network as a large number of quantum communication tasks, including teleportation, super-dense coding, secret key distillation, etc., rely on distant parties sharing highly entangled states. We investigate limitations on entanglement distribution in the presence of a noisy channel, assisted by local operations and one-way classical communication. 

Our formalism allows us to investigate various capacities of a quantum channel as well. We give an upper bound on the zero-error private capacity of a quantum channel assisted by one-way LOCC superchannels, which is defined as the maximum rate at which secret bits can be transmitted exactly over multiple uses of the channel. Zero-error capacities have been studied in classical information theory extensively~\cite{KO98}, and they have been explored to a lesser degree in quantum information theory~\cite{guedes2016quantum} (see also~\cite{CS12,Shirokov_2015,leung2016maximum} and references therein). We also give an upper bound on the zero-error quantum capacity of an arbitrary quantum channel assisted by a one-way LOCC superchannel, which is defined as the maximum rate at which qubits can be transmitted exactly over multiple uses of a quantum channel. We present a brief summary of our results for point-to-point channels in Table~\ref{tab:p2p_ch_results}.

We also extend our formalism to semicausal bipartite quantum channels~\cite{BGNP01}. Semicausal bipartite quantum channels describe quantum operations that allow only one party to send information, quantum or classical, to the other. Such channels can be used to distill bipartite resourceful states, such as ebits or secret keys, from an existing bipartite noisy resource state. We establish the notion of $k$-extendibility of a bipartite superchannel and define an entanglement measure for bipartite quantum channels based on unextendibility. We use the unextendible entanglement of a bipartite channel to investigate its ability to increase the unextendibility of an existing bipartite state, hence, boosting the resource available for quantum communication between two parties. 

\begin{table}[]
    \begin{center}
        \begin{tabular}{|l|c|l|}
        \hline
        Operational Quantity & Upper bound & Reference\\
        \hline\hline
        &&\\[-0.7em]
         Exact distillable key & $\widehat{E}^u_{\min}\!\left(\mathcal{N}_{A\to B}\right)$ & Corollary~\ref{cor:prob_key_distill_2ext} \\
        &&\\[-0.5em]
        Exact distillable entanglement & $\widehat{E}^u_{\min}\!\left(\mathcal{N}_{A\to B}\right)$ & Corollary~\ref{cor:distill_ent_ub_fin}\\
        &&\\[-0.5em]
        Zero-error private capacity & $\widehat{E}^u_{\min}\!\left(\mathcal{N}_{A\to B}\right)$ & Corollary~\ref{cor:0_err_priv_cap_ub}\\
        &&\\[-0.5em]
        Zero-error quantum capacity & $\widehat{E}^u_{\min}\!\left(\mathcal{N}_{A\to B}\right)$ & Corollary~\ref{cor:0_err_cap_ub}\\[0.5em]
        \hline
    \end{tabular}
    \end{center}
    \caption{A list of our results for point-to-point quantum channels. We give upper bounds on several quantities related to quantum communication, as well as private communication, over a quantum channel in terms of the unextendible entanglement of the channel. In each of these scenarios, the quantum channel is denoted by $\mathcal{N}_{A\to B}$, and local operations and forward classical communication from Alice to Bob are allowed for free.}
    \label{tab:p2p_ch_results}
\end{table}

\begin{table*}[]
    \begin{center}
        \begin{tabular}{|l|c|l|}
        \hline
        Operational Quantity & Upper bound & Reference\\
        \hline\hline
        &&\\[-0.7em]
        Probabilistic distillable entanglement  & $\frac{1}{n}\widehat{E}^u\!\left(\rho_{AB}\right)+\widehat{E}^u\!\left(\mathcal{N}_{AB\to A'B'}\right)$ & Proposition~\ref{theo:distill_ent_ch_st_ub} \\
        &&\\[-0.25em]
        Exact distillable entanglement & $\frac{1}{n}\widehat{E}^u_{\min}\!\left(\rho_{AB}\right) + \widehat{E}^u_{\min}\!\left(\mathcal{N}_{AB\to A'B'}\right)$ & Proposition~\ref{theo:distill_ent_ch_st_ub_ex}\\
        &&\\[-0.25em]
        Probabilistic distillable key & $\frac{1}{n}\widehat{E}^u\!\left(\rho_{AB}\right) + \widehat{E}^u\!\left(\mathcal{N}_{AB\to A'B'}\right)$ & Proposition~\ref{theo:distill_key_ch_st_ub}\\
        &&\\[-0.25em]
        Exact distillable key & $\frac{1}{n}\widehat{E}^u_{\min}\!\left(\rho_{AB}\right) + \widehat{E}^u_{\min}\!\left(\mathcal{N}_{AB\to A'B'}\right)$ & Proposition~\ref{theo:distill_key_ch_st_ub_ex}\\[0.5em]
        \hline
    \end{tabular}
    \end{center}
    \caption{A list of our results for bipartite quantum channels. We give upper bounds on the non-asymptotic probabilistic and zero-error distillable entanglement and distillable key of a bipartite quantum state $\rho_{AB}$ and $n$ instances of a bipartite quantum channel $\mathcal{N}_{AB\to A'B'}$, in terms of the unextendible entanglement of the state and channel.}
    \label{tab:bip_ch_results}
\end{table*}

The task of entanglement distillation and secret key distillation from a bipartite quantum state using only local operations and classical communication has been a subject of interest for some time in quantum information theory. Several interesting bounds have been obtained for resources that can be distilled using only local operations and one-way classical communication~\cite{WWW19}. We look at a more general setting in which the two parties involved have access to a bipartite quantum channel that is not necessarily simulable by local operations and one-way classical communication, and we employ the unextendible entanglement of bipartite channels to upper bound the distillable entanglement and distillable secret key of a bipartite channel used in conjunction with a bipartite resource state, along with local operations and classical communication. Table~\ref{tab:bip_ch_results} presents a brief summary of our results for bipartite semicausal channels. 

We consider the example of the erasure channel and the depolarizing channel to demonstrate our results for the point-to-point case. We show that the exact one-way distillable key, exact one-way distillable entanglement, forward-assisted zero-error private capacity, and forward-assisted zero-error quantum capacity of the erasure and the depolarizing channel are all equal to zero. More generally, we show that a quantum channel with a full-rank Choi operator cannot be used for exact entanglement distillation or exact key distillation when employing  one-way LOCC superchannels, and the forward-assisted zero-error quantum capacity as well as the forward-assisted zero-error private capacity of such channels are equal to zero.

To demonstrate our results for bipartite quantum channels, we consider an extension of the erasure channel in which either Bob's system gets erased and Alice retains the state she was trying to send to Bob, or Bob receives the state and Alice receives the erasure symbol. We give an analytical expression for the unextendible entanglement of this channel induced by the Belavkin--Staszewski relative entropy, which is an upper bound on the probabilistic one-way distillable entanglement and probabilistic one-way distillable key of the channel. We also consider a less idealistic setting where Alice receives some classical information indicating if the quantum data she sent to Bob was erased or not. Using a semidefinite program, we compute upper bounds on the probabilistic one-way distillable entanglement and probabilistic one-way distillable key of this channel.

This paper is organized as follows:
\begin{itemize}
    \item Section~\ref{sec:prelimnaries}: Definitions and notations used in the paper along with preliminary information on quantum states, channels, and superchannels.
    \item Section~\ref{sec:k-extendibility}: Background on $k$-extendibility for bipartite states, point-to-point channels, and superchannels.
    \item Section~\ref{sec:unext_ent}: Background on divergences of quantum states and channels, and formal definition of the unextendible entanglement of channels. 
    \item Section~\ref{sec:applications}: Applications of the unextendible entanglement in establishing upper bounds on exact one-way distillable key and exact one-way disillable entanglement of a channel in the one-shot and asymptotic settings, and establishing upper bounds on the forward assisted zero-error quantum capacity and forward assisted zero-error private capacity of channels.
    \item Section~\ref{sec:unext_ent_bip}: Generalizes the notion of $k$-extendibility to bipartite superchannels and defines the unextendible entanglement of bipartite quantum channels.
    \item Section~\ref{sec:applications_bip}: Applications of the unextendible entanglement of bipartite quantum channels in bounding the probabilistic one-way distillable entanglement and probabilistic one-way distillable key of a bipartite state-channel pair.
    \item Section~\ref{sec:numerical_calc}: Semidefinite program to calculate the unextendible entanglement of point-to-point and bipartite channels induced by the $\alpha$-geometric R\'enyi relative entropy. Analytical and numerical calculation of the unextendible entanglement of special channels induced by geometric R\'enyi relative entropies.
\end{itemize}  

\section{Preliminaries}\label{sec:prelimnaries}

In this section we establish some background on the three major elements that we use in the rest of the work: quantum states, channels, and superchannels. 

\subsection{Quantum states and channels}

A quantum state $\rho_A$ is a positive semidefinite, unit-trace operator acting on a Hilbert space $\mathcal{H}_A$. All linear operators acting on the Hilbert space $\mathcal{H}_A$ form the set $\mathcal{L}(A)$, and  all quantum states acting on this Hilbert space form the set $\mathcal{S}(A)$. A bipartite quantum state $\rho_{AB}$ acting on the Hilbert space $\mathcal{H}_{A}\otimes \mathcal{H}_B$ is called separable if it can be written as
\begin{equation}
	\rho_{AB} = \sum_x p(x)\sigma^x_{A}\otimes\tau^x_{B},
 \label{eq:def-sep-state}
\end{equation}
where $\{p(x)\}_x$ is a probability distribution and $\{\sigma^x_{A}\}_x$ and $\{\tau^x_{B}\}_x$ are sets of states. Any quantum state that is not separable is said to be entangled. The $d$-dimensional maximally entangled state vector on the Hilbert space $\mathcal{H}_{A}\otimes\mathcal{H}_B$ is
\begin{equation}
	|\Phi^d\rangle_{AB} \coloneqq \frac{1}{\sqrt{d}}\sum_{i=0}^{d-1} |i\rangle_A|i\rangle_B,
\end{equation}
where $\frac{1}{\sqrt{d}}$ is the normalizing factor and $\{|i\rangle\}_{i=0}^{d-1}$ is an orthonormal basis. The corresponding density operator is written as $\Phi^d_{AB} \equiv |\Phi^d\rangle\!\langle \Phi^d|_{AB}$. 

A quantum channel $\mathcal{N}_{A\to B}$ is a completely positive (CP) and trace preserving (TP) linear map that takes an operator acting on the Hilbert space $\mathcal{H}_A$ as input and outputs an operator acting on the Hilbert space $\mathcal{H}_B$. Let $\Gamma^{\mathcal{N}}_{RB}$ denote the Choi operator of the quantum channel $\mathcal{N}_{A\to B}$, which is defined as follows:
\begin{equation}
	\Gamma^{\mathcal{N}}_{RB} \coloneqq \mathcal{N}_{A\to B}\!\left(\Gamma_{RA}\right),
\end{equation}
where $\Gamma_{RA}\coloneqq d \Phi^d_{RA}$ is the unnormalized maximally entangled operator.

Throughout our paper, we have to consider  extensions of quantum states and channels. We define the set of relevant extensions of a quantum state $\rho$ by $\operatorname{Ext}\!\left(\rho\right)$ and the set of relevant extensions of a quantum channel $\mathcal{N}$ by $\operatorname{Ext}\!\left(\mathcal{N}\right)$. The precise definitions of these sets are given later in~\eqref{eq:extensions_state},~\eqref{eq:ext_set_p2p_ch}, and~\eqref{eq:ext_set_bip_ch}. In the rest of the work, we use the abbreviation $B_S$ for a joint system that contains the isomorphic systems $\{B_i\}_i$ for all values of $i$ in a subset $S$ of all positive integers. For a positive integer $k$, we use the shorthand $[k]$ for the set $\{1,2,\ldots,k\}$, and the shorthand $[k]\setminus i$ for the set $[k]\setminus \{i\}$.

\subsection{Quantum superchannels}

A quantum superchannel $\Theta_{(A\to B)\to (C\to D)}$ is a linear map that transforms a quantum channel $\mathcal{N}_{A\to B}$  to another quantum channel $\mathcal{M}_{C\to D}$. By definition, a superchannel is a completely CPTP preserving map (see Definition~\ref{def:superch} for a formal definition). It can be perceived as a mathematical model for any physical transformation that a quantum channel can undergo, as long as the resulting map is also a quantum channel. Quantum superchannels were introduced in~\cite{Chiribella_2008} and further investigated in~\cite{Gour_2019}, both of which provide a detailed discussion. Below we include a short review on superchannels that is relevant for this work.

\begin{definition}[Superchannel]
\label{def:superch}
Let $\mathcal{T}_{A\to B} : \mathcal{L}(A) \to \mathcal{L}(B)$ be a linear map. Let the space of all such maps be denoted by $\mathbb{L}^{AB}$. 
A linear map $\Theta_{(A\to B)\to (C\to D)}: \mathbb{L}^{AB} \to \mathbb{L}^{CD}$ is a superchannel if 
\begin{enumerate}
    \item It is completely CP preserving; i.e., 
    
      $(\operatorname{id}_{(E)\to(E)}\otimes\Theta_{(A\to B)\to (C\to D)}) (\mathcal{T}_{EA \to EB})$ is a CP map if $\mathcal{T}_{EA \to E'B}$ is a CP map, for all dimensions of system E.
    
    \item It is TP preserving; i.e.,\\
        $\Theta_{(A\to B)\to (C\to D)}(\mathcal{T}_{A\to B})$  is a TP map if $\mathcal{T}_{A\to B}$ is a TP map.
\end{enumerate}
\end{definition}

The fundamental theorem of superchannels states that every superchannel can be decomposed into a pre-processing channel $\mathcal{E}_{C\to MA}$ and a post-processing channel $\mathcal{D}_{MB\to D}$ connected by a memory system $M$~\cite{Chiribella_2008}; i.e., for every superchannel $\Theta_{(A\to B)\to (C\to D)}$, there exist $\mathcal{E}_{C\to MA}$ and $\mathcal{D}_{MB\to D}$ such that
\begin{equation}\label{eq:superchannel_fund_theo}
    \Theta_{(A\to B)\to (C\to D)}(\mathcal{N}_{A\to B}) = \mathcal{D}_{MB\to D}\circ\mathcal{N}_{A\to B}\circ\mathcal{E}_{C\to MA}.
\end{equation}

\begin{figure}
	\centering
	\begin{subfigure}{\linewidth}
		\centering
		\includegraphics[width = \linewidth]{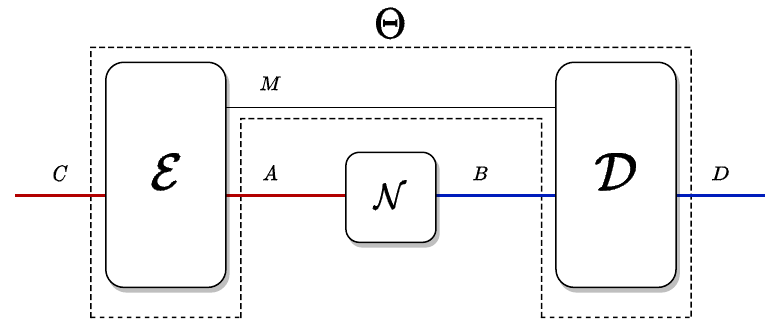}
		\label{fig:basic_superchannel}
	\end{subfigure}
	\begin{subfigure}{\linewidth}
		\centering
		\includegraphics[width = 0.6\linewidth]{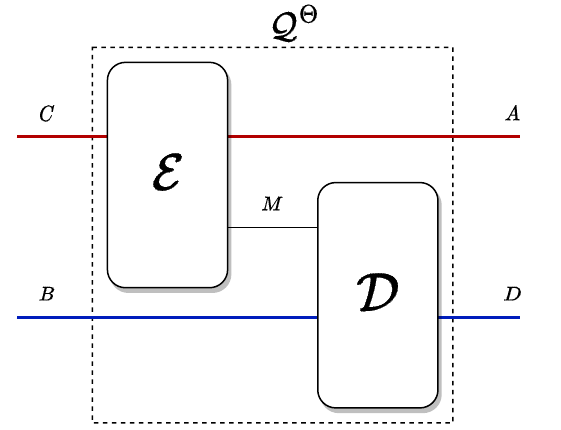}
		\label{fig:unique_CPTP_map}
	\end{subfigure}
	\caption{The figure on the top shows the decomposition of a superchannel $\Theta_{(A\to B)\to (C\to D)}$ into a pre-processing channel $\mathcal{E}_{C\to AM}$ and a post-processing channel $\mathcal{D}_{BM\to D}$ connected by a memory system $M$. The figure on the bottom shows the composition of the unique bipartite channel $\mathcal{Q}^{\Theta}_{CB\to AD}$ associated with a superchannel $\Theta_{(A\to B)\to (C\to D)}$.}
	\label{fig:superchannel_decompose}
\end{figure}

The pre-processing and post-processing channels are not unique to the superchannel. For example, we can introduce an isometric channel  $\mathcal{V}$ and its corresponding reversal map $\mathcal{V}^{\dag}$ acting on the memory system $M$ without changing the superchannel. However, a unique bipartite channel is associated with every superchannel. 

\begin{theorem}[\cite{Gour_2019}]
    A superchannel with the following decomposition
    \begin{equation}
        \Theta_{(A\to B)\to (C\to D)}(\mathcal{N}_{A\to B}) = \mathcal{D}_{MB\to D}\circ\mathcal{N}_{A\to B}\circ\mathcal{E}_{C\to MA},
    \end{equation}
    has a unique bipartite quantum channel associated with it    \begin{equation}\label{eq:superchannel_CPTP_decomp}
        \mathcal{Q}^{\Theta}_{CB\to AD} \coloneqq \mathcal{D}_{MB\to D}\circ\mathcal{E}_{C\to AM}.
    \end{equation}
\end{theorem}

We define the Choi operator of a superchannel using the unique bipartite quantum channel from~\eqref{eq:superchannel_CPTP_decomp}:
\begin{equation}
    \Gamma^{\Theta}_{A'D'CB} \coloneqq (\operatorname{id}_{A'D'}\otimes\mathcal{Q}^{\theta}_{CB\to AD})(\Gamma_{A'C}\otimes\Gamma_{D'B}),
\end{equation}
where $\Gamma_{A'C} \equiv |\Gamma\rangle\!\langle \Gamma |_{A'C}$ is defined from the unnormalized maximally entangled vector:
\begin{equation}
    \ket{\Gamma}_{A'C} \coloneqq \sum_{i}|i\rangle_{A'}|i\rangle_C.
\end{equation}
It suffices to choose  systems $A'$ and $D'$ to be isomorphic to the systems $A$ and $D$, respectively. The following theorem, established as Theorem~1 of~\cite{Gour_2019}, provides conditions on the Choi operator in order for it to correspond to a legitimate superchannel:
\begin{theorem}[\cite{Gour_2019}]\label{theo:superch_SDP}
    The Choi operator of a superchannel $\Theta_{(A\to B)\to (C\to D)}$ satisfies the following constraints:
   \begin{align}
        \Gamma^{\Theta}_{ADCB} &\ge 0,\\
        \operatorname{Tr}_{AD}[\Gamma_{ADCB}^\Theta] &= I_{CB},\\
        \operatorname{Tr}_D[\Gamma_{ADCB}^\Theta] &= \frac{1}{d_B}\operatorname{Tr}_{BD}[\Gamma_{ADCB}^\Theta]\otimes I_{B}.
    \end{align}
\end{theorem}

In the above, the first condition corresponds to the completely CP preserving condition, the second condition corresponds to the TP preserving condition, and the last condition corresponds to the nonsignaling constraint. 

The Choi operator of the input channel $\mathcal{N}_{A\to B}$ and output channel $\mathcal{M}_{C\to D}$ of a superchannel $\Theta_{(A\to B)\to (C\to D)}$ are related through the following propagation rule~\cite{Chiribella_2008, Gour_2019}:
\begin{equation}\label{eq:prop_rule}
    \Gamma^{\mathcal{M}}_{CD} = \operatorname{Tr}_{AB}[T_{AB}(I_{CD}\otimes\Gamma^{\mathcal{N}}_{AB})\Gamma^{\Theta}_{ADCB}],
\end{equation}
where $T_{AB}$ is the partial transpose map acting on systems $A$ and $B$.

\section{\texorpdfstring{$k$}{k}-extendibility}\label{sec:k-extendibility}

In this section, we review the concepts of $k$-extendible states~\cite{Werner1989, PhysRevLett.88.187904, PhysRevA.69.022308} and channels~\cite{Kaur_2019, Kaur_2021}, and we establish the notion of $k$-extendible superchannels. The framework of two-extendible superchannels was employed in~\cite{Berta_2021} and~\cite{Holdsworth_2023} for the purpose of analyzing quantum error correction. We present a more general discussion for arbitrary $k$ in this section. 

Several quantities of interest in quantum information theory are defined in terms of optimizations over the set of separable states, but it is computationally hard to optimize over this set~\cite{Gur03, Gha10}. The framework of $k$-extendibility allows us to approximate the set of separable states in terms of a larger set that contains all separable states. The set of $k$-extendible states is described by semidefinite constraints, and  optimizations over this set are possible using semidefinite programs. The notions of $k$-extendible channels and superchannels are developed to circumnavigate the computational difficulty that arises when optimizing over the set of one-way LOCC channels~\cite{Gur03,Gha10}.

\subsection{\texorpdfstring{$k$}{k}-extendible states}

Let us first recall the definition of a $k$-extendible state~\cite{Werner1989, PhysRevLett.88.187904, PhysRevA.69.022308}. For a positive integer $k\ge 2$, a bipartite  state~$\rho_{AB}$ is $k$-extendible with respect to system $B$ if the following conditions hold:
\begin{enumerate}
    \item There exists an extension state $\omega_{AB_{[k]}}$ such that
    \begin{equation}
        \operatorname{Tr}_{B_{[k]\setminus 1}}[\omega_{AB_{[k]}}] = \rho_{AB},
    \end{equation}
    where each system $B_{i}$ is isomorphic to $B$, for all $i\in [k]$.
    
    \item The extended state is invariant under permutations of the $B$ systems, i.e.,
    \begin{equation}
        \omega_{AB_{[k]}} = W^{\pi}_{B_{[k]}}\omega_{AB_{[k]}}W^{\pi\dagger}_{B_{[k]}} \qquad \forall \pi \in S_k,
    \end{equation}
    where $W^{\pi}_{B_{[k]}}$ is the unitary permutation operator corresponding to the permutation $\pi$ in the symmetric group~$S_k$.
\end{enumerate}

The set of $k$-extendible states is a semidefinite relaxation of the set of separable states. The main idea behind this formulation is that the higher the entanglement content between the systems $A$ and $B$, the lower the number of systems $B_i$, each isomorphic to $B$, that can share correlations in the same way with~$A$. 

The weakest approximation of the set of separable states in this family is the set of two-extendible states. It is straightforward to see that every separable state, written as in~\eqref{eq:def-sep-state}, is two extendible with the extension \mbox{$\omega_{AB_1B_2} = \sum_xp(x)\sigma^x_A\otimes\tau^x_{B_1}\otimes\tau^x_{B_2}$}. However, not all two-extendible states are separable. As an example, sending one share of a maximally entangled state through an erasure channel characterized by erasure probability 1/2 results in the following two-extendible state that is not separable:
\begin{equation}
     \frac{1}{2}\Phi^d_{AB} + \frac{1}{2}\pi_A \otimes|e\rangle\!\langle e|_B,
\end{equation}
where $\pi_A \coloneqq I_{A} / d$ and $|e\rangle\!\langle e|_B$ is the erasure symbol. 

\subsection{\texorpdfstring{$k$}{k}-extendible channels}

For a positive integer $k\ge 2$, a point-to-point quantum channel $\mathcal{N}_{A\to B}$ is $k$-extendible with respect to $B$ if the following conditions hold~\cite{PBHS11}:
\begin{enumerate}
    \item There exists an extension channel $\mathcal{P}_{A\to B_{[k]}}$ such that
    \begin{equation}
        \operatorname{Tr}_{B_{[k]\setminus 1}}\circ\mathcal{P}_{A\to B_{[k]}} = \mathcal{N}_{A\to B}.
    \end{equation}
    \item The extended channel is invariant under permutations of the $B$ systems, i.e.,
    \begin{equation}
        \mathcal{W}^{\pi}_{B_{[k]}}\circ\mathcal{P}_{A\to B_{[k]}} = \mathcal{P}_{A\to B_{[k]}} \qquad \forall \pi \in S_k,
        \label{eq:perm-covariance-broadcast}
    \end{equation}
    where $\mathcal{W}^{\pi}_{B_{[k]}}$ is the unitary channel that permutes the~$B$ systems by the permutation $\pi$ in the symmetric group~$S_k$.
\end{enumerate}

\begin{remark}
    The Choi state of a point-to-point $k$-extendible channel is a $k$-extendible state.
\end{remark}

If a $k$-extendible state is input to a $k$-extendible channel, the output state is also $k$-extendible. This makes the definition of $k$-extendible channels consistent in a resource-theoretic approach; a free operation ($k$-extendible channel) cannot turn a free state ($k$-extendible state) into a resource.

\begin{proposition}[\cite{Kaur_2019}]
    All point-to-point one-way LOCC channels are $k$-extendible for all $ k \ge 2$.
\end{proposition}

\begin{proof}
    The action of a one-way LOCC channel can be written as
    \begin{equation}
        \mathcal{N}_{A\to B}(\rho_{A}) = \sum_x \operatorname{Tr}\!\left[\Lambda^x_A\rho_{A}\right]\sigma^x_B,
    \end{equation}
    where $\{\Lambda^x\}_x$ is a positive operator-valued measure (POVM) and $\left\{\sigma^x_B\right\}_x$ is a set of quantum states. This is also the same as an entanglement-breaking channel~\cite{HSR03}.

    Consider an extension of this channel that acts as follows:
    \begin{equation}
        \mathcal{P}_{A\to B_{[k]}}\!\left(\rho_{A}\right) \\ \coloneqq  \sum_x\operatorname{Tr}\!\left[\Lambda^x_A\rho_{A}\right]\sigma^x_{B_1}\otimes\cdots \otimes\sigma^x_{B_k}.
    \end{equation}
    This is a valid quantum channel because $x$ encodes classical data that can be copied any number of times. This extension obeys the  permutation covariance conditions. Hence, $\mathcal{N}_{A\to B}$ is a $k$-extendible channel. 
\end{proof}

\subsection{\texorpdfstring{$k$}{k}-extendible superchannels}\label{sec:k_ext_superch}

We further build upon the notion of extendibility and define the set of $k$-extendible superchannels. For a positive integer $k\ge 2$, a superchannel $\Theta_{(A\to B)\to (C\to D)}$, associated with the unique bipartite channel $\mathcal{Q}^{\Theta}_{CB\to AD}$, is $k$-extendible if the following conditions hold.
\begin{enumerate}
    \item There exists an extension superchannel $\Upsilon_{(A\to B_{[k]})\to (C\to D_{[k]})}$, with associated unique quantum channel $\mathcal{Q}^{\Upsilon}_{CB_{[k]}\to AD_{[k]}}$, such that,
    \begin{equation}
        \operatorname{Tr}_{D_{[k]\setminus 1}}\circ\mathcal{Q}^{\Upsilon}_{CB_{[k]}\to AD_{[k]}} = \mathcal{Q}^{\Theta}_{CB\to AD}\circ\operatorname{Tr}_{B_{[k]\setminus 1}}.
    \end{equation}
    \item $\mathcal{Q}^{\Upsilon}_{CB_{[k]}\to AD_{[k]}}$ is covariant with respect to permutations of the input systems $B_{[k]}$ and output systems $D_{[k]}$, i.e.,
    \begin{multline}
        \mathcal{W}^{\pi}_{D_{[k]}}\circ\mathcal{Q}^{\Upsilon}_{CB_{[k]}\to AD_{[k]}} = \mathcal{Q}^{\Upsilon}_{CB_{[k]}\to AD_{[k]}}\circ\mathcal{W}^{\pi}_{B_{[k]}}, \\\forall \pi \in S_k
    \end{multline}
    where $\mathcal{W}^{\pi}_{D_{[k]}}$ and $\mathcal{W}^{\pi}_{B_{[k]}}$ are unitary channels representing the permutation $\pi$.
\end{enumerate}
This definition of $k$-extendible superchannels is consistent with the definition of two-extendible superchannels given in~\cite{Holdsworth_2023} and the notion of k-extendible channels considered in~\cite{Kaur_2019,Kaur_2021}.

Let us consider a specific decomposition of a superchannel $\Theta_{(A\to B)\to (C\to D)}$ in terms of a pre-processing channel $\mathcal{E}^{\Theta}_{C\to AM}$ and a post-processing channel $\mathcal{D}^{\Theta}_{MB\to D}$. Neither of the conditions implying $k$-extendibility of a superchannel involve systems $A,C$, and $M$. Thus, the conditions of $k$-extendibility of a superchannel can be reduced to conditions on the post-processing channel only. 

\begin{proposition}
	The $k$-extendibility conditions for a superchannel are semidefinite constraints on the Choi operators of the superchannel $\Theta_{(A\to B)\to(C\to D)}$ and its extension $\Upsilon_{(A\to B_{[k]})\to (C\to D_{[k]})}$. Along with the conditions in Theorem~\ref{theo:superch_SDP}, the Choi operators $\Gamma^{\Theta}_{ADCB}$ and $\Gamma^{\Upsilon}_{AD_{[k]}CB_{[k]}}$ satisfy the following:
	\begin{align}
		\operatorname{Tr}_{D_{[k]\setminus 1}}\!\left[\Gamma^{\Upsilon}_{CB_{[k]}AD_{[k]}}\right] &= \Gamma^{\Theta}_{CBAD}\otimes\frac{I_{B_{[k]\setminus 1}}}{d^{k-1}_B} , \label{eq:two_ext_supch_non_sig_choi} \\
		W^{\pi}_{D_{[k]}} \Gamma^{\Upsilon} W^{\pi\dagger}_{D_{[k]}} &= W^{\pi\dag}_{B_{[k]}} \Gamma^{\Upsilon}  W^{\pi}_{B_{[k]}} \quad \forall \pi \in S_k.
	\end{align}
\end{proposition}

\begin{proposition}\label{prop:two_ext_supch_ch_marginal}
    Let $\Theta_{(A\to B)\to (C\to D)}$ be a $k$-extendible superchannel  with the $k$-extension $\Upsilon_{(A\to B_{[k]})\to (C\to D_{[k]})}$. Then each marginal of the channel obtained by acting with the superchannel $\Upsilon$ on a channel $\mathcal{P}_{A\to B_{[k]}}$ is the same as the channel obtained by acting with the superchannel $\Theta$ on the respective marginal of the channel $\mathcal{P}_{A\to B_{[k]}}$. That is,
  \begin{multline}\label{eq:two_ext_supch_ch_marginal}
        \operatorname{Tr}_{D_{[k]\setminus i}}\circ\left(\Upsilon_{(A\to B_{[k]})\to (C\to D_{[k]})}\!\left(\mathcal{P}_{A\to B_{[k]}}\right)\right) \\= \Theta_{(A\to B)\to (C\to D)}\!\left(\operatorname{Tr}_{B_{[k]\setminus i}}\circ\mathcal{P}_{A\to B_{[k]}}\right) ~\forall i\in [k].
    \end{multline}
\end{proposition}

\begin{proof}
    Let $\Theta_{(A\to B)\to (C\to D)}$ be a $k$-extendible superchannel with the $k$-extension $\Upsilon_{(A\to B_{[k]})\to (C\to D_{[k]})}$. As such,
    \begin{equation}
        \operatorname{Tr}_{D_{[k]\setminus i}}\circ\mathcal{Q}^{\Upsilon}_{CB_{[k]}\to AD_{[k]}} = \mathcal{Q}^{\Theta}_{CB\to AD}\circ\operatorname{Tr}_{B_{[k]\setminus i}} ~\forall i\in [k],
    \end{equation}
    where $\mathcal{Q}^{\Theta}_{CB\to AD}$ and $\mathcal{Q}^{\Upsilon}_{CB_{[k]}\to AD_{[k]}}$ are the unique quantum channels associated with the respective superchannels above. Let $\mathcal{N}_{A\to B}$ be a marginal of the channel $\mathcal{P}_{A\to B_{[k]}}$:
    \begin{equation}
        \operatorname{Tr}_{B_{[k]\setminus i}}\circ\mathcal{P}_{A\to B_{[k]}} = \mathcal{N}_{A\to B} \qquad \forall i\in [k].
    \end{equation}
    
    The Choi operator of the channel obtained by acting with the superchannel $\Theta_{(A\to B)\to (C\to D)}$ on $\mathcal{N}_{A\to B}$ is
    \begin{equation}
        \Gamma^{\Theta\left[\mathcal{N}\right]}_{CD} = \operatorname{Tr}_{AB}\!\left[T_{AB}\!\left(\Gamma^{\mathcal{N}}_{AB}\right)\Gamma^{\Theta}_{CBAD}\right],
    \end{equation}
    where $\Gamma^{\mathcal{N}}_{AB}$ and $\Gamma^{\Theta}_{ADCB}$ are the Choi operators of the channel $\mathcal{N}_{A\to B}$ and the superchannel $\Theta_{(A\to B)\to (C\to D)}$, respectively.
    
    Similarly, the Choi operator of the channel obtained by acting with the superchannel $\Upsilon_{(A\to B_{[k]})\to (C\to D_{[k]})}$ on the channel $\mathcal{P}_{A\to B_{[k]}}$ is,
    \begin{equation}
        \Gamma^{\Upsilon\left[\mathcal{P}\right]}_{CD_{[k]}} = \operatorname{Tr}_{AB_{[k]}}\!\left[T_{AB_{[k]}}\!\left(\Gamma^{\mathcal{P}}_{AB_{[k]}}\right)\Gamma^{\Upsilon}_{CB_{[k]}AD_{[k]}}\right].
    \end{equation}
    This channel has the marginal,
    \begin{align}
        &\operatorname{Tr}_{D_{[k]\setminus i}}\!\left[\Gamma^{\Upsilon\left[\mathcal{P}\right]}_{CD_{[k]}}\right]\notag \\ 
        &= \operatorname{Tr}_{AB_{[k]}D_{[k]\setminus i}}\!\left[T_{AB_{[k]}}\!\left(\Gamma^{\mathcal{P}}_{AB_{[k]}}\right)\Gamma^{\Upsilon}_{CB_{[k]}AD_{[k]}}\right]\\
        &= \operatorname{Tr}_{AB_{[k]}}\!\left[T_{AB_{[k]}}\!\left(\Gamma^{\mathcal{P}}_{A_{[k]}}\right)\operatorname{Tr}_{D_{[k]\setminus i}}\!\left[\Gamma^{\Upsilon}_{CB_{[k]}AD_{[k]}}\right]\right]\\
        &= \frac{1}{d^{k-1}_{B}}\operatorname{Tr}_{AB_{[k]}}\!\left[T_{AB_{[k]}}\!\left(\Gamma^{\mathcal{P}}_{AB_{[k]}}\right)\left(\Gamma^{\Theta}_{CB_iAD_i}\otimes I_{B_{[k]\setminus i}}\right)\right] \label{eq:using_non_sig}\\
        &= \operatorname{Tr}_{AB_i}\!\left[T_{AB_i}\!\left(\Gamma^{\mathcal{N}}_{AB_i}\right)\Gamma^{\Theta}_{CB_iAD_i}\right]\label{eq:using_ch_marginal}\\
        &= \Gamma^{\Theta\left[\mathcal{N}\right]}_{CD_i},
    \end{align}
    where the equality in~\eqref{eq:using_non_sig} follows from the non-signaling condition in~\eqref{eq:two_ext_supch_non_sig_choi} and the equality in~\eqref{eq:using_ch_marginal} follows from the fact that $\mathcal{P}_{A\to B_{[k]}}$ is an extension of the channel $\mathcal{N}_{A\to B}$. Hence, we conclude~\eqref{eq:two_ext_supch_ch_marginal}. 
\end{proof}

\subsubsection{One-way LOCC superchannels}

We briefly review one-way LOCC superchannels to establish their connection with $k$-extendible superchannels. 

A superchannel $\Theta_{(A\to B)\to (C\to D)}$ is called one-way LOCC if it can be implemented using local operations and one-way classical communication only. The action of a one-way LOCC superchannel can be described by
\begin{equation}
    \Theta_{(A\to B)\to (C\to D)}(\mathcal{N}_{A\to B}) = \sum_x \mathcal{D}^x_{B\to D}\circ\mathcal{N}_{A\to B}\circ\mathcal{E}^x_{C\to A},
\end{equation}
where $\{\mathcal{E}^x_{C\to A}\}_x$ is a set of CP maps such that the sum map $\sum_x\mathcal{E}^x_{C\to A}$ is trace preserving and $\{\mathcal{D}^x_{B\to D}\}_x$ is a set of quantum channels. This can be interpreted as the memory system~$M$ being a purely classical system $X$ and the pre-processing and post-processing channels being
\begin{align}
	\mathcal{E}_{C\to AX}\!\left(\rho_C\right) &= \sum_x \mathcal{E}^x_{C\to A}\!\left(\rho_C\right)\otimes |x\rangle\!\langle x|_X,\\
	\mathcal{D}_{BX\to D}\!\left(\sigma_{BX}\right) &= \sum_x \mathcal{D}^x_{B\to D}\!\left(\langle x|\sigma_{BX}|x\rangle_X\right),
\end{align}
such that
\begin{equation}
	\Theta_{(A\to B)\to (C\to D)}(\mathcal{N}_{A\to B}) = \mathcal{D}_{BX\to D}\circ\mathcal{N}_{A\to B}\circ\mathcal{E}_{C\to AX}.
\end{equation}

The unique bipartite channel associated with a one-way LOCC superchannel is of the form
\begin{equation}
    \mathcal{Q}^{\Theta}_{CB\to AD} = \sum_x\mathcal{E}^x_{C\to A}\otimes\mathcal{D}^x_{B\to D}.
\end{equation}

\begin{proposition}
\label{prop:k-ext-contains-1-LOCC}
    All one-way LOCC superchannels are $k$-extendible for all $ k\ge 2$.
\end{proposition}

\begin{proof}
    Once again we exploit the fact that $x$ is classical data and can be copied as many times as needed. Hence, we can construct a superchannel $\Upsilon_{(A\to B_{[k]})\to (C\to D_{[k]})}$ that acts on a quantum channel $\mathcal{P}_{A\to B_{[k]}}$ as
    \begin{multline}
        \Upsilon(\mathcal{P}_{A\to B_{[k]}}) \coloneqq \\  \sum_x\left(\mathcal{D}^x_{B_1\to D_1}\otimes\cdots \otimes\mathcal{D}^x_{B_k\to D_k}\right)\circ(\mathcal{P}_{A\to B_{[k]}})\circ\mathcal{E}^x_{C\to A}.
    \end{multline}
    It is straightforward to verify that such an extension meets the necessary conditions for $\Theta_{(A\to B)\to (C\to D)}$ to be a $k$-extendible superchannel.
\end{proof}

\medskip
We note here that the case of $k=2$ in Proposition~\ref{prop:k-ext-contains-1-LOCC} was already considered in~\cite{Berta_2021,Holdsworth_2023}. 

\section{Unextendible entanglement of quantum channels}\label{sec:unext_ent}

With the framework of $k$-extendibility in hand, we can now define entanglement measures for quantum channels. We will further restrict our development to the case when $k = 2$.
We begin by discussing some mathematical background prior to defining the entanglement measures that arise from unextendibility of quantum channels.

\subsection{Generalized divergence of quantum states}

Let $\mathbb{R}$ denote the field of real numbers. A generalized divergence~\cite{5707067} is a functional $\mathbf{D}: \mathcal{S}(A)\times \mathcal{S}(A) \to \mathbb{R}\cup \{+\infty\}$, such that, for arbitrary states $\rho_A,\sigma_A\in \mathcal{S}(A)$ and an arbitrary channel $\mathcal{N}_{A\to B}$, the data-processing inequality holds
\begin{equation}
    \mathbf{D}\!\left(\rho_A\Vert\sigma_A\right) \ge \mathbf{D}\!\left(\mathcal{N}_{A\to B}(\rho_A)\Vert\mathcal{N}_{A\to B}(\sigma_A)\right).
\end{equation}

Some examples of divergences that commonly appear in quantum information theory are the quantum relative entropy~\cite{10.2996/kmj/1138844604}, Petz-R\'enyi relative entropies~\cite{PETZ198657}, sandwiched R\'enyi relative entropies~\cite{M_ller_Lennert_2013, Wilde_2014}, and geometric R\'enyi relative entropies~\cite{Mat13, Fang_2021}.

We are particularly interested in the geometric R\'enyi relative entropies, which are defined for $\alpha \in (0,1)\cup(1,\infty)$  and all states $\omega$ and $\tau$ as
\begin{align}
    \widehat{D}_{\alpha}(\omega\Vert\tau) &\coloneqq \frac{1}{\alpha -1 }\log_2 \widehat{Q}_{\alpha}(\omega\Vert\tau),\\
    \widehat{Q}_{\alpha}(\omega\Vert\tau) &\coloneqq \lim_{\varepsilon \to 0^+}\operatorname{Tr}\!\left[\tau_{\varepsilon}\!\left(\tau_{\varepsilon}^{-\frac{1}{2}}\omega\tau_{\varepsilon}^{-\frac{1}{2}}\right)^{\alpha}\right],\label{eq:geo_rel_quasi_ent}
\end{align}
where $\tau_{\varepsilon}\coloneqq \tau + \varepsilon I$. Note that $\widehat{D}_{1}(\omega\Vert\tau)$ is defined as $\lim_{\alpha \to 1}\widehat{D}_{\alpha}(\omega\Vert\tau)$ so that $\widehat{D}_{\alpha}\!\left(\omega\Vert\tau\right)$ is defined for all $\alpha \in (0,\infty)$. 

\begin{lemma}[Proposition~74 of \cite{Katariya2021}]
    The geometric R\'enyi relative entropy is strongly faithful; i.e., for all quantum states $\omega$ and $\tau$ and $\forall \alpha \in (0,1) \cup (1,\infty)$, $\widehat{D}_{\alpha}(\omega\Vert\tau)\ge 0$, and $\widehat{D}_{\alpha}(\omega\Vert\tau) = 0$ if and only if $\omega = \tau$. 
\end{lemma}

\begin{lemma}[Proposition~72 of \cite{Katariya2021}]
\label{lem:geo-mono}
    The geometric R\'enyi relative entropy is monotonic in $\alpha$ for all $\alpha > 0 $; i.e., 
    \begin{equation}
        \alpha \geq \beta > 0 \quad \Rightarrow  \quad \widehat{D}_{\alpha}(\omega\Vert\tau) \ge \widehat{D}_{\beta}(\omega\Vert\tau).
    \end{equation}
\end{lemma}

The geometric R\'enyi relative quasi-entropy $\widehat{Q}_{\alpha}\!\left(\omega\Vert\tau\right)$ takes the following form for $\alpha \in (0,1)$:
\begin{equation}
   \widehat{Q}_{\alpha}\!\left(\omega\Vert\tau\right) = \operatorname{Tr}\!\left[\tau\left(\tau^{-\frac{1}{2}}\tilde{\omega}\tau^{-\frac{1}{2}}\right)^{\alpha}\right],
\end{equation}
where
\begin{align}
    \tilde{\omega} &\coloneqq \omega_{0,0} - \omega_{0,1}\omega_{1,1}^{-1}\omega_{0,1}^{\dagger},\label{eq:projected_geo_ent_state}\\
    \omega_{0,0} &\coloneqq \Pi_{\tau}\omega\Pi_{\tau}, \\ \omega_{0,1} & \coloneqq \Pi_{\tau}\omega\Pi^{\perp}_{\tau}, \\ \omega_{1,1}  & \coloneqq \Pi^{\perp}_{\tau}\omega\Pi^{\perp}_{\tau}, 
\end{align}
$\Pi_{\tau}$ is the projection onto the support of $\tau$, and $\Pi_{\tau}^{\perp}$ is the projection onto the kernel of $\tau$. All inverses are taken on the support of the respective operators. Note that when \mbox{$\operatorname{supp}\!\left(\omega\right)\subseteq\operatorname{supp}\!\left(\tau\right)$}, the geometric R\'enyi relative quasi-entropy converges to the following quantity:
\begin{equation}\label{eq:geo_rel_quasi_ent_converge}
    \widehat{Q}_{\alpha}\!\left(\omega\Vert\tau\right) = \operatorname{Tr}\!\left[\tau\left(\tau^{-\frac{1}{2}}\omega\tau^{-\frac{1}{2}}\right)^{\alpha}\right].
\end{equation}

As recalled in Lemma~\ref{lem:geo-mono} above, the $\alpha$-geometric R\'enyi relative entropy increases monotonically in $\alpha$, and the smallest quantity in this family of entropies is achieved when $\alpha \to 0$. As such, we find that 
\begin{equation}
    \widehat{D}_{0}(\omega\|\tau) = \lim_{\alpha\to 0}\widehat{D}_{\alpha}\!\left(\omega\Vert\tau\right) =-\log_2 \operatorname{Tr}[\tau \Pi_{\zeta }],
    \label{eq:min-geometric-div}
\end{equation}
where $\zeta \equiv \tau^{-\frac{1}{2}}\tilde{\omega}\tau^{-\frac{1}{2}}$, the operator $\tilde{\omega}$ was defined in~\eqref{eq:projected_geo_ent_state}, and $\Pi_{\zeta }$ is the projection onto the support of $\zeta$.

When $\alpha \ge 1$, the geometric R\'enyi relative quasi-entropy takes the form in~\eqref{eq:geo_rel_quasi_ent_converge} if $\operatorname{supp}\!\left(\omega\right)\subseteq\operatorname{supp}\!\left(\tau\right)$, and it evaluates to $+\infty$ if $\operatorname{supp}\!\left(\omega\right)\not\subseteq \operatorname{supp}\!\left(\tau\right)$.

The geometric R\'enyi relative entropy obeys the data-processing inequality for $\alpha\in(0,2]$; that is, for all states $\rho$ and $\sigma$, every channel $\mathcal{N}$, and all $\alpha\in(0,2]$, the following inequality holds:
\begin{equation}
    \widehat{D}_{\alpha}(\rho\Vert\sigma) \geq \widehat{D}_{\alpha}(\mathcal{N}(\rho)\Vert\mathcal{N}(\sigma)).
\end{equation}
As $\alpha\to 1$, the geometric R\'enyi relative entropy converges to the Belavkin--Staszewski relative entropy~\cite{Belavkin1982}:
\begin{align}
     \widehat{D}(\omega\Vert\tau) & \equiv 
    \widehat{D}_1(\omega\Vert\tau) \\
    & \coloneqq \lim_{\alpha \to 1}\widehat{D}_{\alpha}(\omega\Vert\tau) \\
    & = \operatorname{Tr}\!\left[\omega\log_2\!\left(\omega^{\frac{1}{2}}\tau^{-1}\omega^{\frac{1}{2}}\right)\right],
\end{align}
as shown in~\cite[Proposition~79]{Katariya2021}.

\subsection{Generalized divergence of quantum channels}

The notion of generalized divergence can be extended to quantum channels, and this extension provides a mathematical framework for comparing channels. The generalized divergence between two channels is defined in terms of the generalized divergence between states, as follows~\cite{Cooney_2016,PhysRevA.97.012332}:
\begin{equation}
\label{eq:gen_div_channel_def}
    \mathbf{D}(\mathcal{N}_{A\to B}\Vert\mathcal{M}_{A\to B}) \coloneqq \sup_{\rho_{RA}} \mathbf{D}(\mathcal{N}_{A\to B}(\rho_{RA})\Vert\mathcal{M}_{A\to B}(\rho_{RA})),
\end{equation}
where $A$ and $B$ are systems of arbitrary size and $\rho_{RA}$ is a quantum state with no constraint on the reference system $R$ in general.
It suffices to restrict the optimization in~\eqref{eq:gen_div_channel_def} to pure states and the dimension of system $R$ to be equal to the dimension of~$A$. This follows from purification, the data-processing inequality, and the fact that all purifications of a state are related by an isometry acting on the purifying system (see~\cite[Proposition~4.79]{khatri2020principles}).

A key property of an arbitrary generalized channel divergence is that it contracts under the action of a superchannel~\cite[Eq.~(92)]{Gour_2019}:

\begin{theorem}[\cite{Gour_2019}]\label{theo:gen_div_channel_data_proc}
    For two arbitrary quantum channels $\mathcal{N}_{A\to B}$ and $\mathcal{M}_{A\to B}$, and an arbitrary superchannel $\Theta_{(A\to B)\to (C\to D)}$, the generalized channel divergence obeys the inequality
    \begin{equation}
        \mathbf{D}\!\left(\mathcal{N}\Vert\mathcal{M}\right) \ge \mathbf{D}\!\left(\Theta(\mathcal{N})\Vert\Theta(\mathcal{M})\right),
    \end{equation}
    where $A,B,C,D$ are systems of arbitrary size.
\end{theorem}

\subsubsection{Geometric R\'enyi relative entropy of channels }

The $\alpha$-geometric R\'enyi relative entropies form a family of generalized divergences between channels. We present a brief review of these quantities and their properties that are relevant for this work. 

The $\alpha$-geometric R\'enyi relative entropy of channels is defined for two arbitrary quantum channels and $\alpha \in (0,1)\cup(1,\infty)$ as
\begin{multline}
    \widehat{D}_{\alpha}(\mathcal{N}_{A\to B}\Vert\mathcal{M}_{A\to B}) \\ \coloneqq \sup_{\rho_{RA}} \widehat{D}_{\alpha}(\mathcal{N}_{A\to B}(\rho_{RA})\Vert\mathcal{M}_{A\to B}(\rho_{RA})) .
\end{multline}
For $\alpha \in (0,1)\cup(1,2]$, the data-processing inequality holds, and so the statement just after~\eqref{eq:gen_div_channel_def} applies, so that it suffices to perform the optimization over pure bipartite states with $R$ isomorphic to $A$.
This quantity is finite for $\alpha > 1$ if and only if $\operatorname{supp}(\Gamma^{\mathcal{N}}_{AB})\subseteq\operatorname{supp}(\Gamma^{\mathcal{M}}_{AB})$~\cite{Fang_2021}. The systems $A$ and $B$ can hold an arbitrary number of subsystems, a property that we shall use in Section~\ref{sec:unext_ent_bip} for bipartite quantum channels.

Note that the $\alpha$-geometric R\'enyi relative entropy of channels converges to the Belavkin--Staszewski relative entropy when $\alpha\to 1$~\cite{Fang_2021,khatri2020principles,DKQSWW23}. Hence, we can remove the discontinuity in $\alpha$ by defining
\begin{equation}\label{eq:geo_ent_ch_conv_Bel_Stas}
    \widehat{D}(\mathcal{N}\Vert\mathcal{M}) \equiv \widehat{D}_1(\mathcal{N}\Vert\mathcal{M}) \coloneqq \lim_{\alpha\to 1}\widehat{D}_{\alpha}\!\left(\mathcal{N}\Vert\mathcal{M}\right).
\end{equation}

\begin{lemma}[\cite{Fang_2021,Katariya2021}]\label{lemma:geo_div}
    The geometric R\'enyi relative entropy for channels satisfies the following properties:
    \begin{enumerate}
        \item It is monotonic in $\alpha$ for $\alpha >0$:
        \begin{equation}\label{eq:geo_div_alpha_monotonic}
            \alpha \geq \beta > 0 \quad  \Rightarrow \quad \widehat{D}_{\alpha}\!\left(\mathcal{N}\Vert\mathcal{M}\right) \ge \widehat{D}_{\beta}\!\left(\mathcal{N}\Vert\mathcal{M}\right)  .
        \end{equation}
        \item It is additive under tensor products of channels for $\alpha \in (0,2]$:
        \begin{equation}\label{eq:geo_div_additive}
            \widehat{D}_{\alpha}(\mathcal{N}^1\otimes\mathcal{N}^2\Vert\mathcal{M}^1\otimes\mathcal{M}^2) =  \widehat{D}_{\alpha}(\mathcal{N}^1\Vert\mathcal{M}^1) +  \widehat{D}_{\alpha}(\mathcal{N}^2\Vert\mathcal{M}^2).    
        \end{equation}
        \item For an arbitrary quantum channel $\mathcal{N}_{A\to B}$, a set $\mathcal{V}$ of completely positive maps   described by semidefinite constraints, and $\alpha = 1+2^{-\ell}$ with $\ell \in \mathbb{N}$, the optimization $\min_{\mathcal{M}\in \mathcal{V}}\widehat{D}_{\alpha}\!\left(\mathcal{N}\Vert\mathcal{M}\right)$ can be computed by a semidefinite program. See Section~\ref{sec:numerical_calc} for the full form of the SDP along with numerical calculations.
    \end{enumerate}
\end{lemma} 

\begin{proof}
    The monotonicity of the $\alpha$-geometric R\'enyi relative entropy of channels for all $ \alpha \in (0,2]$ was shown in~\cite[Appendix~H]{Katariya2021}, and the semidefinite program to minimize the quantity $\widehat{D}_{\alpha}\!\left(\mathcal{N}\Vert\mathcal{M}\right)$ over a set $\mathcal{V}$ of completely positive maps described by semidefinite constraints, was given in~\cite[Lemma 9]{Fang_2021} for all $ \alpha \in (1,2]$.
    
    Additivity of $\alpha$-geometric R\'enyi relative entropy for all $ \alpha \in (0,2]$ is a corollary of~\cite[Proposition 47]{Katariya2021}. For completeness, we present a brief proof of additivity here.
    Consider two arbitrary pure states $\psi_{R_1A_1}$ and $\phi_{R_2A_2}$. The following inequality holds for the $\alpha$-geometric R\'enyi relative entropy between two tensor-product channels:
    \begin{align}
        &\widehat{D}_{\alpha}\!\left(\mathcal{N}^1_{A_1\to B_1}\otimes\mathcal{N}^2_{A_2\to B_2}\middle \Vert\mathcal{M}^1_{A_1\to B_1}\otimes\mathcal{M}^2_{A_2\to B_2}\right)\notag\\
        &\ge \widehat{D}_{\alpha}\!\left(\left(\mathcal{N}^1\otimes\mathcal{N}^2\right)\left(\psi\otimes\phi\right)\middle\Vert\left(\mathcal{M}^1\otimes\mathcal{M}^2\right)\left(\psi\otimes\phi\right)\right)\\
        &= \widehat{D}_{\alpha}\!\left(\mathcal{N}^1\left(\psi\right)\middle\Vert\mathcal{M}^1\left({\psi}\right)\right) + \widehat{D}_{\alpha}\!\left(\mathcal{N}^2\left(\phi\right)\middle\Vert\mathcal{M}^2\left({\phi}\right)\right).
    \end{align}
    Since the above inequality holds for all pure states $\psi_{R_1A_1}$ and $\phi_{R_2A_2}$, we can take a supremum over both and conclude that
    \begin{multline}\label{eq:geo_ent_superadditive}
        \widehat{D}_{\alpha}\!\left(\mathcal{N}^1_{A_1\to B_1}\otimes\mathcal{N}^2_{A_2\to B_2}\middle\Vert\mathcal{M}^1_{A_1\to B_1}\otimes\mathcal{M}^2_{A_2\to B_2}\right)\\
        \ge \widehat{D}_{\alpha}\!\left(\mathcal{N}^1_{A_1\to B_1}\middle\Vert\mathcal{M}^1_{A_1\to B_1}\right) + \widehat{D}_{\alpha}\!\left(\mathcal{N}^2_{A_2\to B_2}\middle\Vert\mathcal{M}^2_{A_1\to B_1}\right).
    \end{multline}
    Now consider the following inequality for channel composition given in Proposition~47 of~\cite{Katariya2021},
    \begin{equation}
      \widehat{D}_{\alpha}(\mathcal{L}_1 \circ \mathcal{L}_2 \| \mathcal{K}_1 \circ \mathcal{K}_2) \leq \\
      \widehat{D}_{\alpha}(\mathcal{L}_1  \| \mathcal{K}_1)
      +\widehat{D}_{\alpha}( \mathcal{L}_2 \|  \mathcal{K}_2),
    \end{equation}
    for channels $\mathcal{L}_1$, $\mathcal{L}_2$, $\mathcal{K}_1$, and $\mathcal{K}_2$, and $\alpha \in (0,1)\cup(1,2]$. Setting
    \begin{align}
    \mathcal{L}_1 & = \mathcal{N}^1_{A_1 \to B_1} \otimes \operatorname{id}_{B_2} , \\
    \mathcal{L}_2 & =  \operatorname{id}_{A_1} \otimes \mathcal{N}^2_{A_2 \to B_2}  , \\
    \mathcal{K}_1 & = \mathcal{M}^1_{A_1 \to B_1} \otimes \operatorname{id}_{B_2} , \\
    \mathcal{K}_2 & =  \operatorname{id}_{A_1} \otimes \mathcal{M}^2_{A_2 \to B_2}  ,
    \end{align}
    we find that
    \begin{align}
        & \widehat{D}_{\alpha}\!\left(\mathcal{N}^1_{A_1\to B_1}\otimes\mathcal{N}^2_{A_2\to B_2}\middle\Vert\mathcal{M}^1_{A_1\to B_1}\otimes\mathcal{M}^2_{A_2\to B_2}\right)\notag \\
        & \leq 
        \widehat{D}_{\alpha}\!\left(\mathcal{N}^1_{A_1\to B_1}\otimes \operatorname{id}_{B_2}\middle\Vert\mathcal{M}^1_{A_1\to B_1}\otimes \operatorname{id}_{B_2}\right) + \notag \\
        & \qquad \widehat{D}_{\alpha}\!\left(\operatorname{id}_{A_1} \otimes\mathcal{N}^2_{A_2\to B_2}\middle\Vert\operatorname{id}_{A_1} \otimes\mathcal{M}^2_{A_1\to B_1}\right) \\
        & = 
        \widehat{D}_{\alpha}\!\left(\mathcal{N}^1_{A_1\to B_1}\middle\Vert\mathcal{M}^1_{A_1\to B_1}\right) + \widehat{D}_{\alpha}\!\left(\mathcal{N}^2_{A_2\to B_2}\middle\Vert\mathcal{M}^2_{A_1\to B_1}\right),
    \end{align}
    where the last equality follows from the stability property of the channel divergence (i.e., it is not changed by tensoring the original channel with an arbitrary identity channel).

    Finally, we can take the limit $\alpha\to 1$ on both sides and use the definition of $\widehat{D}_1\!\left(\cdot\Vert\cdot\right)$ in~\eqref{eq:geo_ent_ch_conv_Bel_Stas} to conclude that the $\alpha$-geometric R\'enyi relative entropy of channels is additive under tensor products for all $\alpha \in (0,2]$.
\end{proof}

\subsection{Generalized unextendible entanglement of quantum states}

The generalized unextendible entanglement of a bipartite  state has been defined in~\cite{WWW19}. We include a short discussion on the topic for necessary development.

\begin{definition}[\cite{WWW19}]
\label{def:unext-ent}
    The generalized unextendible entanglement of a bipartite state $\rho_{AB}$, induced by a generalized divergence $\mathbf{D}$ between states, is defined as
        \begin{multline}
        \label{eq:gen_unext_ent_states}
        \mathbf{E}^u(\rho_{AB}) \coloneqq  \inf_{\rho_{AB_1B_2}\in \mathcal{S}\left(AB_1B_2\right)} \frac{1}{2}\Big\{\mathbf{D}\!\left(\rho_{AB}\Vert\operatorname{Tr}_{B_1}\!\left[\rho_{AB_1B_2}\right]\right)\\ \colon \operatorname{Tr}_{B_2}\!\left[\rho_{AB_1B_2}\right] = \rho_{AB}\Big\},
    \end{multline}    
    where the optimization is over every state $\rho_{AB_1B_2}$ that is an extension of the state $\rho_{AB}$. We also adopt the following alternative notations sometimes because they can be helpful to make the bipartition $A|B$ clear:
    \begin{equation}
        \mathbf{E}^u(A;B)_{\rho} \equiv \mathbf{E}^u(\rho_{A:B}) \equiv \mathbf{E}^u(\rho_{AB}).
    \end{equation}
\end{definition}

Let us define the following set of extensions of a bipartite state~$\rho_{AB}$:
\begin{equation}\label{eq:extensions_state}
    \operatorname{Ext}\!\left(\rho_{AB}\right) \coloneqq \left\{\rho_{AB_1B_2}: \operatorname{Tr}_{B_2}\!\left[\rho_{AB_1B_2}\right] = \rho_{AB}\right\},
\end{equation}
where $B_1$ and $B_2$ are isomorphic to $B$. This allows us to write the generalized unextendible entanglement of  $\rho_{AB}$, induced by the generalized divergence $\mathbf{D}$, as
\begin{equation}
    \mathbf{E}^u(\rho_{AB}) = \inf_{\rho_{AB_1B_2}\in \operatorname{Ext}\left(\rho_{AB}\right)}\frac{1}{2}\mathbf{D}\!\left(\rho_{AB}\Vert\operatorname{Tr}_{B_2}\left[\rho_{AB_1B_2}\right]\right).
\end{equation}

The generalized unextendible entanglement provides a framework for quantifying the unextendibility of a bipartite state $\rho_{AB}$ with respect to the system $B$. A different measure for unextendibility was considered in~\cite{Kaur_2019, Kaur_2021} where the divergence was measured from the fixed set of two-extendible states. However, in Definition~\ref{def:unext-ent}, the divergence is measured by means of a set of states that depend on the input state itself. Although both measures are equal to the minimal possible value of $\mathbf{D}$ when $\rho_{AB}$ is two-extendible, they are not equal in general.

Let us look specifically at the unextendible entanglement induced by the $\alpha$-geometric R\'enyi relative entropy, $\widehat{E}^u_{\alpha}\!\left(\rho_{AB}\right)$, for $ \alpha \in (0,2]$. This quantity is called the $\alpha$-geometric unextendible entanglement in~\cite{WWW19}. 

Note that the underlying divergence of $\alpha$-geometric unextendible entanglement had an intrinsic discontinuity at $\alpha = 1$, which was removed by defining the quantity $\widehat{D}\left(\cdot\Vert\cdot\right)$ as
\begin{equation}
    \widehat{D}\!\left(\cdot\Vert\cdot\right) \equiv \widehat{D}_1\!\left(\cdot\Vert\cdot\right) \coloneqq \lim_{\alpha\to 1} \widehat{D}_{\alpha}\!\left(\cdot\Vert\cdot\right).
\end{equation}
By definition, $\widehat{E}^u_{\alpha}$ is defined for $\alpha = 1$ as the unextendible entanglement induced by $\widehat{D}_1\!\left(\cdot\Vert\cdot\right)$; however, it is necessary to check if the function $\widehat{E}^u_{\alpha}$ is continuous at $\alpha = 1$. We denote the unextendible entanglement induced by $\widehat{D}_1\!\left(\cdot\Vert\cdot\right)$ as $\widehat{E}^u$; that is,
\begin{equation}\label{eq:Bel_Stas_induce_unext_ent_st}
    \widehat{E}^u\!\left(\rho_{AB}\right) \coloneqq \frac{1}{2}\inf_{\sigma\in \operatorname{Ext}\left(\rho\right)}\lim_{\alpha\to 1}\widehat{D}_{\alpha}\!\left(\rho\Vert\operatorname{Tr}_{B_1}\!\left[\sigma_{AB_1B_2}\right]\right).
\end{equation}

\begin{proposition}\label{prop:geo_unext_ent_st_conv_Bel_Stas}
    The $\alpha$-geometric unextendible entanglement of a state, in the limit $\alpha\to 1$, converges to the unextendible entanglement of states induced by the Belavkin--Staszewski relative entropy; i.e.,
    \begin{equation}\label{eq:geo_unext_ent_st_conv_Bel_Stas}
        \lim_{\alpha\to 1}\widehat{E}^u_{\alpha}\!\left(\rho_{AB}\right) = \widehat{E}^u\!\left(\rho_{AB}\right).
    \end{equation}
\end{proposition}

\begin{proof}
    See Appendix~\ref{app:geo_unext_ent_st_conv_Bel_Stas}.
\end{proof}

\medskip
The $\alpha$-geometric unextendible entanglement of states increases monotonically in $\alpha$, which is a consequence of the $\alpha$-geometric R\'enyi relative entropy being monotonic in $\alpha$. The smallest quantity in this family of unextendibility measures is achieved in the limit $\alpha\to 0$. We define this quantity as the  min-geometric unextendible entanglement:
\begin{equation}\label{eq:min_geo_unext_ent_st_def}
    \widehat{E}^u_{\min}\!\left(\rho_{AB}\right) \coloneqq \lim_{\alpha\to 0^+}\widehat{E}^u_{\alpha}\!\left(\rho_{AB}\right).
\end{equation}

\begin{proposition}
    The min-geometric unextendible entanglement of a quantum state is the unextendibe entanglement induced by the $\alpha$-geometric R\'enyi relative entropy when $\alpha\to 0$; that is,
    \begin{equation}
        \widehat{E}^u_{\min}\!\left(\rho_{AB}\right) = \frac{1}{2}\inf_{\sigma\in\operatorname{Ext}\left(\rho\right)}\lim_{\alpha\to 0^+}\widehat{D}_{\alpha}\!\left(\rho\middle 
 \Vert\operatorname{Tr}_{B_1\left[\sigma_{AB_1B_2}\right]}\right).
    \end{equation}
\end{proposition}

\begin{proof}
    Due to the monotonicity of the $\alpha$-geometric unextendible entanglement in $\alpha$, we can write
    \begin{align}
        &\widehat{E}^u_{\min}\!\left(\rho_{AB}\right)\notag\\ 
        &= \lim_{\alpha\to 0^+}\widehat{E}^u_{\alpha}\!\left(\rho_{AB}\right)\\
        &= \inf_{\alpha\in (0,1)}\frac{1}{2}\inf_{\sigma\in\operatorname{Ext}\left(\rho\right)}\widehat{D}_{\alpha}\!\left(\rho\Vert\operatorname{Tr}_{B_1}\!\left[\sigma_{AB_1B_2}\right]\right)\\
        &= \frac{1}{2}\inf_{\sigma\in\operatorname{Ext}\left(\rho\right)}\inf_{\alpha\in (0,1)}\widehat{D}_{\alpha}\!\left(\rho\Vert\operatorname{Tr}_{B_1}\!\left[\sigma_{AB_1B_2}\right]\right)\\
        &= \frac{1}{2}\inf_{\sigma\in\operatorname{Ext}\left(\rho\right)}\lim_{\alpha\to 0^+}\widehat{D}_{\alpha}\!\left(\rho\Vert\operatorname{Tr}_{B_1}\!\left[\sigma_{AB_1B_2}\right]\right),
    \end{align}
    where the final equality follows from the monotonicity of the $\alpha$-geometric R\'enyi relative entropy.
\end{proof}

\begin{theorem}[\cite{WWW19}]\label{theo:props_from_WWW19}
    The $\alpha$-geometric unextendible entanglement of states for all $ \alpha \in (0,2]$ satisfies the following properties:
    \begin{enumerate}
        \item It is monotonic under a selective two-extendible bipartite operation $\left\{\mathcal{E}^y_{AB\to A'B'}\right\}_y$:
        \begin{equation}
            \widehat{E}^u_{\alpha}\!\left(\rho_{AB}\right) \ge \sum_{y:p(y) > 0} p(y)\widehat{E}^u_{\alpha}\!\left(\omega^y_{A'B'}\right),
        \end{equation}
        where
        \begin{align}
            p(y) &\coloneqq \operatorname{Tr}\!\left[\mathcal{E}^y_{AB\to A'B'}\!\left(\rho_{AB}\right)\right] , \\
            \omega^y_{A'B'} &\coloneqq \frac{1}{p(y)}\mathcal{E}^y_{AB\to A'B'}\!\left(\rho_{AB}\right).
        \end{align}
        \item It obeys subadditivity for states $\rho_{A_1B_1}$ and $\sigma_{A_2B_2}$:
        \begin{equation}\label{eq:alpha_geo_unext_ent_subadditive}
            \widehat{E}^u_{\alpha}\!\left(\rho_{A_1B_1}\otimes\sigma_{A_2B_2}\right) \le \widehat{E}^u_{\alpha}\!\left(\rho_{A_1B_1}\right) + \widehat{E}^u_{\alpha}\!\left(\sigma_{A_2B_2}\right).
        \end{equation}
        \item Let $\Phi^d_{AB}$ be a maximally entangled state of Schmidt rank~$d$. Then
        \begin{equation}\label{eq:geo_unext_ent_edit_logd}
            \widehat{E}^u_{\alpha}\!\left(\Phi^d_{AB}\right) = \log_2 d.
        \end{equation}
    \end{enumerate}
\end{theorem}

\begin{proposition}[Direct--sum]\label{prop:geo_unext_direct_sum}
    Let $ \rho_{XX'AB}$ denote the following classical--classical--quantum state:
    \begin{equation}
        \rho_{XX'AB} \coloneqq  \sum_{x}p(x)|x\rangle\!\langle x|_{X}\otimes|x\rangle\!\langle x|_{X'}\otimes \rho^x_{AB}.
    \end{equation}
    For $\alpha \in (1,2]$, the $\alpha$-geometric unextendible entanglement  obeys the following inequalities:
    \begin{equation}
        \label{eq:geo_unext_direct_sum}
        \widehat{E}^u_{\alpha}(\rho_{XA:B}) \ge \sum_{x} p(x)\widehat{E}^u_{\alpha}(\rho^x_{AB}).
    \end{equation}
    For $\alpha \in (0,2]$, it obeys the following:
    \begin{equation}
        \widehat{E}^u_{\alpha}(\rho_{XA:B})   =
        \widehat{E}^u_{\alpha}(\rho_{XA:X'B}) \geq \widehat{E}^u_{\alpha}(\rho_{A:X'B}) .
        \label{eq:geo_others_direct_sum_alpha}
    \end{equation}   
    The following equality holds for the Belavkin--Staszewski unextendible entanglement:
    \begin{equation}
    \label{eq:BS-unext-ent-avg-identity}
    \widehat{E}^u(\rho_{XA:B}) = \sum_{x} p(x)\widehat{E}^u(\rho^x_{AB}).
    \end{equation}
\end{proposition}

\begin{proof}
    Let us first prove the inequality in~\eqref{eq:geo_unext_direct_sum}. A general extension of the state $\rho_{XAB}$ is of the form,
    \begin{equation}\label{eq:CQ_state_ext}
        \rho_{XAB_1B_2} = \sum_x p(x)|x\rangle\!\langle x|_X\otimes\rho^x_{AB_1B_2},
    \end{equation}
    where $\rho^x_{AB_1B_2}$ is an arbitrary extension of $\rho^x_{AB}$. Due to the direct-sum property of the $\alpha$-geometric R\'enyi relative quasi-entropy defined in~\eqref{eq:geo_rel_quasi_ent},
    \begin{align}
        &\log_2\widehat{Q}_{\alpha}\!\left(\rho_{XAB}\Vert\operatorname{Tr}_{B_1}[\rho_{XAB_1B_2}]\right) \notag \\
        &= \log_2\sum_x p(x)\widehat{Q}_{\alpha}\!\left(\rho^x_{AB}\Vert\operatorname{Tr}_{B_1}[\rho^x_{AB_1B_2}]\right)\\
        &\ge \sum_x p(x)\log_2\widehat{Q}_{\alpha}\!\left(\rho^x_{AB}\Vert \operatorname{Tr}_{B_1}\!\left[\rho^x_{AB_1B_2}\right]\right), \label{eq:geo_unext_direct_sum_im}
    \end{align}
    where the inequality follows from  concavity of the logarithm. Furthermore, for $\alpha > 1$, we can divide both sides by $\alpha-1$ while preserving the inequality sign, implying that
    \begin{align}
        &\widehat{D}_{\alpha}\!\left(\rho_{XAB}\Vert\operatorname{Tr}_{B_1}\!\left[\rho_{XAB_1B_2}\right]\right)\notag \\
        &=\frac{1}{\alpha-1} \log_2\widehat{Q}_{\alpha}\!\left(\rho_{XAB}\Vert\operatorname{Tr}_{B_1}[\rho_{XAB_1B_2}]\right)\\
        &\ge \frac{1}{\alpha-1}\sum_x p(x)\log_2\widehat{Q}_{\alpha}\!\left(\rho^x_{AB}\Vert \operatorname{Tr}_{B_1}\!\left[\rho^x_{AB_1B_2}\right]\right)\\
        &= \sum_x p(x)\widehat{D}_{\alpha}\!\left(\rho^x_{AB}\Vert \operatorname{Tr}_{B_1}\!\left[\rho^x_{AB_1B_2}\right]\right)\\
        &\ge 2\sum_x p(x)\widehat{E}^u_{\alpha}(\rho^x_{AB}).
    \end{align}
    Since the above inequality holds for every extension $\rho_{XAB_1B_2}$, we can take the infimum over all such extensions and conclude~\eqref{eq:geo_unext_direct_sum}.

    Let us now prove the statements in~\eqref{eq:geo_others_direct_sum_alpha}.  Observe that Alice can copy the contents of the classical register $X$ to $X'$ and send $X'$ to Bob, and this action is a one-way LOCC channel. By invoking the fact that the $\alpha$-geometric unextendible entanglement does not increase under one-way LOCC for $\alpha \in (0,2]$ (see~\cite[Remark~3]{WWW19}), we conclude that
    \begin{equation}
        \widehat{E}^u_{\alpha}(\rho_{XA:B})   \geq
        \widehat{E}^u_{\alpha}(\rho_{XA:X'B}).
    \end{equation}
    Since performing a partial trace over the register $X'$ in Bob's possession is also a one-way LOCC channel, we conclude the opposite inequality:
     \begin{equation}
       \widehat{E}^u_{\alpha}(\rho_{XA:X'B})    \geq
        \widehat{E}^u_{\alpha}(\rho_{XA:B}),
    \end{equation}
    which, together with the inequality above, implies the equality in~\eqref{eq:geo_others_direct_sum_alpha}. Finally, discarding the register $X$ in Alice's possession is also a one-way LOCC, which implies the inequality in~\eqref{eq:geo_others_direct_sum_alpha}. 

    We finally note that the inequality
    \begin{equation}
    \widehat{E}^u(\rho_{XA:B}) \geq \sum_{x} p(x)\widehat{E}^u(\rho^x_{AB})
    \label{eq:BS-unext-avg-ineq}
    \end{equation}
    holds because~\eqref{eq:geo_unext_direct_sum} holds for all $\alpha > 1$, and thus we can invoke Proposition~\ref{prop:geo_unext_ent_st_conv_Bel_Stas} and take the limit as $\alpha \to 1$. The opposite inequality is a consequence of the following reasoning. Let $\rho_{AB_1B_2}^x$ be an arbitrary extension of $\rho_{AB}^x$ (i.e., $\operatorname{Tr}_{B_2}[\rho_{AB_1B_2}^x] = \rho_{AB}^x$). Consider that
    \begin{align}
        \sum_x p(x) \widehat{D}(\rho_{AB_1}^x \Vert \rho_{AB_2}^x) & = \widehat{D}(\rho_{XAB_1} \Vert \rho_{XAB_2}) \\
        & \geq 2 \widehat{E}^u(\rho_{XA:B}),
    \end{align}
    where we invoked the direct-sum property of the Belavkin--Staszewski relative entropy (see~\cite[Eq.~(4.7.62)]{khatri2020principles}).
    Since the inequality holds for every extension of $\rho_{AB}^x$, we conclude that
    \begin{equation}
    \sum_{x} p(x)\widehat{E}^u(\rho^x_{AB}) \geq 
    \widehat{E}^u(\rho_{XA:B}) .
    \end{equation}
    Combining the above inequality with the inequality in~\eqref{eq:BS-unext-avg-ineq}, we conclude~\eqref{eq:BS-unext-ent-avg-identity}.
\end{proof}

\subsection{Generalized unextendible entanglement of point-to-point quantum channels}

We are now in a position to define the unextendible entanglement of a point-to-point quantum channel. For a channel divergence $\mathbf{D}$, the unextendible entanglement of a quantum channel $\mathcal{N}_{A\to B}$ is defined as
\begin{multline}
    \mathbf{E}^u\!\left(\mathcal{N}_{A\to B}\right) \coloneqq \inf_{\mathcal{P}_{A\to B_1B_2}}\frac{1}{2}\big\{\mathbf{D}\!\left(\mathcal{N}_{A\to B}\middle \Vert\operatorname{Tr}_{B_1}\circ\mathcal{P}_{A\to B_1B_2}\right)\\: \operatorname{Tr}_{B_2}\circ\mathcal{P}_{A\to B_1B_2} = \mathcal{N}_{A\to B}\big\},
\end{multline}
where the infimum is taken over every extension channel $\mathcal{P}_{A\to B_1B_2}$.
Let us define the following set of extensions of the quantum channel $\mathcal{N}_{A\to B}$:
\begin{equation}\label{eq:ext_set_p2p_ch}
    \operatorname{Ext}\!\left(\mathcal{N}_{A\to B}\right) \coloneqq \left\{\mathcal{P}_{A\to B_1B_2} \in \text{CP} : \operatorname{Tr}_{B_2}\circ\mathcal{P}_{A\to B_1B_2} = \mathcal{N}_{A\to B}\right\},
\end{equation}
where $B_1$ and $B_2$ are isomorphic to $B$ and CP denotes the set of completely positive maps. The unextendible entanglement of a point-to-point quantum channel $\mathcal{N}_{A\to B}$, induced by a channel divergence $\mathbf{D}$, can thus be written as
\begin{multline}
    \mathbf{E}^u\!\left(\mathcal{N}_{A\to B}\right) \\= \inf_{\mathcal{P}_{A\to B_1B_2}\in \operatorname{Ext}\left(\mathcal{N}\right)}\frac{1}{2}\mathbf{D}\!\left(\mathcal{N}_{A\to B}\middle\Vert\operatorname{Tr}_{B_2}\circ\mathcal{P}_{A\to B_1B_2}\right).
\end{multline}
This definition is motivated from the definition of unextendible entanglement of bipartite states in~\eqref{eq:gen_unext_ent_states}. The unextendible entanglement of a quantum channel $\mathcal{N}_{A\to B}$ between Alice and Bob quantifies the distinguishability of the two marginals of a quantum broadcast channel $\mathcal{P}_{A\to B_1B_2}$ such that one of the marginals is the channel of interest.

The no-broadcasting theorem~\cite{BCFJS96}, a generalization of the no-cloning theorem to mixed states, implies that there cannot exist a quantum channel $\mathcal{P}_{A\to B_1B_2}$ that can perfectly broadcast an arbitrary quantum state, in the sense that, for such a purported perfect broadcast channel, the marginal states on the output systems $B_1$ and $B_2$ are the same as the input quantum state on system $A$. As such, there cannot exist a quantum channel $\mathcal{P}_{A\to B_1B_2}$ such that each of its marginals is the identity channel. However, a quantum broadcast channel can have identical marginals if the marginals are noisy channels, for example, a trivial channel that replaces the input with a fixed quantum state. The unextendible entanglement of a quantum channel thus can be understood as a measure of entanglement of the quantum channel arising from the limitations imposed by the non-broadcastability of quantum information.

The unextendible entanglement induced by a divergence~$\mathbf{D}$ achieves its minimum for two-extendible channels. If the underlying divergence is strongly faithful, then the unextendible entanglement is equal to zero if and only if the channel is two-extendible. 

\begin{theorem}[Monotonicity]\label{theo:two_ext_monotonic_p2p}
    The generalized unextendible entanglement of a channel does not increase under the action of a two-extendible superchannel. That is, for an arbitrary quantum channel $\mathcal{N}_{A\to B}$ and a two-extendible superchannel~$\Theta_{(A\to B)\to (C\to D)}$,    \begin{equation}\label{eq:gen_unext_ent_monotonic_superch}
         \mathbf{E}^u\!\left(\mathcal{N}_{A\to B}\right) \ge  \mathbf{E}^u\!\left(\Theta\left(\mathcal{N}_{A\to B}\right)\right).
    \end{equation}
\end{theorem}

\begin{proof}
      To begin with, consider a two-extendible superchannel $\Theta_{(A\to B)\to (C\to D)}$ with the two-extension $\Upsilon_{(A\to B_1B_2)\to (C\to D_1D_2)}$.
      
      Let $\mathcal{P}_{A\to B_1B_2}$ be an extension of the channel $\mathcal{N}_{A\to B}$, implying $\mathcal{N}_{A \to B} = \operatorname{Tr}_{B_2}\circ\mathcal{P}_{A\to B_1B_2}$. The generalized divergence between the two marginals of the channel $\mathcal{P}_{A\to B_1B_2}$ satisfies
     \begin{align}
         &\mathbf{D}\!\left(\mathcal{N}_{A \to B}\Vert\operatorname{Tr}_{B_1}\circ\mathcal{P}_{A \to B_1B_2}\right)\notag \\
          &\ge \mathbf{D}\!\left(\Theta\left(\mathcal{N}_{A \to B}\right)\Vert\Theta\left(\operatorname{Tr}_{B_1}\circ\mathcal{P}_{A \to B_1B_2}\right)\right)\\
          &= \mathbf{D}\!\left(\operatorname{Tr}_{D_2}\circ\Upsilon\left(\mathcal{P}_{A\to B_1B_2}\right)\Vert \operatorname{Tr}_{D_1}\circ\Upsilon\left(\mathcal{P}_{A\to B_1B_2}\right) \right)\\
          &\ge 2\mathbf{E}^u\!\left(\Theta\left(\mathcal{N}_{A\to B}\right)\right).\label{eq:ch_div_ge_unext_ent_superch}
     \end{align}
The first inequality follows from the contraction of a generalized channel divergence under superchannels (Theorem~\ref{theo:gen_div_channel_data_proc}). The equality follows from the nonsignaling property of two-extendible superchannels (Proposition~\ref{prop:two_ext_supch_ch_marginal}), and the last inequality follows from the fact that $\Upsilon\left(\mathcal{P}_{A\to B_1B_2}\right)$ is a valid extension of the channel $\Theta\left(\mathcal{N}_{A\to B}\right)$. 

Since~\eqref{eq:ch_div_ge_unext_ent_superch} holds true for every extension $\mathcal{P}_{A\to B_1B_2}$ of the channel $\mathcal{N}_{A\to B}$, we can take the infimum over all such channels and conclude~\eqref{eq:gen_unext_ent_monotonic_superch}.
\end{proof}

\begin{remark}
    Since all one-way LOCC channels are two-extendible, the generalized unextendible entanglement of a channel does not increase under the action of a one-way LOCC superchannel.
\end{remark}

A superchannel $\Theta$ can also convert a point-to-point quantum channel $\mathcal{N}_{A\to B}$ to a bipartite quantum channel $\mathcal{M}_{C\to C'D}$ that is nonsignaling from $D$ to $C$. We have not yet defined the unextendible entanglement of such channels (see Section~\ref{sec:unext_ent_bip} for unextendible entanglement of general bipartite channels); however, we can still compare the unextendible entanglement of any bipartite state that can be established by a channel of the form $\mathcal{M}_{C\to C'D}$ with the unextendible entanglement of point-to-point channels. 

\begin{theorem}\label{theo:unext_ent_state_le_unext_ent_ch_gen}
    The unextendible entanglement of a quantum state $\sigma_{RC'D}$, with respect to the partition $RC':D$, that can be established between two parties using a point-to-point quantum channel $\mathcal{N}_{A\to B}$ and a two-extendible superchannel $\Theta_{(A\to B)\to (C\to C'D)}$ is no greater than the unextendible entanglement of the quantum channel $\mathcal{N}_{A\to B}$; i.e., 
    \begin{equation}\label{eq:unext_ent_state_le_unext_ent_ch_gen}
        \sup_{\rho_{RC}}\mathbf{E}^u\!\left(\sigma_{RC':D}\right) \le  \mathbf{E}^u\!\left(\mathcal{N}_{A\to B}\right), 
    \end{equation}
    where 
    \begin{equation}
        \sigma_{RC'D} \coloneqq \left(\Theta_{(A\to B)\to (C\to C'D)}\!\left(\mathcal{N}_{A\to B}\right)\right)\left(\rho_{RC}\right),
    \end{equation}
    and $\rho_{RC}$ is a quantum state.
\end{theorem}

\begin{proof}
    Let $\Theta_{(A\to B)\to (C\to C'D)}$ be a two-extendible superchannel, and let $\Upsilon_{(A\to B_1B_2)\to (C\to C'D_1D_2)}$ be a two-extension of it. Consider the following state:
    \begin{equation}
        \sigma_{RC'D} \coloneqq \left(\Theta_{(A\to B)\to (C\to C'D)}\!\left(\mathcal{N}_{A\to B}\right)\right)\left(\rho_{RC}\right),
    \end{equation}
    where $\mathcal{N}_{A\to B}$ is a point-to-point channel. 
    Let $\mathcal{P}_{A\to B_1B_2}$ be an arbitrary extension of the point-to-point channel $\mathcal{N}_{A\to B}$. Proposition~\ref{prop:two_ext_supch_ch_marginal} implies the following equality:
    \begin{align}
        & \operatorname{Tr}_{D_2}\!\left[\left(\Upsilon\left(\mathcal{P}_{A\to B_1B_2}\right)\right)\left(\rho_{RC}\right)\right] \notag \\
        &= \left(\Theta\left(\operatorname{Tr}_{B_2}\circ\mathcal{P}_{A\to B_1B_2}\right)\right)\left(\rho_{RC}\right)\\
        &= \left(\Theta\left(\mathcal{N}_{A\to B}\right)\right)\left(\rho_{RC}\right)\\
        &= \sigma_{RC'D}.
    \end{align}
    This implies that $\Upsilon\left(\mathcal{P}_{A\to B_1B_2}\right)\left(\rho_{RC}\right)$ is an extension of $\sigma_{RC'D}$ with respect to the partition $RC':D$. By definition, the unextendible entanglement of the state $\sigma_{RC'D}$ satisfies the following inequality:
    \begin{align}
        &\mathbf{E}^u\!\left(\sigma_{RC':D}\right) \notag\\
        &\le \inf_{\mathcal{P}\in \operatorname{Ext}\left(\mathcal{N}\right)}\frac{1}{2}\mathbf{D}\!\left(\sigma_{RC'D}\middle \Vert\operatorname{Tr}_{D_1}\!\left[\Upsilon\left(\mathcal{P}\right)\left(\rho_{RC}\right)\right]\right)\\
        &= \inf_{\mathcal{P}\in \operatorname{Ext}\left(\mathcal{N}\right)}\frac{1}{2}\mathbf{D}\!\left(\sigma_{RC'D} \middle \Vert\left[\Theta\left(\operatorname{Tr}_{B_1}\circ\mathcal{P}\right)\left(\rho_{RC}\right)\right]\right).
    \end{align}
    Taking a supremum over every input state $\rho_{RC}$, we arrive at the following relations:
    \begin{align}
        &\sup_{\rho_{RC}}\mathbf{E}^u\!\left(\sigma_{RC':D}\right) \notag\\
        &\le \sup_{\rho_{RC}}\inf_{\mathcal{P}\in \operatorname{Ext}\left(\mathcal{N}\right)}\frac{1}{2}\mathbf{D}\!\left(\sigma \middle \Vert\left(\Theta\left(\operatorname{Tr}_{B_1}\circ\mathcal{P}_{A\to B_1B_2}\right)\right)\left(\rho_{RC}\right)\right)\notag\\
        &= \sup_{\rho_{RC}}\inf_{\mathcal{P}\in \operatorname{Ext}\left(\mathcal{N}\right)} \frac{1}{2}\mathbf{D}\!\left(\left(\Theta\left(\mathcal{N}\right)\right)\left(\rho\right) \middle \Vert\left(\Theta\left(\operatorname{Tr}_{B_1}\circ\mathcal{P}\right)\right)\left(\rho\right)\right) \notag\\
        &\le \inf_{\mathcal{P}\in \operatorname{Ext}\left(\mathcal{N}\right)}\sup_{\rho_{RC}} \frac{1}{2}\mathbf{D}\!\left(\left(\Theta\left(\mathcal{N}\right)\right)\left(\rho\right) \middle \Vert\left(\Theta\left(\operatorname{Tr}_{B_1}\circ\mathcal{P}\right)\right)\left(\rho\right)\right)\notag\\
        &= \inf_{\mathcal{P}\in\operatorname{Ext}\left(\mathcal{N}\right)}\frac{1}{2}\mathbf{D}\!\left(\Theta\left(\mathcal{N}\right)\middle \Vert\Theta\left(\operatorname{Tr}_{B_1}\circ\mathcal{P}\right)\right)\notag\\
        &\le \inf_{\mathcal{P}\in\operatorname{Ext}\left(\mathcal{N}\right)} \frac{1}{2}\mathbf{D}\!\left(\mathcal{N}_{A\to B}\middle \Vert\operatorname{Tr}_{B_1}\circ\mathcal{P}_{A\to B_1B_2}\right)\notag\\
        &=\mathbf{E}^u\!\left(\mathcal{N}_{A\to B}\right),
    \end{align}
    where the second inequality follows from the max-min inequality, the second equality follows from the definition of generalized divergence of channels, the penultimate inequality follows from the data-processing inequality in Theorem~\ref{theo:gen_div_channel_data_proc}, and the final equality follows from the definition of the unextendible entanglement of  $\mathcal{N}_{A\to B}$. 
\end{proof}

\subsubsection{\texorpdfstring{$\alpha$}{Alpha}-geometric unextendible entanglement of quantum channels}

In this section, we consider the unextendible entanglement induced by the  geometric R\'enyi relative entropy:
\begin{multline}\label{eq:geo_unext_ent_ch_def}
    \widehat{E}^u_{\alpha}\!\left(\mathcal{N}_{A\to B}\right) \coloneqq \inf_{\mathcal{P}_{A\to B_1B_2}}\Big\{\widehat{D}_{\alpha}\!\left(\mathcal{N}_{A\to B}\Vert\operatorname{Tr}_{B_1}\circ\mathcal{P}_{A\to B_1B_2}\right)\\: \operatorname{Tr}_{B_2}\circ\mathcal{P}_{A\to B_1B_2} = \mathcal{N}_{A\to B}\Big \} \quad \forall \alpha \in (0,2],
\end{multline}
where the infimum is taken over every channel $\mathcal{P}_{A\to B_1B_2}$.
We shall refer to this quantity as the $\alpha$-geometric unextendible entanglement of channels following the convention used for unextendible entanglement of states. We list some important properties of the $\alpha$-geometric unextendible entanglement of quantum channels.
\begin{enumerate}
    \item \textbf{Subadditivity:} The $\alpha$-geometric unextendible entanglement of a channel is subadditive with respect to tensor products of channels, for all $\alpha \in (0,2]$ (Proposition~\ref{prop:unext_ent_subadditive}).
    \item \textbf{Monotonicity in $\alpha$:} The $\alpha$-geometric unextendible entanglement of a channel increases monotonically with increasing $\alpha$ (Proposition~\ref{prop:geo_unext_ent_monotonic}).
    \item \textbf{Computable via SDP:} For a positive integer $\ell$, and $\alpha = 1+2^{-\ell}$, the $\alpha$-geometric unextendible entanglement of a quantum channel can be computed using a semidefinite program (Proposition~\ref{prop:SDP_for_geo_unext_ent}).
\end{enumerate}

The $\alpha$-geometric unextendible entanglement of quantum channels is subadditive under tensor products of channels, due to the additive property of the underlying divergence (see Lemma~\ref{lemma:geo_div}).

\begin{proposition}[Subadditivity]\label{prop:unext_ent_subadditive}
For every two quantum channels $\mathcal{N}^{1}_{A\to B}$ and $\mathcal{N}^{2}_{A\to B'}$, the $\alpha$-geometric unextendible entanglement is subadditive for all $\alpha \in (0,2]$:
    \begin{equation}\label{eq:geo_unext_ent_dec_tensor}
        \widehat{E}^u_\alpha(\mathcal{N}^{1}_{A\to B}\otimes\mathcal{N}^{2}_{A'\to B'}) \le \widehat{E}^u_\alpha(\mathcal{N}^{1}_{A\to B}) + \widehat{E}^u_\alpha(\mathcal{N}^{2}_{A'\to B'}).
    \end{equation}
\end{proposition}

\begin{proof}
    Let $\mathcal{P}^{1}_{A\to B_1B_2}$ and $\mathcal{P}^{2}_{A'\to B'_1B'_2}$ be arbitrary extensions of the channels $\mathcal{N}^{1}_{A\to B}$ and $\mathcal{N}^{2}_{A'\to B'}$. As such, 
    \begin{align}
        \mathcal{N}^{1}_{A\to B} &= \operatorname{Tr}_{B_2}\circ\mathcal{P}^{1}_{A\to B_1B_2} , \\
        \mathcal{N}^{2}_{A'\to B'} &= \operatorname{Tr}_{B'_2}\circ\mathcal{P}^{2}_{A'\to B'_1B'_2}.
    \end{align}
    The tensor-product channel $\mathcal{P}^{1}_{A\to B_1B_2}\otimes\mathcal{P}^{2}_{A'\to B'_1B'_2}$ is an extension of the channel $\mathcal{N}^{1}_{A\to B}\otimes\mathcal{N}^{2}_{A'\to B'}$ because
    \begin{equation}
        \mathcal{N}^{1}_{A\to B}\otimes\mathcal{N}^{2}_{A'\to B'} = \operatorname{Tr}_{B_2B'_2}\circ\left(\mathcal{P}^{1}_{A\to B_1B_2}\otimes\mathcal{P}^{2}_{A'\to B'_1B'_2}\right).
    \end{equation}
    Note that in this case we are interested in the extendibility of the joint system $BB'$ with respect to the joint system $AA'$. Then consider that
    \begin{align}
        &\widehat{E}^u_{\alpha}\!\left(\mathcal{N}^{1}\otimes\mathcal{N}^{2}\right)\\
        &\le \frac{1}{2}\widehat{D}_{\alpha}\!\left(\mathcal{N}^{1}\otimes\mathcal{N}^{2}\Vert\operatorname{Tr}_{B_1B'_1}\circ\left(\mathcal{P}^{1}\otimes\mathcal{P}^{2}\right)\right)\\
        &= \frac{1}{2}\widehat{D}_{\alpha}\!\left(\mathcal{N}^{1}\Vert\operatorname{Tr}_{B_1}\circ\mathcal{P}^{1}\right) + \frac{1}{2}\widehat{D}_{\alpha}\!\left(\mathcal{N}^{2}\Vert\operatorname{Tr}_{B'_1}\circ\mathcal{P}^{2}\right),\label{eq:geo_unext_ent_tensor_ge_geo_rel_ent_sum}
    \end{align}
    where the last equality follows from the additive property of $\alpha$-geometric R\'enyi relative entropy of channels in~\eqref{eq:geo_div_additive}. Since~\eqref{eq:geo_unext_ent_tensor_ge_geo_rel_ent_sum} holds for arbitrary extensions $\mathcal{P}^{1}_{A\to B_1B_2}$ and $\mathcal{P}^{2}_{A'\to B'_1B'_2}$, we can take the infimum over both and conclude~\eqref{eq:geo_unext_ent_dec_tensor}.
\end{proof}

\begin{proposition}\label{prop:geo_unext_ent_monotonic}
For all $\alpha \in (0,2]$,
    the $\alpha$-geometric unextendible entanglement is monotonic in $\alpha$, i.e.,
    \begin{equation}
        2 \ge \alpha \ge \beta > 0 \quad \Rightarrow  \quad \widehat{E}^u_{\alpha}\!\left(\mathcal{N}_{A\to B}\right) \ge \widehat{E}^u_{\beta}\!\left(\mathcal{N}_{A\to B}\right). 
    \end{equation}
\end{proposition}

\begin{proof}
    This is a direct consequence of the monotonicity of geometric R\'enyi relative entropy of channels in $\alpha$, which itself is a direct consequence of the $\alpha$-monotonicity of the geometric R\'enyi relative entropy for states~\cite[Eq.~(6.16)]{Katariya2021}. Let $\mathcal{P}_{A\to B_1B_2}$ be an arbitrary extension of $\mathcal{N}_{A\to B}$. Then
    \begin{align}
        & \frac{1}{2}\widehat{D}_{\alpha}\!\left(\mathcal{N}_{A\to B}\Vert\operatorname{Tr}_{B_1}\circ\mathcal{P}_{A\to B_1B_2}\right) \notag \\
        &\ge \frac{1}{2}\widehat{D}_{\beta}\!\left(\mathcal{N}_{A\to B}\Vert\operatorname{Tr}_{B_1}\circ\mathcal{P}_{A\to B_1B_2}\right)\\
        &\ge \widehat{E}^u_{\beta}\!\left(\mathcal{N}_{A\to B}\right),
    \end{align}
    where the first inequality follows from the monotonicity of geometric R\'enyi relative entropy in $\alpha$ and the last inequality follows from the definition of $\beta$-geometric R\'enyi unextendible entanglement of channels. Since the inequality holds for every extension, we conclude the desired claim by taking the infimum over all such extensions. 
\end{proof}

\begin{remark}
\label{rem:def-min-unext}
    The smallest quantity in the family of $\alpha$-geometric unextendible entanglement is achieved in the limit $\alpha\to 0^+$, as implied by Proposition \ref{prop:geo_unext_ent_monotonic}. We call this quantity the \mbox{min-geometric} unextendible entanglement and define it as follows:
    \begin{equation}\label{eq:min_unext_ent_ch_def}
        \widehat{E}^u_{\operatorname{min}}\!\left(\mathcal{N}_{A\to B}\right) \coloneqq \lim_{\alpha\to 0^+}\widehat{E}^u_{\alpha}\!\left(\mathcal{N}_{A\to B}\right).
    \end{equation}
\end{remark}

\begin{proposition}
\label{prop:limit-zero-min-geo}
    The min-geometric unextendible entanglement of a channel is the unextendible entanglement induced by the $\alpha$-geometric R\'enyi relative entropy as $\alpha \to 0$:
    \begin{equation}
         \widehat{E}^u_{\operatorname{min}}\!\left(\mathcal{N}_{A\to B}\right)
         = \inf_{\mathcal{P}\in \operatorname{Ext}\left(\mathcal{N}\right)}\sup_{\psi_{RA}} \widehat{D}_{0}\!\left(\mathcal{N}\!\left(\psi\right)\Vert\operatorname{Tr}_{B_1}\!\left[\mathcal{P}\!\left(\psi\right)\right]\right).
    \end{equation}
\end{proposition}

\begin{proof}
    Invoking the definition of min-geometric unextendible entanglement,
    \begin{align}
        &\widehat{E}^u_{\operatorname{min}}\!\left(\mathcal{N}_{A\to B}\right)\notag\\
        &= \lim_{\alpha\to 0^+}\widehat{E}^u_{\alpha}\!\left(\mathcal{N}_{A\to B}\right)\\
        &= \inf_{\alpha \in (0,2]} \inf_{\mathcal{P}\in \operatorname{Ext}\left(\mathcal{N}\right)}\sup_{\psi_{RA}} \widehat{D}_{\alpha}\!\left(\mathcal{N}\!\left(\psi\right)\Vert\operatorname{Tr}_{B_1}\!\left[\mathcal{P}\!\left(\psi\right)\right]\right)\\
        &= \inf_{\mathcal{P}\in \operatorname{Ext}\left(\mathcal{N}\right)}\inf_{\alpha \in (0,2]}\sup_{\psi_{RA}} \widehat{D}_{\alpha}\!\left(\mathcal{N}\!\left(\psi\right)\Vert\operatorname{Tr}_{B_1}\!\left[\mathcal{P}\!\left(\psi\right)\right]\right),
    \end{align}
    where the second equality follows from the monotonicity of $\alpha$-geometric unextendible entanglement with $\alpha$ and the third equality is arrived at by exchanging the infimums.
    
    The $\alpha$-geometric R\'enyi relative entropy $\widehat{D}_{\alpha}\!\left(\rho\Vert\sigma\right)$ is lower semi-continuous in $\left(\rho,\sigma\right)$~\cite[Lemma~A.3]{Fawzi2021}, and increases monotonically in $\alpha$  in the range $(0,2]$. Hence, we can employ the Mosonyi--Hiai minimax theorem from~\cite[Corollary~A.2]{MH11} (up to a minus sign on the outside therein) and establish that
    \begin{align}
        &\inf_{\mathcal{P}\in \operatorname{Ext}\left(\mathcal{N}\right)}\inf_{\alpha \in (0,2]}\sup_{\psi_{RA}} \widehat{D}_{\alpha}\!\left(\mathcal{N}\!\left(\psi\right)\Vert\operatorname{Tr}_{B_1}\!\left[\mathcal{P}\!\left(\psi\right)\right]\right)\notag \\
        &= \inf_{\mathcal{P}\in \operatorname{Ext}\left(\mathcal{N}\right)}\sup_{\psi_{RA}}\inf_{\alpha \in (0,2]} \widehat{D}_{\alpha}\!\left(\mathcal{N}\!\left(\psi\right)\Vert\operatorname{Tr}_{B_1}\!\left[\mathcal{P}\!\left(\psi\right)\right]\right)\\
        &= \inf_{\mathcal{P}\in \operatorname{Ext}\left(\mathcal{N}\right)}\sup_{\psi_{RA}}\lim_{\alpha \to 0^+} \widehat{D}_{\alpha}\!\left(\mathcal{N}\!\left(\psi\right)\Vert\operatorname{Tr}_{B_1}\!\left[\mathcal{P}\!\left(\psi\right)\right]\right) \\
        &= \inf_{\mathcal{P}\in \operatorname{Ext}\left(\mathcal{N}\right)}\sup_{\psi_{RA}} \widehat{D}_{0}\!\left(\mathcal{N}\!\left(\psi\right)\Vert\operatorname{Tr}_{B_1}\!\left[\mathcal{P}\!\left(\psi\right)\right]\right) .
    \end{align}
    This concludes the proof.
\end{proof}

\begin{proposition}\label{prop:geo_unext_ent_converge_Bel_Stas}
     The $\alpha$-geometric unextendible entanglement of quantum channels converges to the unextendible entanglement induced by the Belavkin--Staszewski relative entropy as $\alpha \to 1$:
     \begin{align}\label{eq:geo_unext_ent_converge_Bel_Stas}
         & \lim_{\alpha \to 1} \widehat{E}^u_{\alpha}\!\left(\mathcal{N}_{A\to B}\right) \notag \\
         & = \widehat{E}^u\!\left(\mathcal{N}_{A\to B}\right) \\
         & \coloneqq \inf_{\mathcal{P}\in \operatorname{Ext}\left(\mathcal{N}\right)} \widehat{D}\!\left(\mathcal{N}_{A\to B}\Vert \operatorname{Tr}_{B_1}\circ\mathcal{P}_{A\to B_1B_2}\right) .
         \label{eq:BS-rel-ent-induces-unext-ent}
     \end{align}
\end{proposition}
\begin{proof}
    See Appendix~\ref{app:geo_unext_ent_converg_Bel_Stas}.
\end{proof}

\begin{remark}
    Using a semidefinite program, we can compute the $\alpha$-geometric unextendible entanglement of a channel for $\alpha = 1+2^{-\ell}$, where $\ell$ is a positive integer. This quantity closely approximates the unextendible entanglement of the channel induced by the Belavkin--Staszewski relative entropy for large enough values of $\ell$. However, it remains open to determine a semidefinite program to compute the $\alpha$-geometric unextendible entanglement for $\alpha \in (0,1)$, which makes it difficult to estimate the min-geometric unextendible entanglement of a channel computationally. 
\end{remark}

\section{Applications}\label{sec:applications}

In this section, we discuss some applications of the unextendible entanglement of a point-to-point quantum channel in establishing upper bounds on some operational quantities of interest, including a channel's exact one-way distillable key (Section~\ref{sec:prob_distill_key}), its 
 probabilistic distillable entanglement (Section~\ref{sec:prob_distill_ent}), its zero-error private capacity (Section~\ref{sec:zero_err_private}), and its zero-error quantum capacity (Section~\ref{sec:zero_error_capacity}).

 \subsection{Exact one-way distillable key of a channel}\label{sec:prob_distill_key}

Quantum key distribution (QKD)~\cite{PhysRevLett.67.661, BB84, 10.1145/382780.382781} is the process of distributing information-theoretic secret keys between two parties, for the purpose of conducting private communication after the keys are established. The key distribution task establishes a maximally classically-correlated state between Alice and Bob, and it ensures secrecy by forcing any eavesdropper's system to be completely uncorrelated with this state. As such, the joint state between Alice, Bob, and an eavesdropper holding system $E$, after an ideal key distillation protocol, can be mathematically described as follows:
\begin{equation}
    \tau_{ABE} \coloneqq \frac{1}{K}\sum_{k=1}^{K}|k\rangle\!\langle k|_A\otimes|k\rangle\!\langle k|_B\otimes\sigma_{E}. 
\end{equation}
This is called an ideal tripartite key state. Such a tripartite key state is said to hold $\log_2 K$ secret bits.

In~\cite{Horodecki_2005, Horodecki_2009}, it was shown that the task of distilling tripartite secret keys is equivalent to establishing bipartite private states between Alice and Bob, which yield a secret key between Alice and Bob upon measurement. A quantum state $\rho_{ABA'B'}$ is called a bipartite private state~\cite{Horodecki_2005, K_Horodecki_2009} if Alice and Bob can extract a maximally classically-correlated state by applying local measurements, such that the resulting state is in product form with any purifying system of $\rho_{ABA'B'}$:  
\begin{multline}
    \left(\mathcal{M}_A\otimes\mathcal{M}_B\otimes \operatorname{Tr}_{A'B'}\right)\left(\psi^{\rho}_{ABA'B'E}\right) \\= \frac{1}{K}\sum_{k=1}^K|k\rangle\!\langle k|_A\otimes|k\rangle\!\langle k|_B\otimes \sigma_E,
    \label{eq:bi-priv-state-to-tri-key-state}
\end{multline}    
where $\psi^{\rho}_{ABA'B'E}$ is a purification of the state $\rho_{ABA'B'}$, $\mathcal{M}\!\left(\cdot\right) \coloneqq  \sum_k|k\rangle\!\langle k|\left(\cdot\right)|k\rangle\!\langle k|$ is a projective measurement channel, and $\sigma_E$ is an arbitrary state of the purifying system $E$.  Here $A$ and $B$ are the key systems, and $A'$ and $B'$ are the shield systems (see Figure~\ref{fig:QKD_protocol}). 

While a maximally entangled state is an example of a bipartite private state, it was shown in~\cite{Horodecki_2005,Horodecki_2009} that there exist private states that hold a finite number of secret bits but have vanishing distillable entanglement. Therefore, the task of secret-key distillation is distinct from entanglement distillation, and it is not easily understood using the resource theory of entanglement.

We consider the task of establishing a bipartite private state $\gamma^K_{CDC'D'}$ between Alice and Bob, capable of generating a secret key of  $\log_2 K$ bits, using multiple instances of a quantum channel $\mathcal{N}_{A\to B}$, an arbitrary state $\psi_{C\hat{C}}$, and a one-way LOCC superchannel $\Theta_{(A^n\to B^n)\to \left(\hat{C}\to C'DD'\right)}$. The result of the protocol acting on an input state $\psi_{C\hat{C}}$ is specified mathematically as follows:
\begin{equation}
    \Theta_{\left(A^n\to B^n\right)\to \left(\hat{C}\to C'DD'\right)}\!\left(\mathcal{N}^{\otimes n}_{A\to B}\right)\left(\psi_{C\hat{C}}\right) = \gamma^K_{CDC'D'},
\end{equation}
where the pre-processing and post-processing channels in $\Theta$ are taken to be isometric channels, with $C'$ and $D'$ as the corresponding purifying systems (see Figure~\ref{fig:QKD_protocol}).

\begin{figure}
	\centering
	\includegraphics[width = \linewidth]{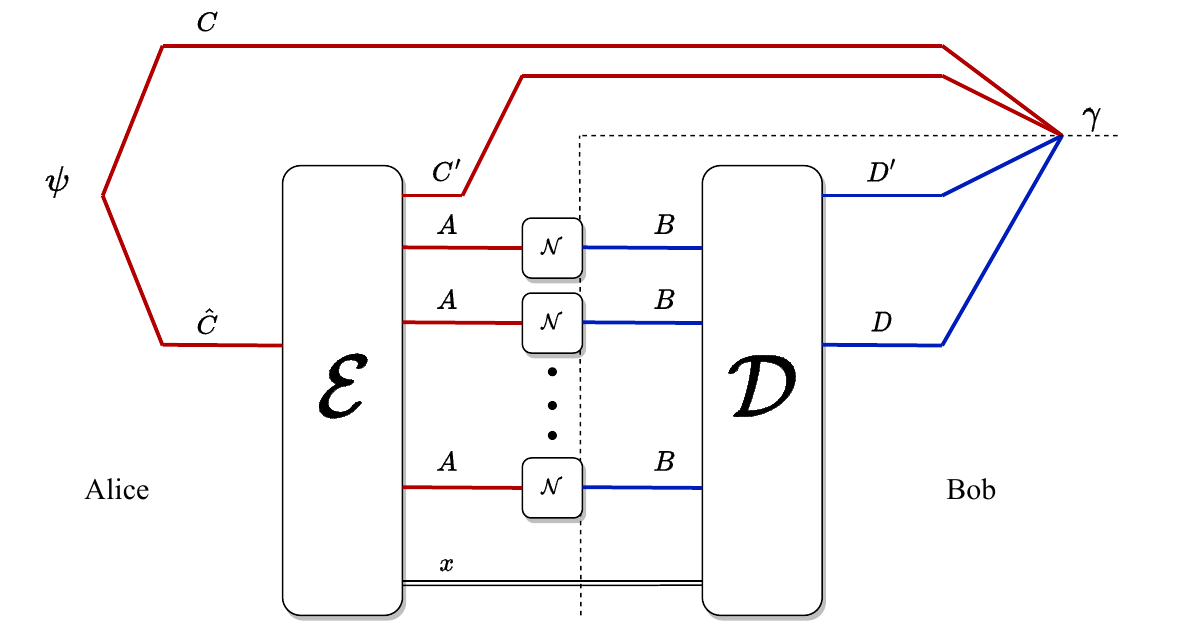}
	\caption{Diagrammatic representation of a protocol generating a bipartite private state $\gamma_{CDC'D'}$ between Alice and Bob through multiple uses of a quantum channel $\mathcal{N}_{A\to B}$. The protocol is enacted through the local operations $\mathcal{E}$ and $\mathcal{D}$ and one-way classical communication from Alice to Bob. Systems $C$ and $D$ form the key while $C'$ and $D'$ are the shield systems.}
	\label{fig:QKD_protocol}
\end{figure}

\subsubsection{One-shot exact one-way distillable key}

Let us first define the one-shot exact distillable key of a quantum channel; that is, we are considering the one-shot setting in which a single instance of the quantum channel $\mathcal{N}_{A\to B}$ is used. Let us define the set of all possible bipartite private states, holding $\log_2K$ secret bits, as follows:
\begin{equation}
    \mathcal{K} \coloneqq \left\{\gamma^K_{CDC'D'}: K \in \mathbb{N}\right\}, 
\end{equation}
where $\gamma^K_{CDC'D'}$ is a private state holding $\log_2 K$ secret bits shared between two parties possessing systems $CC'$ and $DD'$, respectively, with $C$ and $D$ key systems and $C'$ and $D'$ shield systems.

\begin{definition}
\label{def:1-shot-exact-dist-key}
    The one-shot exact distillable key of a quantum channel $\mathcal{N}_{A \to B}$ is the  number of secret bits that can be distilled by a single use of the channel assisted by a one-way LOCC superchannel $\Theta_{(A\to B)\to (\hat{C}\to C'DD')}$:
\begin{multline}\label{eq:dist_key_1shot_defn}
    K^{(1)}_{\operatorname{1WL}}\!\left(\mathcal{N}_{A\to B}\right) \coloneqq\\
    \sup_{\substack{\psi_{C\hat{C}}\in \mathcal{S}(C\hat{C}), \Theta \in \operatorname{1WL},\\ \gamma^K_{CDC'D'} \in \mathcal{K}}}\left\{\begin{array}{c}
    	\log_2 K:\\
    	\Theta\left(\mathcal{N}\right)\left(\psi_{C\hat{C}}\right) = \gamma^K_{CDC'D'}
    \end{array} \right\}.
\end{multline}
\end{definition}

Definition~\ref{def:1-shot-exact-dist-key} is motivated by the fact that once a bipartite private state with $\log_2 K$ bits of secrecy is established between two parties, a secret key of $\log_2 K$ bits can be distilled by only local operations.

We can also define the one-shot exact distillable key in a similar way when the channel is assisted by two-extendible superchannels:
\begin{multline}\label{eq:dist_key_2_ext_1shot_defn}
    K^{(1)}_{\operatorname{2-EXT}}\!\left(\mathcal{N}_{A\to B}\right) \coloneqq\\
    \sup_{\substack{\psi_{C\hat{C}}\in \mathcal{S}(C\hat{C}), \Theta \in \operatorname{2-EXT},\\ \gamma^K_{CDC'D'} \in \mathcal{K}}}\left\{\begin{array}{c}
    	\log_2 K:\\
    	\Theta\left(\mathcal{N}\right)\left(\psi_{C\hat{C}}\right) = \gamma^K_{CDC'D'}
    \end{array} \right\}.
\end{multline}
Since all one-way LOCC superchannels are two-extendible,
\begin{equation}
    	K^{(1)}_{\operatorname{2-EXT}}\!\left(\mathcal{N}_{A\to B}\right) \ge K^{(1)}_{\operatorname{1WL}}\!\left(\mathcal{N}_{A\to B}\right).
\end{equation}

\begin{proposition}\label{prop:distill_key_ub_fin}
    The number of secret bits that can be distilled by using $n$ instances of a quantum channel $\mathcal{N}_{A\to B}$, assisted by one-way LOCC or two-extendible superchannels, is bounded from above by the min-geometric unextendible entanglement of the channel $\mathcal{N}_{A\to B}$ (as defined in~\eqref{eq:min_unext_ent_ch_def}). That is,
    \begin{equation}\label{eq:distill_key_ub_fin}
        \frac{1}{n}K^{(1)}_{\operatorname{1WL}}\!\left(\mathcal{N}^{\otimes n}_{A\to B}\right) \le \frac{1}{n}K^{(1)}_{\operatorname{2-EXT}}\!\left(\mathcal{N}^{\otimes n}_{A\to B}\right) \le \widehat{E}^u_{\min}\!\left(\mathcal{N}_{A\to B}\right).
    \end{equation}
\end{proposition}

\begin{proof}
    Let $\Theta_{(A^n\to B^n)\to (\hat{C}\to DC'D')}$ be a two-extendible superchannel that transforms $n$ instances of quantum channel $\mathcal{N}_{A\to B}$ into a channel that acts on an input state $\psi_{C\hat{C}}$ to yield a private state $\gamma^K_{CDC'D'}$ holding $\log_2 K$ bits of secrecy with probability $p$; i.e.,
    \begin{multline}
        \Theta_{(A^n\to B^n)\to (\hat{C}\to DC'D')}\!\left(\mathcal{N}^{\otimes n}_{A\to B}\right)\left(\psi_{C\hat{C}}\right) = \gamma^K_{CDC'D'}.
    \end{multline}
    The $\alpha$-geometric unextendible entanglement of a channel can be bounded from below as follows: using Theorem~\ref{theo:unext_ent_state_le_unext_ent_ch_gen}, monotonicity of  and subadditivity of $\alpha$-geometric unextendible entanglement under tensor products from Proposition~\ref{prop:unext_ent_subadditive},
    \begin{align}
        n\widehat{E}^u_{\alpha}\!\left(\mathcal{N}\right) & \ge \widehat{E}^u_{\alpha}\!\left(\mathcal{N}^{\otimes n}\right)\\
        & \ge \widehat{E}^u_{\alpha}\!\left(\Theta\left(\mathcal{N}^{\otimes n}\right)\right)\\
        &\ge \widehat{E}^u_{\alpha}\!\left(\Theta\left(\mathcal{N}^{\otimes n} \right)\left(\psi_{C\hat{C}}\right)\right)\\
        &=  \widehat{E}^u_{\alpha}\!\left(\gamma^K_{CDC'D'}\right)\\
        &\ge \log_2 K,
    \end{align}
    where the first inequality follows from the subadditivity of $\alpha$-geometric unextendible entanglement (Proposition~\ref{prop:unext_ent_subadditive}), the second inequality follows from the monotonicity of unextendible entanglement of a channel under the action of a two-extendible superchannel (Theorem~\ref{theo:two_ext_monotonic_p2p}), the third inequality follows from Theorem~\ref{theo:unext_ent_state_le_unext_ent_ch_gen}, and the last inequality follows from~\cite[Corollary 22]{WWW19}. 
    
    We can get the tightest possible upper bound with this technique by taking the limit $\alpha\to 0$ and applying Proposition~\ref{prop:limit-zero-min-geo}, arriving at the following inequality:
    \begin{equation}
        \widehat{E}^u_{\min}\!\left(\mathcal{N}_{A\to B}\right) \ge \frac{1}{n}\log_2 K.
    \end{equation}
    Since this inequality holds for every two-extendible superchannel $\Theta$, input state $\psi_{C\hat{C}}$, and secret-key dimension $K$, we conclude~\eqref{eq:distill_key_ub_fin}.
\end{proof}

\medskip

\subsubsection{Asymptotic exact one-way distillable key}

Now let us now consider the asymptotic setting in which an arbitrarily large number of independent uses of a channel are allowed along with a restricted set of superchannels to establish a secret key between two parties. We begin by defining the asymptotic exact one-way distillable key of a quantum channel.

\begin{definition}
	The asymptotic exact one-way distillable key of a channel $\mathcal{N}_{A\to B}$ is the maximum achievable rate at which secret bits can be distilled by an arbitrarily large number of channel uses in parallel, along with one-way LOCC superchannels:
	\begin{equation}
		K_{\operatorname{1WL}}\!\left(\mathcal{N}_{A\to B}\right) \coloneqq \liminf_{n\to\infty}\frac{1}{n} K^{(1)}_{\operatorname{1WL}}\!\left(\mathcal{N}^{\otimes n}_{A\to B}\right).
	\end{equation}
\end{definition}

We can define the asymptotic two-extendible exact distillable key of the channel by relaxing the constraint on the superchannels to be two-extendible:
\begin{equation}
	K_{\operatorname{2-EXT}}\!\left(\mathcal{N}_{A\to B}\right) \coloneqq \liminf_{n\to\infty }\frac{1}{n} K^{(1)}_{\operatorname{2-EXT}}\!\left(\mathcal{N}^{\otimes n}_{A\to B}\right).
\end{equation}
The definitions imply the following inequalities:
\begin{equation}
	K_{\operatorname{2-EXT}}\!\left(\mathcal{N}_{A\to B}\right) \ge K^{(1)}_{\operatorname{2-EXT}}\!\left(\mathcal{N}_{A\to B}\right) \ge K^{(1)}_{\operatorname{1WL}}\!\left(\mathcal{N}_{A\to B}\right),
\end{equation}
and
\begin{equation}
	K_{\operatorname{2-EXT}}\!\left(\mathcal{N}_{A\to B}\right) \ge K_{\operatorname{1WL}}\!\left(\mathcal{N}_{A\to B}\right) \ge K^{(1)}_{\operatorname{1WL}}\!\left(\mathcal{N}_{A\to B}\right).
\end{equation}

\begin{corollary}\label{cor:prob_key_distill_2ext}
The asymptotic exact distillable key of a channel $\mathcal{N}_{A\to B}$, assisted by one-way LOCC or two-extendible superchannels, is bounded from above by the min-geometric unextendible entanglement of the channel:
\begin{equation}
    \widehat{E}^u_{\min}\!\left(\mathcal{N}_{A\to B}\right) \ge K_{\operatorname{2-EXT}}\!\left(\mathcal{N}_{A\to B}\right) \ge K_{\operatorname{1WL}}\!\left(\mathcal{N}_{A\to B}\right).
\end{equation}
\end{corollary}

\begin{proof}
    The proof follows from the fact that Proposition~\ref{prop:distill_key_ub_fin} is true for all values of $n\in \mathbb{N}$.
\end{proof}

\subsection{Exact one-way distillable entanglement of a channel}\label{sec:prob_distill_ent}

All current quantum network models rely on distant parties holding one or more shares of a maximally entangled bipartite state (also known as an ebit). It is crucial to analyse the ability of a quantum channel to distribute entanglement between its participants such that a quantum network can be sustained by employing several of these channels.

\begin{figure}
	\centering
	\includegraphics[width = \linewidth]{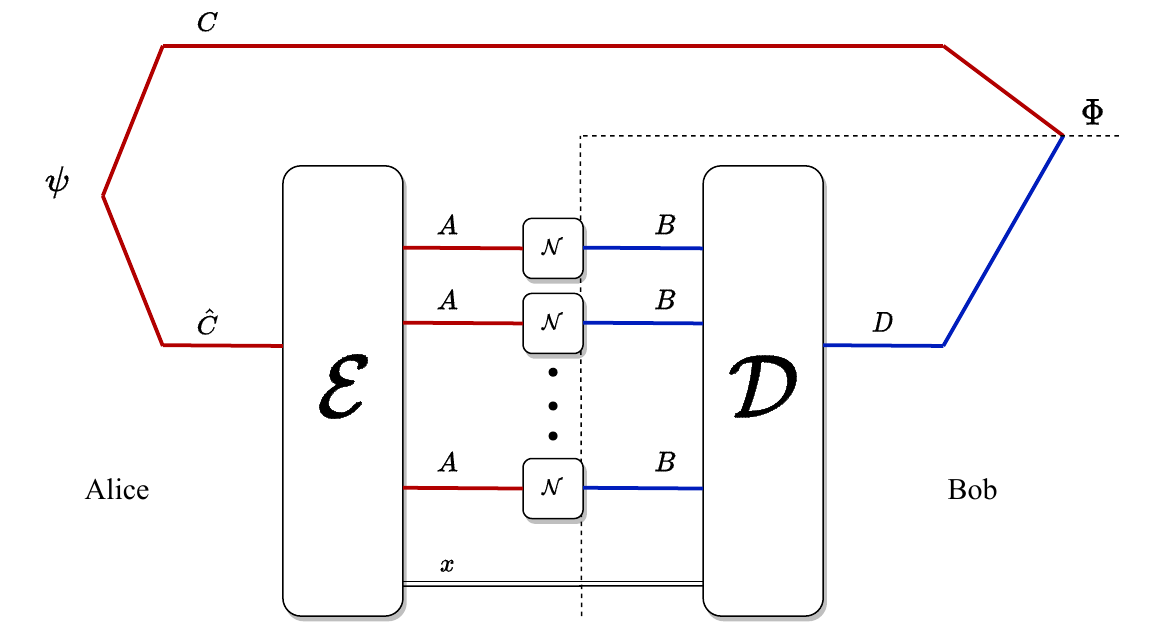}
	\caption{Diagrammatic representation of the protocol establishing a $d$-dimensional maximally entangled state $\Phi^d_{CD}$ between Alice and Bob through multiple uses of a quantum channel $\mathcal{N}_{A\to B}$ in parallel. The protocol is enacted through the local operations $\mathcal{E}$ and $\mathcal{D}$ and one-way classical communication from Alice to Bob. }
	\label{fig:Ent_distill_p2p}
\end{figure}

We consider the task of entanglement distillation where Alice prepares ebits locally and sends shares of the ebits to Bob using a noisy quantum channel $\mathcal{N}_{A\to B}$ assisted by local operations and one-way classical communication from Alice to Bob. For a more general setting, we allow Alice to prepare any arbitrary bipartite state instead of ebits (see Figure~\ref{fig:Ent_distill_p2p}).

\subsubsection{One-shot exact one-way distillable entanglement}

Let us first consider the one-shot case where Alice and Bob use a point-to-point quantum channel assisted by one-way LOCC superchannels to distill a maximally entangled state. The ability of a channel to distill a maximally entangled state between the two parties can be quantified by its one-shot exact one-way distillable entanglement, as defined below. 

\begin{definition}
    The one-shot exact one-way distillable entanglement of a channel $\mathcal{N}_{A\to B}$ is the maximum number of ebits that can be established between two parties using the channel $\mathcal{N}_{A\to B}$ and a one-way LOCC superchannel $\Theta_{(A\to B)\to (\hat{C}\to D)}$:
	\begin{multline}
 E^{(1)}_{D,\operatorname{1WL}}\!\left(\mathcal{N}_{A\to B}\right) \coloneqq \\ 	\sup_{\substack{\psi_{C\hat{C}}\in \mathcal{S}(C\hat{C}),\\ \Theta \in \operatorname{1WL}}}\left\{\begin{array}{c}
			\log_2 d :\\
			\Theta\left(\mathcal{N}\right)\left(\psi_{C\hat{C}}\right) = \Phi^d_{CD}
		\end{array} \right\}.
	\end{multline}
\end{definition}

We can relax the set of allowed superchannels to be the set of two-extendible superchannels and define the one-shot two-extendible exact distillable entanglement of a channel:
\begin{multline}\label{eq:one_shot_dist_ent_two_ext}
	E^{(1)}_{D,\operatorname{2-EXT}}\!\left(\mathcal{N}_{A\to B}\right) \coloneqq \\ 	\sup_{\substack{\psi_{C\hat{C}}\in \mathcal{S}(C\hat{C}),\\ \Theta \in \operatorname{2-EXT}}}\left\{\begin{array}{c}
			\log_2 d :\\
			\Theta\left(\mathcal{N}\right)\left(\psi_{C\hat{C}}\right) = \Phi^d_{CD}
		\end{array} \right\}.
\end{multline}

An ebit shared between Alice and Bob can always be used to distill one bit of secret key between them. Hence, the number of secret bits that can be distilled from a channel is no less than the number of ebits that can be distilled from the channel, which leads to the following inequalities that hold for all $n\in \mathbb{N}$:
\begin{equation}
    \frac{1}{n}E^{(1)}_{D,\operatorname{1WL}}\!\left(\mathcal{N}^{\otimes n}_{A\to B}\right) \le
    \frac{1}{n}K^{(1)}_{\operatorname{1WL}}\!\left(\mathcal{N}^{\otimes n}_{A\to B}\right)\le \widehat{E}^u_{\min}\!\left(\mathcal{N}_{A\to B}\right),
\end{equation}
and 
\begin{equation}\label{eq:2_ext_dist_ent_1shot_le_min_unext_ent}
    \frac{1}{n}E^{(1)}_{D,\operatorname{2-EXT}}\!\left(\mathcal{N}^{\otimes n}_{A\to B}\right)  \le \frac{1}{n}K^{(1)}_{\operatorname{2-EXT}}\!\left(\mathcal{N}^{\otimes n}_{A\to B}\right) \le \widehat{E}^u_{\min}\!\left(\mathcal{N}_{A\to B}\right).
\end{equation}

\subsubsection{Asymptotic exact one-way distillable entanglement of a channel}
In general, we expect the exact one-way distillable key of the channel to be strictly larger than the exact one-way distillable entanglement of the channel because there exist private states that hold a finite number of exact secret bits but have no exact distillable entanglement~\cite{Horodecki_2005}. Thus, the min-geometric unextendible entanglement of a channel is a tighter bound on the one-shot exact one-way distillable key of the channel than it is on the one-shot exact one-way distillable entanglement of the channel.

The asymptotic exact one-way distillable entanglement of a channel is defined as follows:
\begin{equation}
	E_{D,\operatorname{1WL}}\!\left(\mathcal{N}_{A\to B}\right) \coloneqq \liminf_{n\to \infty} \frac{1}{n}E^{(1)}_{D,\operatorname{1WL}}\!\left(\mathcal{N}^{\otimes n}_{A\to B}\right),
\end{equation}
which is the maximum achievable rate at which ebits can be distilled with an arbitrarily large number of uses of the channel $\mathcal{N}_{A\to B}$, assisted by a one-way LOCC superchannel. Similarly, the asymptotic generalization of the one-shot quantity in~\eqref{eq:one_shot_dist_ent_two_ext} is the following:
\begin{equation}\label{eq:dist_ent_two_ext}
	E_{D,\operatorname{2-EXT}}\!\left(\mathcal{N}_{A\to B}\right) \coloneqq \liminf_{n\to \infty} \frac{1}{n}E^{(1)}_{D,\operatorname{2-EXT}}\!\left(\mathcal{N}^{\otimes n}_{A\to B}\right),
\end{equation}
which is the asymptotic two-extendible exact distillable entanglement of the channel.
For every given channel $\mathcal{N}_{A\to B}$, the following inequalities are consequences of their definitions:
\begin{equation}
	E_{D,\operatorname{2-EXT}}\!\left(\mathcal{N}_{A\to B}\right) \ge E_{D,\operatorname{1WL}}\!\left(\mathcal{N}_{A\to B}\right) \ge E^{(1)}_{D,\operatorname{1WL}}\!\left(\mathcal{N}_{A\to B}\right).
\end{equation}
This leads us to identify the min-geometric unextendible entanglement as an upper bound on the exact one-way distillable entanglement of a channel as we state formally in Corollary~\ref{cor:distill_ent_ub_fin}.
\begin{corollary}\label{cor:distill_ent_ub_fin}
    The exact one-way distillable entanglement of a channel is bounded from above by the min-geometric unextendible entanglement of the channel. That is, the following inequality holds for every channel $\mathcal{N}_{A\to B}$:
    \begin{equation}
        E_{D,\operatorname{1WL}}\!\left(\mathcal{N}_{A\to B}\right) \le \widehat{E}^u_{\min}\!\left(\mathcal{N}_{A\to B}\right). 
    \end{equation}
\end{corollary}

\begin{proof}
    The inequalities in~\eqref{eq:2_ext_dist_ent_1shot_le_min_unext_ent} hold for all $n \in \mathbb{N}$. The exact one-way distillable entanglement of a channel $\mathcal{N}_{A\to B}$ can be bounded as follows:
\begin{align}
    E_{D,\operatorname{1WL}}\!\left(\mathcal{N}_{A\to B}\right) &\le E_{D,\operatorname{2-EXT}}\!\left(\mathcal{N}_{A\to B}\right)\\
    &= \liminf_{n\to \infty} \frac{1}{n}E_{D,\operatorname{2-EXT}}\!\left(\mathcal{N}^{\otimes n}_{A\to B}\right)\\
    &\le \widehat{E}^u_{\min}\!\left(\mathcal{N}_{A\to B}\right),
\end{align}
where the first inequality follows from the fact that all one-way LOCC sueprchannels are two-extendible, the equality follows from the definition of the asymptotic two-extendible exact distillable entanglement of a channel given in~\eqref{eq:dist_ent_two_ext}, and the final inequality follows from~\eqref{eq:2_ext_dist_ent_1shot_le_min_unext_ent}.
\end{proof}

\subsection{Zero-error private capacity assisted by one-way LOCC}\label{sec:zero_err_private}

The private capacity of a channel is defined as the largest rate at which private bits can be sent through the channel. 

Consider a general protocol for private communication:
\begin{enumerate}
    \item Alice encodes her classical message $|m\rangle\!\langle m|_{A_0}$ into a quantum state by means of  an encoder $\mathcal{E}_{A_0\to A}$.
    \begin{equation}
        \omega_A^m \coloneqq \mathcal{E}_{A_0\to A}\!\left(|m\rangle\!\langle m|_{A_0}\right).
    \end{equation}
    
    \item Alice sends the encoded data to Bob through the channel $\mathcal{N}_{A\to B}$, which is extended by an isometric channel $\mathcal{U}^{\mathcal{N}}_{A \to BE}$.
    
    \item Bob decodes the state he received by applying a POVM $\left\{\Lambda^{\hat{m}}_B\right\}_{\hat{m}}$. The final state of the overall $BE$ system is,
    \begin{align}
        & \rho_{B_0 E}^m \notag \\
        &= \mathcal{D}_{B\to B_0}\!\left(\mathcal{U}^{\mathcal{N}}_{A \to BE}\!\left(\omega^m_A\right)\right)\\ &=\sum_{\hat{m}}\operatorname{Tr}_B\!\left[\Lambda^{\hat{m}}_B\left(\mathcal{U}^{\mathcal{N}}_{A \to BE}\circ\mathcal{E}_{A_0\to A}\!\left(|m\rangle\!\langle m|_{A_0}\right)\right)\right]|\hat{m}\rangle\!\langle \hat{m}|_{B_0}.
    \end{align}
\end{enumerate}
If the communication is successful and private, there exists a state $\sigma_E$ of the environment system $E$, such that for every message $m$, the final state has the form
\begin{equation}
    \rho^m_{B_0 E} = |m\rangle\!\langle m|_{B_0} \otimes \sigma_E.
    \label{eq:privacy_cond_states}
\end{equation}

The action of the encoder $\mathcal{E}_{A_0\to A}$ and the decoder $\mathcal{D}_{B\to B_0}$ realizes a superchannel $\Theta_{(A\to B)\to (A_0\to B_0)}$ constructed by local operations. Let $\mathcal{P}_{A\to BE}$ be an arbitrary extension of the quantum channel $\mathcal{N}_{A\to B}$, where the system $E$ can be accessed by an eavesdropper. Note that this extension is different from the extension of channels that we have discussed in the rest of this work, as the system $E$ is not required to be isomorphic to system $B$. 

This can be understood in a purified setting where the eavesdropper can access the environment system $E$ output from the isometric extension $\mathcal{U}^{\mathcal{N}}_{A\to BE}$~\cite{Devetak2005} of the quantum channel $\mathcal{N}_{A\to B}$. However, we assume Alice's and Bob's laboratories to be secure in this setting and the eavesdropper cannot access any output systems from the encoding or decoding channel used in the protocol. Let $\Upsilon_{(A\to BE)\to (A_0\to B_0E_0)}$ be an arbitrary extension of the superchannel $\Theta$ that includes any post-processing that the eavesdropper can perform on their system $E$. The privacy condition in~\eqref{eq:privacy_cond_states} implies that for every extended superchannel $\Upsilon$, and isometric extension $\mathcal{U}^{\mathcal{N}}_{A\to BE}$, there exists a quantum state $\sigma$ such that,
\begin{equation}
    \Upsilon_{(A\to BE)\to (A_0\to B_0E_0)}\!\left(\mathcal{U}^{\mathcal{N}}_{A\to BE}\right) = \overline{\Delta}_{A_0\to B_0}\otimes \mathcal{A}^{\sigma}_{E_0},
\end{equation}
where $\overline{\Delta}_{A_0\to B_0}$ is the completely dephasing channel,
\begin{equation}
    \overline{\Delta}_{A_0\to B_0}\!\left(\rho_{A_0}\right) = \sum_k \langle k|\rho_{A_0}|k\rangle|k\rangle\!\langle k|_{B_0},
\end{equation}
and $\mathcal{A}^{\sigma}_{E_0}$ denotes an appending channel that prepares the state $\sigma$ on system $E_0$.

We assume that the eavesdropper only intends to learn about the information being sent from Alice to Bob, and not distort the data. Hence, we can assume that the channel established between Alice, Bob, and the eavesdropper is nonsignaling from the eavesdropper to Bob; i.e.,
\begin{equation}\label{eq:supch_ext_private_comm_cond}
    \operatorname{Tr}_{E_0}\circ\left(\Upsilon\left(\mathcal{U}^{\mathcal{N}}_{A\to BE}\right)\right) = \Theta\left(\operatorname{Tr}_{E}\circ\mathcal{U}^{\mathcal{N}}_{A\to BE}\right) = \Theta\left(\mathcal{N}\right).
\end{equation}
Let $\operatorname{Ext}_P\!\left(\Theta\right)$ denote the set of all such extensions of the superchannel $\Theta$.

Now consider a superchannel $\Theta_{(A^n\to B^n)\to (A_0\to B_0)}$, composed of local operations by Alice and Bob, that acts on $n$ instances of a quantum channel $\mathcal{N}_{A\to B}$ to communicate $\log_2 d$ bits of private data from Alice to Bob. An extended superchannel $\Upsilon_{(A\to BE)\to (A_0\to B_0E_0)}$ acts on some isometric extension $\mathcal{U}^{\mathcal{N}}_{A\to BE}$ of the channel $\mathcal{N}_{A\to B}$ such that it lies in the set $\operatorname{Ext}_P$. We define the one-shot zero-error private capacity of a quantum channel $\mathcal{N}_{A\to B}$ as follows:
\begin{multline}\label{eq:private_cap_lo_def}
    P^{\left(1\right)}_{0,\operatorname{LO}}\!\left(\mathcal{N}_{A\to B}\right) \\ \coloneqq \sup_{\substack{d \in \mathbb{N},\\ \Theta\in \operatorname{LO}}}\left\{
    \begin{array}{c}
        \log_2 d : \\
        \forall~ \Upsilon \in \operatorname{Ext}_P\!\left(\Theta\right) \ \exists \sigma \in \mathcal{S}(E_0) \\
        \Upsilon\!\left(\left(\mathcal{U}^{\mathcal{N}}_{A\to BE}\right)\right) = \overline{\Delta}_{A_0\to B_0}\otimes \mathcal{A}^{\sigma}_{E_0}
    \end{array}
    \right\},
\end{multline}
where $\overline{\Delta}_{A_0\to B_0}$ is a $d$-dimensional dephasing channel and $\mathcal{U}^{\mathcal{N}}_{A\to BE}$ is an arbitrary isometric extension of the channel~$\mathcal{N}_{A\to B}$.

Another quantity of interest is the zero-error private capacity of a channel assisted by one-way LOCC superchannels, where Alice can send classical data of arbitrary size to Bob along with the encoded quantum state. This classical data is naturally not private and can be copied by the eavesdropper. However, the nonsignaling condition from~\eqref{eq:supch_ext_private_comm_cond} still holds. 
We define the one-shot zero-error private capacity assisted by one-way LOCC superchannels as follows:
\begin{multline}\label{eq:private_cap_1W_def}
    P^{\left(1\right)}_{0,\operatorname{1WL}}\!\left(\mathcal{N}_{A\to B}\right) \\ \coloneqq \sup_{\substack{d\in \mathbb{N},\\ \Theta\in \operatorname{1WL}}}\left\{
    \begin{array}{c}
        \log_2 d : \\
        \forall~ \Upsilon \in \operatorname{Ext}_P\!\left(\Theta\right) \ \exists \sigma \in \mathcal{S}(E_0) \\
        \Upsilon\!\left(\left(\mathcal{U}^{\mathcal{N}}_{A\to BE}\right)\right) = \overline{\Delta}_{A_0\to B_0}\otimes \mathcal{A}^{\sigma}_{E_0}
    \end{array}
    \right\},
\end{multline}
where the allowed superchannels are one-way LOCC superchannels instead of superchannels that can be constructed by only local operations.

We can further relax the set of superchannels to be two-extendible. The eavesdropper in this case also has access to a memory system that cannot necessarily be  copied owing to its quantum nature. We define the one-shot zero-error private capacity assisted by two-extendible superchannels as follows:
\begin{multline}\label{eq:private_cap_2_ext_def}
    P^{\left(1\right)}_{0,\operatorname{2-EXT}}\!\left(\mathcal{N}_{A\to B}\right) \\ \coloneqq \sup_{\substack{d\in \mathbb{N},\\ \Theta\in \operatorname{2-EXT}}}\left\{
    \begin{array}{c}
        \log_2 d : \\
        \forall~ \Upsilon \in \operatorname{Ext}_P\!\left(\Theta\right) \ \exists \sigma \in \mathcal{S}(E_0) \\
        \Upsilon\!\left(\left(\mathcal{U}^{\mathcal{N}}_{A\to BE}\right)\right) = \overline{\Delta}_{A_0\to B_0}\otimes \mathcal{A}^{\sigma}_{E_0}
    \end{array}
    \right\}.
\end{multline}
Since the three sets of superchannels in consideration follow the hierarchy,
\begin{equation}
    \operatorname{LO} \subseteq \operatorname{1WL} \subseteq \operatorname{2-EXT},
\end{equation}
the respective private capacities obey the following inequalities:
\begin{equation}
    P^{\left(1\right)}_{0,\operatorname{LO}}\!\left(\mathcal{N}_{A\to B}\right) \le P^{\left(1\right)}_{0,\operatorname{1WL}}\!\left(\mathcal{N}_{A\to B}\right) \le P^{\left(1\right)}_{0,\operatorname{2-EXT}}\!\left(\mathcal{N}_{A\to B}\right). 
\end{equation}

Once again, we can define asymptotic versions of these quantities. The zero-error private capacity of a quantum channel $\mathcal{N}_{A\to B}$ is defined as
\begin{equation}
    P_{0,\operatorname{LO}}\!\left(\mathcal{N}_{A\to B}\right) \coloneqq \liminf_{n\to \infty}\frac{1}{n}P^{(1)}_{0,\operatorname{LO}}\!\left(\mathcal{N}_{A\to B}^{\otimes n}\right),
\end{equation}
the forward-assisted zero-error private capacity of the channel as
\begin{equation}
    P_{0,\operatorname{1WL}}\!\left(\mathcal{N}_{A\to B}\right) \coloneqq \liminf_{n\to \infty}\frac{1}{n}P^{(1)}_{0,\operatorname{1WL}}\!\left(\mathcal{N}^{\otimes n}_{A\to B}\right),
\end{equation}
and the zero-error private capacity of the channel assisted by two-extendible superchannels as
\begin{equation}
    P_{0,\operatorname{2-EXT}}\!\left(\mathcal{N}_{A\to B}\right) \coloneqq \liminf_{n\to \infty}\frac{1}{n}P^{(1)}_{0,\operatorname{2-EXT}}\!\left(\mathcal{N}^{\otimes n}_{A\to B}\right).
\end{equation}
Since Alice and Bob can always choose a superchannel $\Theta$ that acts independently on all instances of the channel $\mathcal{N}_{A\to B}$, we have the following inequalities:
\begin{align}
    P^{(1)}_{0,\operatorname{LO}}\!\left(\mathcal{N}_{A\to B}\right) &\le P_{0,\operatorname{LO}}\!\left(\mathcal{N}_{A\to B}\right),\\
    P^{(1)}_{0,\operatorname{1WL}}\!\left(\mathcal{N}_{A\to B}\right) &\le P_{0,\operatorname{1WL}}\!\left(\mathcal{N}_{A\to B}\right),\\
    P^{(1)}_{0,\operatorname{2-EXT}}\!\left(\mathcal{N}_{A\to B}\right) &\le P_{0,\operatorname{2-EXT}}\!\left(\mathcal{N}_{A\to B}\right).
\end{align}

\subsubsection{Private capacity and secret key distillation}

The zero-error private capacity of a channel is closely related to its exact distillable key. A secret key of $m$ bits can be established by sending $m$ bits of private data from Alice to Bob. Thus, the number of private bits that can be transmitted using a quantum channel $\mathcal{N}_{A\to B}$ and any superchannel $\Theta$ is no larger than the number of bits of a secret key that can be established using the same channel $\mathcal{N}_{A\to B}$ and the superchannel $\Theta$. 

In addition, if Alice has the ability to send public classical data to Bob without using the channel $\mathcal{N}_{A\to B}$, the secret key distillation protocol can be converted into a private channel. Alice and Bob can use the one-time-pad scheme to send $m$ bits of private data from Alice to Bob through the public classical channel and consume an $m$-bit secret key generated from the channel $\mathcal{N}_{A\to B}$.  

The above arguments lead to the following conclusion: the one-shot zero-error private capacity of a channel assisted by one-way LOCC superchannels is equal to the one-shot exact one-way distillable key of the channel; that is,
\begin{equation}\label{eq:0_priv_cap_le_ex_distill_key_1W}
    P^{(1)}_{0,\operatorname{1WL}}\!\left(\mathcal{N}_{A\to B}\right) = K^{(1)}_{\operatorname{1WL}}\!\left(\mathcal{N}_{A\to B}\right),
\end{equation}
where $K^{(1)}_{\operatorname{1WL}}\!\left(\cdot\right)$ was defined in~\eqref{eq:dist_key_1shot_defn}. Since two-extendible superchannels also allow Alice to send an arbitrary number of classical bits to Bob, the one-shot zero-error private capacity of the channel is equal to the following quantity:
\begin{equation}\label{eq:0_priv_cap_le_ex_distill_key_2_ext}
    P^{(1)}_{0,\operatorname{2-EXT}}\!\left(\mathcal{N}_{A\to B}\right) = K^{(1)}_{\operatorname{2-EXT}}\!\left(\mathcal{N}_{A\to B}\right),
\end{equation}
where $K^{(1)}_{\operatorname{2-EXT}}\!\left(\cdot\right)$ was defined in~\eqref{eq:dist_key_2_ext_1shot_defn}. Combining the equality in~\eqref{eq:0_priv_cap_le_ex_distill_key_2_ext} and the inequality in Proposition~\ref{prop:distill_key_ub_fin}, we conclude the following proposition:

\begin{proposition}\label{prop:priv_bits_le_min_geo_unext_ent}
    Given $n$ instances of a quantum channel $\mathcal{N}_{A\to B}$, assisted by one-way LOCC or two-extendible superchannels, the rate at which private bits can be transmitted without error is bounded from above by the min-geometric unextendible entanglement of the channel $\mathcal{N}_{A\to B}$; that is, for all $n\in \mathbb{N}$,
    \begin{align}
        \frac{1}{n}P^{(1)}_{0,\operatorname{1WL}}\!\left(\mathcal{N}^{\otimes n}_{A\to B}\right) &\le \frac{1}{n}P^{(1)}_{0,\operatorname{2-EXT}}\!\left(\mathcal{N}^{\otimes n}_{A\to B}\right)\notag\\
        &\le \widehat{E}^u_{\min}\!\left(\mathcal{N}_{A\to B}\right).
    \end{align}
\end{proposition}

Since Proposition~\ref{prop:priv_bits_le_min_geo_unext_ent} holds for all values of $n\in \mathbb{N}$, the min-geometric unextendible entanglement of a channel is an upper bound on the zero-error private capacity of the channel, as we state formally in Corollary~\ref{cor:0_err_priv_cap_ub} below.

\begin{corollary}\label{cor:0_err_priv_cap_ub}
    The zero-error private capacity of a channel $\mathcal{N}_{A\to B}$, assisted by one-way LOCC or two-extendible superchannels, is bounded from above by the min-geometric unextendible entanglement of the channel $\mathcal{N}_{A\to B}$; i.e.,
    \begin{equation}\label{eq:zero_error_priv_cap_ub_min_geo}
        P_{0,\operatorname{1WL}}\!\left(\mathcal{N}_{A\to B}\right) \le P_{0,\operatorname{2-EXT}}\!\left(\mathcal{N}_{A\to B}\right) \le \widehat{E}^u_{\operatorname{min}}\!\left(\mathcal{N}_{A\to B}\right).
    \end{equation}
\end{corollary}

\subsection{Zero-error quantum capacity assisted by one-way LOCC}\label{sec:zero_error_capacity}

The zero-error quantum capacity of a channel is the maximum rate at which the channel can transmit quantum information with zero error over an arbitrarily large number of channel uses~\cite{Shirokov_2015} with the assistance of local operations. The notion of zero-error quantum capacity of a channel can be extended to the case where the channel is assisted by one-way LOCC superchannels, which has significance not only from a theoretical perspective but also from a practical viewpoint, given existing state-of-the-art classical networks.

Here we look at the zero-error quantum capacity of a channel assisted by one-way LOCC superchannels. Using one instance of a quantum channel $\mathcal{N}_{A\to B}$ and an arbitrary one-way ideal classical channel, the following channel can be simulated between the two parties:
\begin{equation}
    \Theta\left(\mathcal{N}_{A\to B}\right) \coloneqq \sum_x \mathcal{D}^x_{B \to D}\circ\mathcal{N}_{A\to B}\circ \mathcal{E}^x_{C\to A},
\end{equation}
where $\{\mathcal{E}^x_{C\to A}\}_x$ is a set of completely-positive maps whose sum is trace preserving and $\{\mathcal{D}^x_{B \to D}\}_x$ is a set of quantum channels.

The one-shot zero-error quantum capacity of a channel $\mathcal{N}_{A\to B}$ assisted by one-way LOCC superchannels is defined as follows: 
\begin{multline}
    Q^{(1)}_{0,\operatorname{1WL}}(\mathcal{N}_{A\to B})\\ \coloneqq \sup_{\substack{d \in \mathbb{N},\\\Theta\in \operatorname{1WL}}}\left\{\log_2d: \Theta\left(\mathcal{N}\right) = \operatorname{id}^d_{C\to D}\right\},    
\end{multline}
where $\operatorname{id}^d_{C\to D}$ is the $d$-dimensional identity channel. We can relax the set of allowed superchannels to be the set of two-extendible superchannels, and we can define the one-shot zero-error quantum capacity of a quantum channel assisted by two-extendible superchannels as follows:
\begin{multline}\label{eq:1shot_q_cap_2_ext_defn}
    Q^{(1)}_{0,\operatorname{2-EXT}}(\mathcal{N}_{A\to B}) \\ \coloneqq \sup_{\substack{d \in \mathbb{N},\\ \Theta \in \operatorname{2-EXT}}}\left\{\log_2 d :  \Theta\left(\mathcal{N}_{A\to B}\right) = \operatorname{id}^d_{C\to D}\right\}.
\end{multline}
Note that the following inequality holds because all one-way LOCC superchannels are two-extendible:
\begin{equation}\label{eq:1shot_q_cap_1W_le_2_ext}
    Q^{(1)}_{0,\operatorname{1WL}}(\mathcal{N}_{A\to B}) \leq Q^{(1)}_{0,\operatorname{2-EXT}}(\mathcal{N}_{A\to B}).
\end{equation}

The zero-error quantum capacity of a channel assisted by one-way LOCC superchannels can be defined in terms of the one-shot zero-error quantum capacity of the channel assisted by one-way LOCC superchannels as follows:
\begin{equation}\label{eq:1WL_0_capacity}
    Q_{0,\operatorname{1WL}}\!\left(\mathcal{N}_{A\to B}\right) \coloneqq \liminf_{n\to \infty}\frac{1}{n}Q^{(1)}_{0,\operatorname{1WL}}\!\left(\mathcal{N}^{\otimes n}_{A\to B}\right).
\end{equation}
Similarly, the zero-error quantum capacity of the channel assisted by two-extendible superchannels can be defined as follows:
\begin{equation}\label{eq:2_ext_0_capacity}
    Q_{0,\operatorname{2-EXT}}\!\left(\mathcal{N}_{A\to B}\right) \coloneqq \liminf_{n\to \infty}\frac{1}{n}Q^{(1)}_{0,\operatorname{2-EXT}}\!\left(\mathcal{N}^{\otimes n}_{A\to B}\right).
\end{equation}

Alice and Bob can always choose a quantum superchannel $\Theta$ that acts separately on each instance of the quantum channel $\mathcal{N}_{A\to B}$ as a strategy to simulate an identity channel using multiple instances of the channel $\mathcal{N}_{A\to B}$ and one-way LOCC or two-extendible superchannels. Hence, we have the following inequalities
\begin{align}
    Q^{(1)}_{0,\operatorname{1WL}}\!\left(\mathcal{N}_{A\to B}\right) &\le Q_{0,\operatorname{1WL}}\!\left(\mathcal{N}_{A\to B}\right),\\
    Q^{(1)}_{0,\operatorname{2-EXT}}\!\left(\mathcal{N}_{A\to B}\right) &\le Q_{0,\operatorname{2-EXT}}\!\left(\mathcal{N}_{A\to B}\right).
\end{align}

The zero-error quantum capacity of a channel cannot be larger than the zero-error private capacity of the channel in the one-shot or asymptotic setting. This is because a $d$-dimensional ideal quantum channel can be used to send $d$-dimensional private data from Alice to Bob as follows: Alice can encode her classical data in a pure $d$-dimensional quantum state and send it to Bob such that it is protected from any eavesdropper by the no-cloning theorem. Bob can then perform a measurement on the received quantum state in a predetermined basis and deterministically decode $d$ bits of classical data encoded by Alice in the quantum state. Since we only need local operations to transform an ideal quantum channel to an ideal private channel, the following inequalities hold:
\begin{align}
    Q^{(1)}_{0,\operatorname{1WL}}\!\left(\mathcal{N}_{A\to B}\right) & \le P^{(1)}_{0,\operatorname{1WL}}\!\left(\mathcal{N}_{A\to B}\right),\\
    Q^{(1)}_{0,\operatorname{2-EXT}}\!\left(\mathcal{N}_{A\to B}\right) & \le P^{(1)}_{0,\operatorname{2-EXT}}\!\left(\mathcal{N}_{A\to B}\right),\label{eq:1shot_q_cap_2_ext_le_priv_cap}
\end{align}
where $P^{(1)}_{0,\operatorname{1WL}}\!\left(\mathcal{N}_{A\to B}\right)$ is the one-shot zero-error private capacity of a channel assisted by one-way LOCC superchannels, as defined in~\eqref{eq:private_cap_1W_def}, and $P^{(1)}_{0,\operatorname{2-EXT}}\!\left(\mathcal{N}_{A\to B}\right)$ is the one-shot zero-error private capacity of a channel assisted by two-extendible superchannels, as defined in~\eqref{eq:private_cap_2_ext_def}. 

\begin{proposition}
\label{prop:one-shot-zero-error-bound-min-unext}
    Consider an arbitrary zero-error protocol for quantum communication over a channel $\mathcal{N}_{A \to B}$ assisted by one-way LOCC or two-extendible superchannels, with $n$ the number of channel uses. Then the following upper bound holds for all $n\in \mathbb{N}$:
    \begin{equation}
        \frac{1}{n}Q^{(1)}_{0,\operatorname{1WL}}(\mathcal{N}_{A\to B}^{\otimes n}) \le \frac{1}{n}Q^{(1)}_{0,\operatorname{2-EXT}}(\mathcal{N}_{A\to B}^{\otimes n}) \le \widehat{E}^u_{\operatorname{min}}\!\left(\mathcal{N}_{A\to B}\right),
    \end{equation}
    where $\widehat{E}^u_{\operatorname{min}}\!\left(\mathcal{N}_{A\to B}\right)$ is defined in~\eqref{eq:min_unext_ent_ch_def}.
\end{proposition}

\begin{proof}
    The statement of the proposition follows simply by combining~\eqref{eq:1shot_q_cap_1W_le_2_ext},~\eqref{eq:1shot_q_cap_2_ext_le_priv_cap}, and Proposition~\ref{prop:priv_bits_le_min_geo_unext_ent}.
\end{proof}

\medskip

Since Proposition~\ref{prop:one-shot-zero-error-bound-min-unext} holds for all $n\in \mathbb{N}$, we conclude the following:

\begin{corollary}\label{cor:0_err_cap_ub}
The zero-error quantum capacity of a quantum channel $\mathcal{N}_{A\to B}$, assisted by one-way LOCC or two-extendible superchannels, is bounded from above by the min-geometric unextendible entanglement of $\mathcal{N}_{A\to B}$:
    \begin{equation}\label{eq:two_ext_0_capacity_le_unext_ent_ch}
        Q_{0,\operatorname{1WL}}(\mathcal{N}) \leq 
        Q_{0,\operatorname{2-EXT}}(\mathcal{N}) \le \widehat{E}^u_{\operatorname{min}}\!\left(\mathcal{N}\right).
    \end{equation}
\end{corollary}

\begin{proof}
    This is justified as mentioned above. See Appendix~\ref{app:0_err_cap_ub_proof} for an alternate proof.
\end{proof}

\begin{remark}
    It is known that forward classical assistance does not increase the zero-error quantum capacity of a quantum channel~\cite{Bennett_1996, BKN00} (see also~\cite[Lemmas~9.6-9.8]{khatri2020principles}). Thus, the quantum capacity of a quantum channel assisted by one-way LOCC superchannels is the same as the quantum capacity of the channel assisted by local operations.
\end{remark}

\begin{remark}
    The zero-error capacities of quantum channels are known to exhibit superactivation~\cite{CCHS10, Shirokov_2015}. That is, there exist quantum channels, say $\mathcal{N}^1$ and $\mathcal{N}^2$, such that the zero-error capacity of each channel individually is equal to zero, but the zero-error capacity of the tensor-product channel $\mathcal{N}^1\otimes \mathcal{N}^2$ is strictly positive. 

    Note that the subadditivity of the min-geometric unextendible entanglement implies that \mbox{$\widehat{E}^u_{\min}\!\left(\mathcal{N}^1\otimes\mathcal{N}^2\right) \le \widehat{E}^u_{\min}\!\left(\mathcal{N}^1\right) + \widehat{E}^u_{\min}\!\left(\mathcal{N}^2\right)$}, which is an upper bound on the zero-error quantum capacity and the zero-error private capacity of the tensor-product channel $\mathcal{N}^1\otimes\mathcal{N}^2$. Therefore, if a pair of channels is expected to exhibit superactivation of zero-error private capacity or zero-error quantum capacity, then at least one of the channels should have a strictly positive min-geometric unextendible entanglement.
\end{remark}

\section{Unextendible entanglement of bipartite quantum channels}\label{sec:unext_ent_bip}

In this section, we extend our developments on unextendibility to bipartite quantum channels. Bipartite quantum channels are generalizations of point-to-point channels in the sense that every point-to-point channel can be simulated using some bipartite channel by ignoring the input of Bob and the output of Alice; i.e., for every point-to-point quantum channel $\mathcal{N}_{A\to B'}$, there exists a bipartite quantum channel $\mathcal{M}_{AB\to A'B'}$ such that,
\begin{equation}
	\mathcal{N}_{A\to B'} = \operatorname{Tr}_{A'}\circ\mathcal{M}_{AB\to A'B'}\circ\operatorname{Tr}_{B}.
\end{equation}

Before discussing extensions of bipartite quantum channels, we should first establish what we mean by the marginal of a channel in the multipartite case. A quantum channel $\mathcal{N}_{AB_i\to A'B'_i}$ is a marginal of the channel $\mathcal{P}_{AB_{[k]}\to A'B'_{[k]}}$ if the following condition holds~\cite{Kaur_2019}:
\begin{equation}
    \operatorname{Tr}_{B'_{[k]\setminus i}}\circ\mathcal{P}_{AB_{[k]}\to A'B'_{[k]}} = \mathcal{N}_{AB_i\to A'B'_i}\otimes\operatorname{Tr}_{B_{[k]\setminus i}}.
\end{equation}
The Choi operators of the two channels are related as
\begin{equation}
     \operatorname{Tr}_{B'_{[k]\setminus i}}\!\left[\Gamma^{\mathcal{P}}_{AB_{[k]}A'B'_{[k]}}\right] = 
     \Gamma^{\mathcal{N}}_{AB_iA'B'_i}\otimes I_{B_{[k]\setminus i}}.
\end{equation}

Unlike the broadcast channels $\mathcal{P}_{A\to B_{[k]}}$ considered in previous sections of this work, not all multipartite channels have a well defined marginal. In fact, a quantum channel $\mathcal{P}_{AB_{[k]}\to A'B'_{[k]}}$ has a well defined marginal $\mathcal{N}_{AB_i\to A'B'_i}$ if and only if it does not allow systems $B_{[k]\setminus i}$ to send any data, quantum or classical, to systems $A'$ and $B'_i$. Moreover, if there exists a marginal channel $\mathcal{N}_{AB_i\to A'B'_i}$ for every $i \in [k]$, the quantum channel $\mathcal{P}_{AB_{[k]}\to A'B'_{[k]}}$ is non-signaling from $B_i$ to $A'$, for every $i\in [k]$.

\begin{figure}
    \centering
    \includegraphics[width = \linewidth]{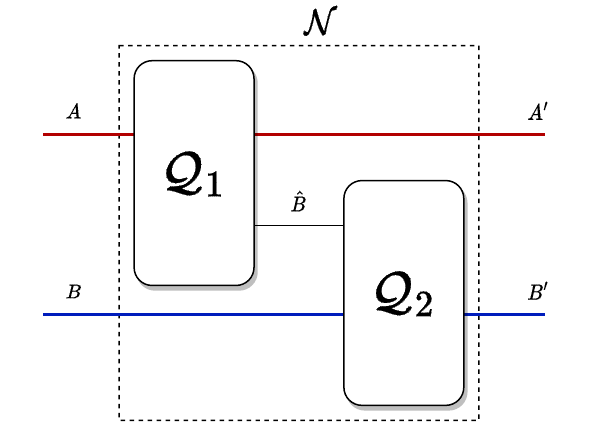}
    \caption{Decomposition of a semicausal channel between Alice and Bob that is nonsignaling from Bob to Alice.}
    \label{fig:semicausal_channel}
\end{figure}

We restrict our discussion on unextendibility of bipartite quantum channels to semicausal quantum channels between Alice and Bob~\cite{BGNP01} that are non-signaling from Bob to Alice. It has been shown that all semicausal channels are semi-localizable~\cite{Eggeling_2002}. As such, these quantum channels are of the following form (see Figure~\ref{fig:semicausal_channel}):
\begin{equation}
    \mathcal{N}_{AB\to A'B'} = \mathcal{Q}^2_{B\hat{B}\to B'}\circ\mathcal{Q}^1_{A\to A'\hat{B}},
\end{equation}
where $\mathcal{Q}^2_{B\hat{B}\to B'}$  and $\mathcal{Q}^1_{A\to A'\hat{B}}$ are quantum channels.
This ensures that there always exists an extension of such a channel that is of the form $\mathcal{P}_{AB_{[k]}\to A'B'_{[k]}}$, such that it has a well defined marginal $\mathcal{N}^i_{AB_i\to A'B'_i}$ for every $i\in [k]$. The marginal $\mathcal{N}^i_{AB_i\to A'B'_i}$ can be uniquely obtained from the channel $\mathcal{P}_{AB_{[k]}\to A'B'_{[k]}}$ using the following relation:
\begin{equation}
    \operatorname{Tr}_{B'_{[k]\setminus i}}\circ\mathcal{P}_{AB_{[k]}\to A'B'_{[k]}}\circ\mathcal{A}_{B_{[k]\setminus i}} = \mathcal{N}^i_{AB_i\to A'B'_i}, 
\end{equation}
where $\mathcal{A}_{B_{[k]\setminus i}}$ is a quantum channel that appends an arbitrary quantum state on the systems $B_{[k]\setminus i}$.

\subsection{Bipartite \texorpdfstring{$k$}{k}-extendible quantum channels}

Let us now define $k$-extendibility of bipartite quantum channels and superchannels. Multiple definitions of $k$-extendibility have been proposed for bipartite quantum channels~\cite{Kaur_2019, Berta_2021}. In this work, we use the notion of bipartite $k$-extendible channels presented in~\cite{Kaur_2019, Kaur_2021}.

\begin{definition}[Bipartite $k$-extendible channel]
A quantum channel $\mathcal{N}_{AB\to A'B'}$ is $k$-extendible if there exists a channel $\mathcal{P}_{AB_{[k]}\to A'B'_{[k]}}$ such that the following conditions hold.
\begin{enumerate}
    \item The extended channel is covariant under permutations of the $B$ systems,
    \begin{multline}\label{eq:perm_cov_bipartite}
        \mathcal{W}^{\pi}_{B'_{[k]}}\circ \mathcal{P}_{AB_{[k]}\to A'B'_{[k]}} \\= \mathcal{P}_{AB_{[k]}\to A'B'_{[k]}} \circ \mathcal{W}^{\pi}_{B_{[k]}} ~\forall \pi \in \mathcal{S}_k,
    \end{multline}
    where $\mathcal{W}^{\pi}$ is a permutation unitary  channel as defined previously, just after~\eqref{eq:perm-covariance-broadcast}. 
    \item The channel $\mathcal{N}_{AB\to A'B'}$ is a marginal of $\mathcal{P}_{AB_{[k]}\to A'B'_{[k]}}$:
    \begin{equation}\label{eq:no_signal_bipartite}
        \operatorname{Tr}_{B'_{[k]\setminus 1}}\circ \mathcal{P}_{AB_{[k]}\to A'B'_{[k]}} = \mathcal{N}_{AB \to A'B'} \otimes \operatorname{Tr}_{B_{[k]\setminus 1}}.
    \end{equation}
\end{enumerate}
\end{definition}

As a consequence of part 2.~of the above definition, it follows that the systems $B_2,B_3,\ldots,B_k$ cannot send any information to the systems $A$ and $B_1$.

The conditions in~\eqref{eq:perm_cov_bipartite} and~\eqref{eq:no_signal_bipartite} can be written as semidefinite constraints on the Choi operator of the quantum channel $\mathcal{P}_{AB_{[k]}\to A'B'_{[k]}}$ as follows~\cite[Eqs.~(24)--(27)]{Kaur_2021}:
\begin{align}
    \left(\mathcal{W}^{\pi}_{B_{[k]}}\otimes\mathcal{W}^{\pi}_{B'_{[k]}}\right)\Gamma^{\mathcal{P}} &= \Gamma^{\mathcal{P}} \label{eq:perm_cov_bip_choi} \quad \forall \pi \in \mathcal{S}_k,\\
    \operatorname{Tr}_{B'_{[k]\setminus 1}}\!\left[\Gamma^{\mathcal{P}}\right] &= \Gamma^{\mathcal{N}}_{AB_1A'B'_1}\otimes I_{B_{[k]\setminus 1}},\label{eq:no_signal_bip_choi}
\end{align}
where $\Gamma^{\mathcal{N}}_{ABA'B'}$ and $\Gamma^{\mathcal{P}} \equiv \Gamma^{\mathcal{P}}_{AB_{[k]}A'B'_{[k]}}$ are the Choi operators of the quantum channels $\mathcal{N}_{AB\to A'B'}$ and $\mathcal{P}_{AB_{[k]}\to A'B'_{[k]}}$, respectively. It is straightforward to verify that this definition of $k$-extendibility is consistent with the definition of $k$-extendibility of point-to-point channels.

It is also worth noting that the unique quantum channel corresponding to a point-to-point superchannel is a bipartite semi-localizable quantum channel. Since the definitions for bipartite $k$-extendible channels and point-to-point $k$-extendible superchannels are the same, we can assert that a point-to-point superchannel is $k$-extendible if and only if the unique bipartite quantum channel associated with it is a $k$-extendible channel.  

The set of bipartite $k$-extendible quantum channels is a relaxation of the set of bipartite one-way LOCC channels. This can be seen by considering a general bipartite one-way LOCC channel,
\begin{equation}
	\mathcal{N}_{AB\to A'B'} = \sum_x \mathcal{E}^x_{A\to A'}\otimes\mathcal{D}^x_{B\to B'},
\end{equation}
where $\left\{\mathcal{E}^x_{A\to A'}\right\}_x$ is a set of completely positive, trace non-increasing maps such that the sum map $\sum_x \mathcal{E}^x_{A\to A'}$ is a quantum channel and $\left\{\mathcal{D}^x_{A\to A'}\right\}_x$ is a set of quantum channels. This is understood to be a one-way LOCC channel as Alice applies the quantum instrument $\left\{\mathcal{E}^x_{A\to A'}\right\}_x$ on her system and sends the classical outcome $x$ to Bob, who then applies the quantum channel $\mathcal{D}^x_{B\to B'}$ on his system based on the classical data he received. We can construct an extension of this channel as follows:
\begin{equation}
	\mathcal{P}_{AB_{[k]}\to A'B'_{[k]}} \coloneqq \sum_x \mathcal{E}^x_{A\to A'}\otimes\mathcal{D}^x_{B_1\to B'_1}\otimes\cdots\otimes\mathcal{D}^x_{B_k\to B'_k},
\end{equation}
because $x$ is classical data that can be copied an arbitrary number of times. Such an extension obeys the permutation covariance and non-signaling conditions stated in~\eqref{eq:perm_cov_bipartite} and~\eqref{eq:no_signal_bipartite}. Thus, all one-way LOCC bipartite quantum channels are $k$-extendible for all $ k\ge 2$. 

\subsection{Bipartite \texorpdfstring{$k$}{k}-extendible superchannels}

In this section, we discuss the extendibility of bipartite superchannels, and in order to do so, we first establish the notion of marginal superchannels. Let $Q^{\Theta}_{CDA'B'\to C'D'AB}$ and $Q^{\Upsilon}_{CD_{[k]}A'B'_{[k]}\to C'D'_{[k]}AB_{[k]}}$ be the unique quantum channels corresponding to the superchannels $\Theta_{(AB\to A'B')\to (CD\to C'D')}$ and $\Upsilon_{(AB_{[k]}\to A'B'_{[k]})\to (CD_{[k]}\to C'D'_{[k]})}$, respectively. The superchannel $\Theta$ is said to be a marginal of the superchannel $\Upsilon$ if and only if the quantum channel $Q^{\Theta}_{CDA'B'\to C'D'AB}$ is a marginal of the channel $Q^{\Upsilon}_{CD_{[k]}A'B'_{[k]}\to C'D'_{[k]}AB_{[k]}}$; that is,
\begin{equation}
    \operatorname{Tr}_{D'_{[k]\setminus i}B_{[k]\setminus i}}\circ\mathcal{Q}^{\Upsilon} = \mathcal{Q}^{\Theta}\otimes \operatorname{Tr}_{D_{[k]\setminus i}B'_{[k]\setminus i}}.
\end{equation}

All bipartite superchannels do not have a well defined marginal. Similar to the case of bipartite channels, we restrict our discussion to such bipartite superchannels $\Theta_{(AB\to A'B')\to (CD\to C'D')}$ for which the associated quantum channel $\mathcal{Q}^{\Theta}_{CDA'B'\to C'D'AB}$ is semicausal and does not allow any data to be transmitted from the joint system $DB'$ to the joint system $C'A$. Now that we have defined the marginal of a bipartite superchannel, let us define $k$-extendible superchannels. 

\begin{definition}[Bipartite $k$-extendible superchannels]
	A bipartite superchannel $\Theta_{(AB\to A'B')\to (CD\to C'D')}$ is $k$-extendible if there exists a superchannel $\Upsilon_{(AB_{[k]}\to A'B'_{[k]})\to (CD_{[k]}\to C'D'_{[k]})}$ such that the following conditions hold for the unique quantum channels, $\mathcal{Q}^{\Theta}_{CDA'B'\to C'D'AB}$ and $\mathcal{Q}^{\Upsilon}_{CD_{[k]}A'B'_{[k]}\to C'D'_{[k]}AB_{[k]}}$, of the superchannels $\Theta$ and $\Upsilon$, respectively.
	\begin{enumerate}
		\item Permutation covariance:
		\begin{multline}\label{eq:two_ext_bip_supch_perm_cov}
			\mathcal{Q}^{\Upsilon}\circ\left(\mathcal{W}^{\pi}_{D_{[k]}}\otimes\mathcal{W}^{\pi}_{B'_{[k]}}\right) \\= \left(\mathcal{W}^{\pi}_{D'_{[k]}}\otimes\mathcal{W}^{\pi}_{B_{[k]}}\right)\circ \mathcal{Q}^{\Upsilon} \quad \forall \pi \in S_k.
		\end{multline}
		\item Non-signaling:
        \begin{equation}\label{eq:k_ext_bip_supch_non_sig}
            \operatorname{Tr}_{D'_{[k]\setminus 1}}\circ\mathcal{Q}^{\Upsilon} = \operatorname{Tr}_{D'_{[k]\setminus 1}}\circ\mathcal{Q}^{\Upsilon}\circ\mathcal{R}_{B'_{[k]\setminus 1}},
        \end{equation}
        where $\mathcal{R}$ is a quantum channel that traces out the input and replaces with an arbitrary quantum state.
        \item Marginality:
        \begin{equation}
			\operatorname{Tr}_{B_{[k]\setminus 1}D'_{[k]\setminus 1}}\circ \mathcal{Q}^{\Upsilon} 
			 = \mathcal{Q}^{\Theta}_{CDA'B'\to C'D'AB}\otimes\operatorname{Tr}_{B'_{[k]\setminus 1}D_{[k]\setminus 1}}.\label{eq:k_ext_bip_supch_marg_cond}
		\end{equation}
	\end{enumerate} 
\end{definition}

The conditions in~\eqref{eq:two_ext_bip_supch_perm_cov},~\eqref{eq:k_ext_bip_supch_non_sig}, and~\eqref{eq:k_ext_bip_supch_marg_cond} can be written as semidefinite constraints on the Choi operators of the superchannels $\Theta$ and $\Upsilon$ as
\begin{equation}
    \left(\mathcal{W}^{\pi}_{D'_{[k]}}\otimes \mathcal{W}^{\pi}_{B_{[k]}}\otimes \mathcal{W}^{\pi}_{D_{[k]}}\otimes \mathcal{W}^{\pi}_{B'_{[k]}}\right)\left(\Gamma^{\Upsilon}\right) = \Gamma^{\Upsilon} \quad \forall \pi \in S_k,\label{eq:perm_cov_bip_supch_choi}
\end{equation}
\vspace{-.5cm}
\begin{align}
    \operatorname{Tr}_{D'_{[k]\setminus 1}}\!\left[\Gamma^{\Upsilon}\right] &= \operatorname{Tr}_{D'_{[k]\setminus 1}B'_{[k]\setminus 1}}\!\left[\Gamma^{\Upsilon}\right]\otimes \frac{I_{B'_{[k]\setminus 1}}}{\left|B'\right|^{k-1}},\label{eq:k_ext_bip_supch_non_sig_choi}\\
    \operatorname{Tr}_{B_{[k]\setminus 1}D'_{[k]\setminus 1}}\!\left[\Gamma^{\Upsilon}\right] &= \Gamma^{\Theta}\otimes I_{B'_{[k]\setminus 1}D_{[k]\setminus 1}},\label{eq:k_ext_bip_supch_marg_choi}
\end{align}
where $\left|B'\right|$ is the dimension of the system $B'$ and $\Gamma^{\Theta}$ and $\Gamma^{\Upsilon}$ are the Choi operators of the quantum channels $\mathcal{Q}^{\Theta}$ and $\mathcal{Q}^{\Upsilon}$, respectively.

\begin{proposition}\label{prop:bip_such_ch_marg_theo}
    Let $\Theta_{(AB\to A'B')\to (CD\to C'D')}$ be a $k$-extendible superchannel with the $k$-extension $\Upsilon_{(AB_{[k]}\to A'B'_{[k]})\to (CD_{[k]}\to C'D'_{[k]})}$. Then each marginal of the channel obtained by acting with the superchannel $\Upsilon$ on a channel $\mathcal{P}_{AB_{[k]}\to A'B'_{[k]}}$ is the same as the channel obtained by acting with the superchannel $\Theta$ on the respective marginal of the channel $\mathcal{P}_{AB_{[k]}\to A'B'_{[k]}}$. That is,
    \begin{multline}\label{eq:bip_supch_ch_marg_theo}
        \operatorname{Tr}_{D'_{[k]\setminus i}}\circ\left(\Upsilon\left(\mathcal{P}\right)\right) = \\
         \Theta\left(\operatorname{Tr}_{B'_{[k]\setminus i}}\circ\mathcal{P}\circ\mathcal{A}_{B_{[k]\setminus i}}\right)\otimes\operatorname{Tr}_{D_{[k]\setminus i}} \quad \forall i\in [k].
    \end{multline}
\end{proposition}
\begin{proof}
    See Appendix~\ref{app:bip_supch_ch_marg_theo_proof}.    
\end{proof}

\begin{proposition}\label{prop:bip_supch_preserve_ext}
    The channel formed by applying a $k$-extendible superchannel on a $k$-extendible quantum channel is also a $k$-extendible quantum channel.
\end{proposition}

\begin{proof}
    See Appendix~\ref{app:bip_supch_preserve_ext}.
\end{proof}
    
\begin{proposition}\label{prop:1WL_supch_k_ext}
	All bipartite superchannels that can be realized by local operations and one-way classical communication are $k$-extendible for all $ k\ge 2$.
\end{proposition}

\begin{proof}
    See Appendix~\ref{app:1WL_supch_k_ext}.
\end{proof}

\subsection{Unextendible entanglement of bipartite quantum channels}

With the notion of $k$-extendibility established for bipartite quantum channels and superchannels, we can now look for a measure to quantify the unextendibility of a bipartite quantum channel. Let us first define the following set of extensions of a bipartite channel $\mathcal{N}_{AB\to A'B'}$:
\begin{equation}\label{eq:ext_set_bip_ch}
    \operatorname{Ext}\!\left(\mathcal{N}\right) \coloneqq \left\{
    \begin{array}{c}
        \mathcal{P}_{AB_1B_2\to A'B'_1B'_2}:\\
        \mathcal{N}_{AB\to A'B'} = \operatorname{Tr}_{B'_2}\circ\mathcal{P}_{AB_1B_2\to A'B'_1B'_2}\circ\mathcal{A}_{B_2},\\
        \operatorname{Tr}_{B'_1}\circ\mathcal{P} = \operatorname{Tr}_{B'_1}\circ\mathcal{P}\circ\mathcal{R}_{B_1}
    \end{array}
    \right\},
\end{equation}
where $\mathcal{A}$ is a  channel that appends an arbitrary state and $\mathcal{R}$ is a  channel that traces out the input and replaces it with an arbitrary  state. The generalized unextendible entanglement of a bipartite  channel is defined as follows.
\begin{definition}
The generalized unextendible entanglement of a bipartite quantum channel is defined for a generalized divergence between quantum channels, $\mathbf{D}$, as follows:
\begin{multline}\label{eq:gen_unext_ent_bip_def}
    \mathbf{E}^u\!\left(\mathcal{N}_{AB\to A'B'}\right) 
    \\= \frac{1}{2}\inf_{\mathcal{P}_{AB_1B_2\to A'B'_1B'_2} \in \operatorname{Ext}\left(\mathcal{N}\right)}\Big\{\mathbf{D}\!\left(\mathcal{N}_{AB\to A'B'}\Vert\mathcal{M}_{AB\to A'B'}\right):\\
     \mathcal{N}_{AB\to A'B'} = \operatorname{Tr}_{B'_2}\circ\mathcal{P}_{AB_1B_2\to A'B'_1B'_2}\circ\mathcal{A}_{B_2}, \\
    \mathcal{M}_{AB\to A'B'} = \operatorname{Tr}_{B'_1}\circ\mathcal{P}_{AB_1B_2\to A'B'_1B'_2}\circ\mathcal{A}_{B_1}\Big\},
\end{multline}
\end{definition}

As is evident from the definition, the minimum value of the generalized unextendible entanglement is achieved for a two-extendible quantum channel. If the underlying divergence is strongly faithful, then the unextendible entanglement of a bipartite quantum channel is equal to zero if and only if the channel is two-extendible. We also find that the action of a two-extendible superchannel on a bipartite quantum channel cannot increase the unextendible entanglement of the said channel. 

\begin{theorem}[Monotonicity]\label{theo:two_ext_monotonic_bip}
    The generalized unextendible entanglement of a bipartite quantum channel does not increase under the action of a two-extendible superchannel. That is, for an arbitrary bipartite quantum channel $\mathcal{N}_{AB\to A'B'}$ and a two-extendible superchannel $\Theta_{(AB\to A'B')\to (CD\to C'D')}$,
    \begin{equation}\label{eq:unext_ent_ch_ge_unext_ent_superch_bip}
         \mathbf{E}^u\!\left(\mathcal{N}_{AB \to A'B'}\right) \ge \mathbf{E}^u\!\left(\Theta\left(\mathcal{N}_{AB \to A'B'}\right)\right).
     \end{equation}
\end{theorem}

\begin{proof}
	Monotonicity of generalized unextendible entanglement of bipartite channels under two-extendible superchannels follows from similar arguments as given in the proof of Theorem~\ref{theo:two_ext_monotonic_p2p}. See Appendix~\ref{app:two_ext_monotonic_bip} for a complete proof.
\end{proof}

\medskip

Proposition~\ref{prop:bip_such_ch_marg_theo} allows us to generalize the statement of Theorem~\ref{theo:unext_ent_state_le_unext_ent_ch_gen} to bipartite semicausal channels, which we state as Theorem~\ref{theo:unext_ent_state_le_unext_ent_ch_bip}.

\begin{theorem}\label{theo:unext_ent_state_le_unext_ent_ch_bip}
    The unextendible entanglement of a quantum state $\sigma_{R_CC'C''R_DD'D''}$, with respect to the partition $R_CC'C'':R_DD'D''$, that can be established between two parties using a bipartite semicausal channel $\mathcal{N}_{AB\to A'B'}$, a bipartite two-extendible superchannel $\Theta_{(AB\to A'B')\to (CD\to C'C''D'D'')}$, and any two-extendible state $\rho_{R_CC'R_DD'}$, is not greater than the unextendible entanglement of the quantum channel $\mathcal{N}_{AB\to A'B'}$; that is, 
    \begin{equation}\label{eq:unext_ent_state_le_unext_ent_ch_bip}
        \sup_{\rho\in \operatorname{2-EXT}\!\left(R_CC:R_DD\right)}\mathbf{E}^u\!\left(\sigma_{R_CC'C'':R_DD'D''}\right) \le  \mathbf{E}^u\!\left(\mathcal{N}_{AB\to A'B'}\right), 
    \end{equation}
    where $\operatorname{2-EXT}\!\left(S_A:S_B\right)$ is the set of two-extendible states with respect to systems $S_A$ and $S_B$, and
    \begin{equation}
        \sigma_{R_CC'C''R_DD'D''} \coloneqq \left(\Theta\!\left(\mathcal{N}_{AB\to A'B'}\right)\right)\left(\rho_{R_CCR_DD}\right).
    \end{equation}
\end{theorem}
\begin{proof}
    The proof follows the same line of reasoning as the proof of Theorem~\ref{theo:unext_ent_state_le_unext_ent_ch_gen}. We present a complete proof in Appendix~\ref{app:unext_ent_state_le_unext_ent_ch_bip}.
\end{proof}

\subsubsection{\texorpdfstring{$\alpha$}{Alpha}-geometric unextendible entanglement of bipartite quantum channels}

We have seen that the $\alpha$-geometric R\'enyi relative entropies offer several desirable properties when used as the underlying divergence for defining unextendible entanglement of point-to-point quantum channels. Here we define the $\alpha$-geometric unextendible entanglement of a bipartite quantum channel $\mathcal{N}_{AB\to A'B'}$ as
\begin{multline}\label{eq:geo_unext_ent_bip_def}
    \widehat{E}^u_{\alpha}\!\left(\mathcal{N}_{AB\to A'B'}\right) \coloneqq\\
    \frac{1}{2}\inf_{\mathcal{P} \in \operatorname{Ext}\left(\mathcal{N}\right)}\sup_{\psi_{RAB}}\left\{
    \begin{array}{c}
    \widehat{D}_{\alpha}\!\left(\mathcal{N}\!\left(\psi_{RAB}\right)\Vert\mathcal{M}\!\left(\psi_{RAB}\right)\right):\\
     \mathcal{N} = \operatorname{Tr}_{B'_2}\circ\mathcal{P}_{AB_1B_2\to A'B'_1B'_2}\circ\mathcal{A}_{B_2}, \\
    \mathcal{M} = \operatorname{Tr}_{B'_1}\circ\mathcal{P}_{AB_1B_2\to A'B'_1B'_2}\circ\mathcal{A}_{B_1}
    \end{array}\right\}, 
\end{multline}
for all $\alpha \in (0,1)\cup(1,2]$.

Several properties of the $\alpha$-geometric unextendible entanglement of point-to-point channels hold for the $\alpha$-geometric unextendible entanglement of bipartite quantum channels as well. Since the input and output systems of  point-to-point channels can be considered to have multiple subsystems, the following properties trivially extend from the $\alpha$-geometric unextendible entanglement of point-to-point quantum channels to the $\alpha$-geometric unextendible entanglement of bipartite quantum channels.
\begin{enumerate}
    \item The $\alpha$-geometric unextendible entanglement of a bipartite channel increases monotonically with $\alpha$.
    \item The smallest quantity in the family of $\alpha$-geometric unextendible entanglement is induced by the $\alpha$-geometric R\'enyi relative entropy when $\alpha\to 0$, and is called the min-geometric unextendible entanglement of the channel,
    \begin{align}
        &\widehat{E}^u_{\min}\!\left(\mathcal{N}_{AB\to A'B'}\right) \notag \\&\coloneqq \lim_{\alpha\to 0^+}\widehat{E}^u_{\alpha}\!\left(\mathcal{N}_{AB\to A'B'}\right)\\
        &= \frac{1}{2}\inf_{\mathcal{P}\in\operatorname{Ext}\left(\mathcal{N}\right)}\lim_{\alpha\to 0^+}\widehat{D}_{\alpha}\!\left(\mathcal{N}\Vert\operatorname{Tr}_{B'_1}\circ\mathcal{P}\circ\mathcal{A}_{B_1}\right)\label{eq:min_geo_unext_ent_bip_def}.
    \end{align}
    \item The $\alpha$-geometric unextendible entanglement converges to the unextendible entanglement induced by the Belavkin--Staszewski relative entropy as $\alpha\to 1$,
    \begin{align}
        &\widehat{E}^u\!\left(\mathcal{N}_{AB\to A'B'}\right) \notag \\&\coloneqq \lim_{\alpha\to 1}\widehat{E}^u_{\alpha}\!\left(\mathcal{N}_{AB\to A'B'}\right)\\
        &= \frac{1}{2}\inf_{\mathcal{P}\in\operatorname{Ext}\left(\mathcal{N}\right)}\lim_{\alpha\to 1}\widehat{D}_{\alpha}\!\left(\mathcal{N}\Vert\operatorname{Tr}_{B'_1}\circ\mathcal{P}\circ\mathcal{A}_{B_1}\right)\label{eq:BS_geo_unext_ent_bip_def}.
    \end{align}
    \item The $\alpha$-geometric unextendible entanglement is subadditive under tensor products of bipartite quantum channels; that is,
    \begin{equation}
        \widehat{E}^u_{\alpha}\!\left(\mathcal{N}\otimes \mathcal{M}\right) \le \widehat{E}^u_{\alpha}\!\left(\mathcal{N}\right) + \widehat{E}^u_{\alpha}\!\left(\mathcal{M}\right).
    \end{equation}
\end{enumerate}

One would expect an unextendible bipartite  channel to have the ability to boost the unextendibility of a bipartite  state. The $\alpha$-geometric R\'enyi relative entropy follows the chain rule $\forall \alpha \in (0,1)\cup(1,2]$~\cite[Proposition 45]{Katariya2021},
\begin{equation}\label{eq:Chain_rule_geo_ent}
    \widehat{D}_{\alpha}\!\left(\mathcal{N}\!\left(\rho\right)\Vert\mathcal{M}\!\left(\sigma\right)\right) \le \widehat{D}_{\alpha}\!\left(\mathcal{N}\Vert\mathcal{M}\right) + \widehat{D}_{\alpha}\!\left(\rho\Vert\sigma\right),
\end{equation}
which allows us to quantify the effect of an unextendible bipartite  channel on the unextendibility of a bipartite  state, as seen from the theorem below.

\begin{theorem}\label{theo:unext_ent_out_le_input}
    A bipartite quantum channel $\mathcal{N}_{AB\to A'B'}$ cannot increase the $\alpha$-geometric unextendible entanglement of a bipartite state $\rho_{AB}$ by more than the $\alpha$-geometric unextendible entanglement of the channel itself; that is, for all $ \alpha \in (0,1)\cup(1,2]$,
    \begin{equation}\label{eq:unext_ent_out_le_input}
        \widehat{E}^u_{\alpha}\!\left(\mathcal{N}_{AB\to A'B'}\!\left(\rho_{AB}\right)\right) \le \widehat{E}^u_{\alpha}\!\left(\rho_{AB}\right) + \widehat{E}^u_{\alpha}\!\left(\mathcal{N}_{AB\to A'B'}\right).
    \end{equation}
\end{theorem}
\begin{proof}
    Consider a bipartite quantum channel $\mathcal{N}_{AB\to A'B'}$ and a quantum state $\rho_{AB}$. Let $\mathcal{P}_{AB_1B_2\to A'B'_1B'_2}$ be an extension of the channel $\mathcal{N}_{AB\to A'B'}$ in the set $\operatorname{Ext}\left(\mathcal{N}\right)$ defined in~\eqref{eq:ext_set_bip_ch}, and let $\sigma_{AB_1B_2}$ be an arbitrary extension of the quantum state $\rho_{AB}$. By the definition of $\alpha$-geometric unextendible entanglement of states,
    \begin{align}
        &\widehat{E}^u_{\alpha}\!\left(\mathcal{N}_{AB\to A'B'}\!\left(\rho_{AB}\right)\right)\notag \\
        &\le \frac{1}{2}\widehat{D}_{\alpha}\!\left(\mathcal{N}\!\left(\rho_{AB}\right)\Vert \operatorname{Tr}_{B'_1}\circ\mathcal{P}\!\left(\sigma\right)\right)\\
        &= \frac{1}{2}\widehat{D}_{\alpha}\!\left(\mathcal{N}\!\left(\rho_{AB}\right)\Vert \operatorname{Tr}_{B'_1}\circ\mathcal{P}\circ\mathcal{A}_{B_1}\!\left(\operatorname{Tr}_{B_1}\!\left[\sigma\right]\right)\right)\\
        &\le \frac{1}{2}\left\{\widehat{D}_{\alpha}\!\left(\rho\middle\Vert\operatorname{Tr}_{B_1}\!\left[\sigma\right]\right) + \widehat{D}_{\alpha}\!\left(\mathcal{N}\middle\Vert\operatorname{Tr}_{B'_1}\circ\mathcal{P}\circ\mathcal{A}_{B_1}\right)\right\},\label{eq:unext_ent_le_div_state_ch}
    \end{align}
    where the equality follows from the definition of $\operatorname{Ext}\!\left(\mathcal{N}\right)$ in~\eqref{eq:ext_set_bip_ch} and the second inequality follows from the chain rule of $\alpha$-geometric R\'enyi relative entropy given in~\eqref{eq:Chain_rule_geo_ent}. Since the inequality in~\eqref{eq:unext_ent_le_div_state_ch} holds for every quantum channel $\mathcal{P}_{AB_1B_2\to A'B'_1B'_2}$ in $\operatorname{Ext}\!\left(\mathcal{N}\right)$ and every extension $\sigma_{AB_1B_2}$ of the state $\rho_{AB}$, we can take an infimum over all such channels $\mathcal{P}_{AB_1B_2\to A'B'_1B'_2}$ and states $\sigma_{AB_1B_2}$, and conclude~\eqref{eq:unext_ent_out_le_input}.
\end{proof}

\section{Applications of the unextendible entanglement of bipartite quantum channels}\label{sec:applications_bip}

In this section, we discuss some applications of the unextendible entanglement of bipartite quantum channels. Bipartite quantum channels can be used to increase the entanglement in a shared bipartite quantum state, hence increasing the distillable entanglement and distillable key of the state. In Section~\ref{sec:dist_ent_st_ch_pair}, we give an upper bound on the expected rate of distilling ebits probabilistically from a bipartite quantum state using a bipartite quantum channel assisted by one-way LOCC or two-extendible superchannels, and in Section~\ref{sec:dist_sec_key_st_ch_pair}, we give an upper bound on the rate of distilling exact secret bits probabilistically from a bipartite quantum state using a bipartite quantum channel assisted by one-way LOCC or two-extendible superchannels.

\subsection{One-way distillable entanglement of a quantum state-channel pair}\label{sec:dist_ent_st_ch_pair}

The distillable entanglement of a bipartite quantum state under a restricted set of operations has been a subject of interest in quantum information theory. In quantum communication theory, the number of ebits that can be distilled from an existing bipartite state under local operations and one-way classical communication as well as two-way classical communication is of special interest due to state-of-the-art classical networks available to us. 

\begin{figure}
    \centering
    \includegraphics[width = \linewidth]{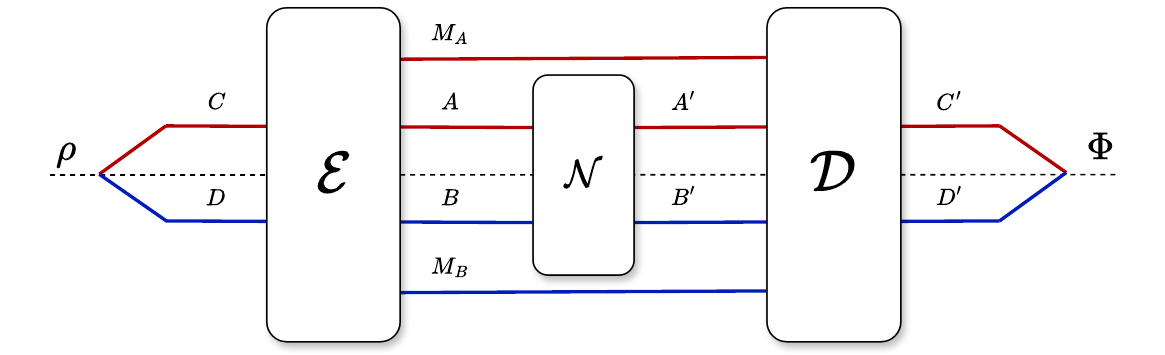}
    \caption{Protocol to distill a maximally entangled state $\Phi_{C'D'}$ from a bipartite quantum state $\rho_{CD}$, a bipartite channel $\mathcal{N}_{AB\to A'B'}$, and one-way LOCC pre-processing and post-processing channels $\mathcal{E}_{CD\to AM_ABM_B}$ and $\mathcal{D}_{A'M_AB'M_B\to C'D'}$, respectively. The dotted line represents the separation between Alice and Bob's labs. All systems above the dotted line are held by Alice and all systems below the dotted line are held by Bob.}
    \label{fig:Ent_dist_bip}
\end{figure}

An upper bound on the probabilistic one-way distillable entanglement of a bipartite quantum state has been established in~\cite{WWW19}. In this section, we consider a more general setting where the two distant parties, Alice and Bob, hold a bipartite quantum state $\rho_{CD}$, and also have access to a semicausal bipartite quantum channel $\mathcal{N}_{AB\to A'B'}$ that is not necessarily simulable by local operations and one-way classical communication. Alice and Bob can encode the available quantum state $\rho_{CD}$ using a one-way LOCC channel $\mathcal{E}_{CD\to AB}$, and then transform the encoded state using the quantum channel $\mathcal{N}_{AB\to A'B'}$ followed by a one-way LOCC decoding channel $\mathcal{D}_{A'B'\to C'D'}$ to distill a maximally entangled state (see Figure~\ref{fig:Ent_dist_bip}). As such, the channel $\mathcal{N}_{AB\to A'B'}$ can be used to boost the amount of entanglement that can be distilled probabilistically, or deterministically, from the shared bipartite state using one-way LOCC.

A simple example of a semicausal channel that is useful for probabilistic distillation of entanglement is the following channel: 
\begin{multline}\label{eq:eras_ch_full_ext}
    \mathcal{N}_{A\to A'B}\!\left(\rho_{RA}\right) = p\rho_{RA'}\otimes |e\rangle\!\langle e|_B \\
    + (1-p)\rho_{RB}\otimes |e\rangle\!\langle e|_{A'},
\end{multline}
where $|e\rangle\!\langle e|_{S}$ is the erasure symbol, which is orthogonal to every state in the Hilbert space of system $S$. This is a channel where Alice and Bob share an erasure channel with erasure probability $p$, and the output of the complementary channel is received by Alice herself; as such, nothing is lost to the environment. Alice can send one share of a maximally entangled state to Bob using the channel $\mathcal{N}_{A\to A'B}$, and then she can measure if the state she received back was erased or not by applying the POVM $\{\Pi_{A'},|e\rangle\!\langle e|_{A'}\}$, where $\Pi_{A'}$ is the projection onto the entire Hilbert space of system $A'$. If Alice measures her system to be erased, she knows that a maximally entangled state has been established between her and Bob, hence, distilling a maximally entangled state with probability $1-p$.

\subsubsection{Probabilistic one-way distillable entanglement}

Let us first consider a probabilistic setting in which an entanglement distillation protocol distills a maximally entangled state of Schmidt rank $d$ with some probability $p$. The expected number of ebits distilled by this protocol is given by $p\log_2 d$. Alice and Bob use one instance of the quantum channel $\mathcal{N}_{AB\to A'B'}$ and a one-way LOCC superchannel $\Theta_{(AB\to A'B')\to (CD\to C'D')}$ to distill a maximally entangled state of Schmidt rank $d$ from a bipartite quantum state $\rho_{CD}$ with probability $p$. The action of this protocol can be mathematically described as follows:
\begin{multline}
     \left(\Theta\left(\mathcal{N}\right)\right)\left(\rho\right) \\= p|1\rangle\!\langle 1|_X \otimes \Phi^d_{C'D'}  + (1-p)|0\rangle\!\langle 0|_X \otimes \sigma_{C'D'} ,
 \end{multline}
where system $X$ is held by Alice. We define the probabilistic one-way distillable entanglement of the quantum state-channel pair $\left(\rho_{CD},\mathcal{N}_{AB\to A'B'}\right)$ as follows:
\begin{multline}
    \widetilde{E}_{D,\operatorname{1WL}}\!\left(\rho_{CD},\mathcal{N}_{AB\to A'B'}\right) \\
    \coloneqq \sup_{\substack{p \in [0,1],d\in \mathbb{N},\\\Theta\in \operatorname{1WL}}}\left\{
    \begin{array}{c}
        p\log_2 d:\\
         \left(\Theta\!\left(\mathcal{N}\right)\right)\!\left(\rho\right) = p|1\rangle\!\langle 1|_X \otimes \Phi^d_{C'D'} \\
         + (1-p)|0\rangle\!\langle 0|_X \otimes \sigma_{C'D'} ,\\
        \sigma\in \mathcal{S}(C'D')
    \end{array}
    \right\}.
\end{multline}

We can also define the probabilistic two-extendible distillable entanglement of a quantum state-channel pair by relaxing the set of allowed superchannels to the set of two-extendible superchannels:
\begin{multline}
    \widetilde{E}_{D,\operatorname{2-EXT}}\!\left(\rho_{CD},\mathcal{N}_{AB\to A'B'}\right) \\
    \coloneqq \sup_{\substack{p \in [0,1],d\in \mathbb{N},\\\Theta\in \operatorname{2-EXT}}}\left\{
    \begin{array}{c}
         p\log_2 d:\\
         \left(\Theta\!\left(\mathcal{N}\right)\right)\!\left(\rho\right) = p|1\rangle\!\langle 1|_X \otimes \Phi^d_{C'D'} \\
         + (1-p)|0\rangle\!\langle 0|_X \otimes \sigma_{C'D'} ,\\
         \sigma\in \mathcal{S}(C'D')
    \end{array}
    \right\}.
\end{multline}
Since the set of one-way LOCC superchannels lies inside the set of two-extendible superchannels, the following inequality holds for every bipartite state $\rho_{CD}$ and bipartite channel $\mathcal{N}_{AB\to A'B'}$:
\begin{equation}
    \widetilde{E}_{D,\operatorname{2-EXT}}\!\left(\rho_{CD},\mathcal{N}_{AB\to A'B'}\right) \ge \widetilde{E}_{D,\operatorname{1WL}}\!\left(\rho_{CD},\mathcal{N}_{AB\to A'B'}\right).
\end{equation}

\begin{proposition}\label{theo:distill_ent_ch_st_ub}
    The expected rate at which ebits can be probabilistically distilled from a bipartite quantum state $\rho_{CD}$ and $n$ instances of a quantum channel $\mathcal{N}_{AB\to A'B'}$, assisted by one-way LOCC superchannels or two-extendible superchannels, is bounded from above as follows:
    \begin{multline}\label{eq:distill_ent_ch_st_ub}
        \frac{1}{n}\widetilde{E}_{D,\operatorname{1WL}}\!\left(\rho,\mathcal{N}^{\otimes n}\right) \le \frac{1}{n}\widetilde{E}_{D,\operatorname{2-EXT}}\!\left(\rho,\mathcal{N}^{\otimes n}\right) \\ \le \frac{1}{n}\widehat{E}^u\!\left(\rho_{CD}\right) + \widehat{E}^u\!\left(\mathcal{N}_{AB\to A'B'}\right),
    \end{multline}
    where $\widehat{E}^u\!\left(\rho_{CD}\right)$ is the unextendible entanglement of the quantum state $\rho_{CD}$ induced by the Belavkin--Staszewski relative entropy (defined in~\eqref{eq:Bel_Stas_induce_unext_ent_st}), and $\widehat{E}^u\!\left(\mathcal{N}_{AB\to A'B'}\right)$ is the unextendible entanglement of the quantum channel $\mathcal{N}_{AB\to A'B'}$ induced by the Belavkin--Staszewski relative entropy (defined in~\eqref{eq:BS_geo_unext_ent_bip_def}).
\end{proposition}

\begin{proof}
    Let $\rho_{CD}$ be a quantum state shared between Alice (holding system $C$) and Bob (holding system $D$). Let $\Theta_{(A^nB^n\to A'^nB'^n)\to (CD\to C'D')}$ be a two-extendible superchannel such that it acts on $n$ instances of a quantum channel $\mathcal{N}_{AB\to A'B'}$, and the resultant channel consumes the quantum state $\rho_{AB}$ to generate a maximally entangled state of Schmidt rank $d$ with probability $p$. This process can be mathematically described as,
    \begin{multline}\label{eq:prob_distill_ent_ch_st}
        \left(\Theta\!\left(\mathcal{N}^{\otimes n}\right)\right)\!\left(\rho_{CD}\right) =\\ p|1\rangle\!\langle 1|_X \otimes\Phi^d_{C'D'}  + (1-p) |0\rangle\!\langle 0|_X \otimes \sigma_{C'D'} ,
    \end{multline}
    where $\sigma_{C'D'}$ is an arbitrary bipartite state and system $X$ is held by Alice.

    Recall that the $\alpha$-geometric uenxtendible entanglement of the maximally entangled state $\Phi^d_{C'D'}$ is equal to $\log_2 d$ as mentioned in Theorem~\ref{theo:props_from_WWW19}. The direct-sum property in Proposition~\ref{prop:geo_unext_direct_sum} then implies that the $\alpha$-geometric unextendible entanglement of the quantum state described in~\eqref{eq:prob_distill_ent_ch_st} is no less than $p\log_2 d$ for all $\alpha\in (1,2]$. Therefore,
    \begin{align}
        p\log_2 d &\le \widehat{E}^u_{\alpha}\!\left(\Theta\!\left(\mathcal{N}^{\otimes n}\right)\!\left(\rho\right)\right)\\
        &\le \widehat{E}^u_{\alpha}\!\left(\Theta\!\left(\mathcal{N}^{\otimes n}\right)\right) + \widehat{E}^u_{\alpha}\!\left(\rho\right)\\
        &\le n\widehat{E}^u_{\alpha}\!\left(\mathcal{N}_{AB\to A'B'}\right) + \widehat{E}^u_{\alpha}\!\left(\rho_{CD}\right),
    \end{align}
    where the second inequality follows from Theorem~\ref{theo:unext_ent_out_le_input}, and the final inequality follows from the monotonicity of unextendible entanglement under the action of two-extendible superchannels (Theorem~\ref{theo:two_ext_monotonic_bip}) and subadditivity of $\alpha$-geometric unextendible entanglement under tensor product of quantum channels (Proposition~\ref{prop:unext_ent_subadditive}). We can take the limit $\alpha\to 1$ to get the tightest upper bound using this technique.
    Since the above inequality is true for all values of $p$, every dimension $d$, and every two-extendible superchannel $\Theta$, we can take a supremum over all of these quantities and conclude~\eqref{eq:distill_ent_ch_st_ub}.
\end{proof}

\subsubsection{Exact one-way distillable entanglement}

We can consider a special case of the probabilistic one-way distillable entanglement by demanding that the maximally entangled state is distilled deterministically. This means that Alice and Bob use a quantum state $\rho_{CD}$ and $n$ instances of a quantum channel $\mathcal{N}_{AB\to A'B'}$ along with local operations and unbounded forward classical communication from Alice to Bob to deterministically distill the maximum number of ebits possible. We define the exact one-way distillable entanglement of a state-channel pair $\left(\rho_{CD},\mathcal{N}_{AB\to A'B'}\right)$ as follows:
\begin{multline}
    \widetilde{E}^{e}_{D,\operatorname{1WL}}\!\left(\rho_{CD},\mathcal{N}_{AB\to A'B'}\right) \\
    \coloneqq \sup_{d\in \mathbb{N}, \Theta\in \operatorname{1WL}}\left\{
    \begin{array}{c}
        \log_2 d:\\
         \left(\Theta\!\left(\mathcal{N}\right)\right)\!\left(\rho\right) = \Phi^d_{C'D'}
    \end{array}
    \right\}.
\end{multline}

We can relax the set of allowed superchannels to the set of two-extendible superchannels and define the exact two-extendible distillable entanglement of a state-channel pair as follows:
\begin{multline}
    \widetilde{E}^{e}_{D,\operatorname{2-EXT}}\!\left(\rho_{CD},\mathcal{N}_{AB\to A'B'}\right) \\
    \coloneqq \sup_{d\in \mathbb{N}, \Theta\in \operatorname{2-EXT}}\left\{
    \begin{array}{c}
        \log_2 d:\\
         \left(\Theta\!\left(\mathcal{N}\right)\right)\!\left(\rho\right) = \Phi^d_{C'D'}
    \end{array}
    \right\}.
\end{multline}
\begin{proposition}\label{theo:distill_ent_ch_st_ub_ex}
    The rate at which ebits can be deterministically distilled from a bipartite quantum state $\rho_{CD}$ and $n$ instances of a quantum channel $\mathcal{N}_{AB\to A'B'}$, assisted by one-way LOCC or two-extendible superchannels, is bounded from above by the following quantity:
    \begin{align}
        &\frac{1}{n}\widetilde{E}^{e}_{D,\operatorname{1WL}}\!\left(\rho,\mathcal{N}^{\otimes n}\right)\notag\\
        &\le \frac{1}{n}\widetilde{E}^{e}_{D,\operatorname{2-EXT}}\!\left(\rho,\mathcal{N}^{\otimes n}\right) \\ 
        &\le \frac{1}{n}\widehat{E}^u_{\min}\!\left(\rho_{CD}\right) + \widehat{E}^u_{\min}\!\left(\mathcal{N}_{AB\to A'B'}\right),\label{eq:distill_ent_ch_st_ub_ex}
    \end{align}
    where $\widehat{E}^u_{\min}\!\left(\rho_{CD}\right)$ is the min-unextendible entanglement of the quantum state $\rho_{CD}$ (defined in~\eqref{eq:min_geo_unext_ent_st_def}) and $\widehat{E}^u_{\min}\!\left(\mathcal{N}_{AB\to A'B'}\right)$ is the min-geometric unextendible entanglement of the quantum channel $\mathcal{N}_{AB\to A'B'}$ (defined in~\eqref{eq:min_geo_unext_ent_bip_def}).
\end{proposition}

\begin{proof}
    The proof follows from the same line of reasoning as the proof of Proposition~\ref{theo:distill_ent_ch_st_ub}. Since we do not need the direct-sum property of the $\alpha$-geometric unextendible entanglement of states in the determinstic distillation case, we can take the limit $\alpha\to 0^+$ and arrive at~\eqref{eq:distill_ent_ch_st_ub_ex}.
\end{proof}

\begin{remark}
    In the asymptotic limit, the expected rate at which perfect ebits can be distilled from a quantum state $\rho_{CD}$ and an arbitrarily large number of instances of the quantum channel $\mathcal{N}_{AB\to A'B'}$, assisted by one-way LOCC or two-extendible superchannels, is bounded from above by the unextendible entanglement of the quantum channel induced by the Belavkin--Staszewski relative entropy, and the rate of exact one-way distillable entanglement is upper bounded by the min-geometric unextendible entanglement of the channel; i.e.,
    \begin{align}
        \liminf_{n\to \infty}\frac{1}{n}\widetilde{E}_{d,\operatorname{1WL}}\!\left(\rho,\mathcal{N}^{\otimes n}\right) &\le \widehat{E}^u\!\left(\mathcal{N}_{AB\to A'B'}\right) , \\
        \liminf_{n\to \infty}\frac{1}{n}\widetilde{E}_{d,e,\operatorname{1WL}}\!\left(\rho,\mathcal{N}^{\otimes n}\right) &\le \widehat{E}^u_{\operatorname{min}}\!\left(\mathcal{N}_{AB\to A'B'}\right).
    \end{align}
    The above inequalities follow from the fact that the $\alpha$-geometric unextendible entanglement of a quantum state is finite. As we take the limit $n\to \infty$, the quantities $\frac{1}{n}\widehat{E}^u\!\left(\rho_{CD}\right)$ and $\frac{1}{n}\widehat{E}^u_{\min}\!\left(\rho_{CD}\right)$ approach zero for every state $\rho_{CD}$, yielding the above inequalities.
\end{remark}

\subsection{Distillable key of a quantum state-channel pair}\label{sec:dist_sec_key_st_ch_pair}

In this section we consider the task of distilling secret keys from a bipartite quantum state and multiple instances of an unextendible quantum channel. It has been shown in~\cite[Corollary 22]{WWW19} that the $\alpha$-geometric unextendible entanglement of a bipartite quantum state serves as an upper bound on the number of secret bits that can be distilled from the quantum state using one-way LOCC or two-extendible channels for all $\alpha \in (0,2]$. A no-go theorem for probabilistic secret-key distillation from a bipartite state was given in~\cite{SW24} using the unextendible entanglement of the state induced by the min-relative entropy~\cite{Dat09}. We extend these results to give an upper bound on the number of exact secret key bits that can be distilled from a bipartite quantum state and an unextendible quantum channel, assisted by local operations and one-way classical communication.

\subsubsection{Probabilistic one-way distillable key}

Let us first look at the probabilistic setting in which a secret key is established between Alice and Bob using a quantum state $\rho_{CD}$ and a quantum channel $\mathcal{N}_{AB\to A'B'}$ assisted by local operations and one-way classical communication from Alice to Bob. Similar to the formalism in Section~\ref{sec:prob_distill_key}, we utilize the fact that distilling a secret key of $\log_2 K$ bits is equivalent to distilling a bipartite private state $\gamma^K_{C'D'C''D''}$ holding $\log_2 K$ bits of secrecy when local operations are allowed for free~\cite{Horodecki_2005,K_Horodecki_2009} (see Section~\ref{sec:prob_distill_key} for details). We define the probabilistic one-way distillable key of a quantum state-channel pair $\left(\rho_{CD},\mathcal{N}_{AB\to A'B'}\right)$ as
\begin{multline}
    \widetilde{K}_{\operatorname{1WL}}\!\left(\rho_{CD},\mathcal{N}_{AB\to A'B'}\right)  \coloneqq \\
    \sup_{\substack{ \gamma^K_{C'D'C''D''}\in \mathcal{K}, \\ p \in [0,1], \Theta\in \operatorname{1WL}}}\left\{
    \begin{array}{c}
         p\log_2 K:\\
         \left(\Theta\!\left(\mathcal{N}\right)\right)\!\left(\rho\right) = p |1\rangle\!\langle 1|_X \otimes \gamma^K_{C'D'C''D''}\\
         + (1-p)|0\rangle\!\langle 0|_X \otimes \sigma_{C'D'C''D''} ,\\
         \sigma\in \mathcal{S}(C'D'C''D'')
    \end{array}
    \right\}.
\end{multline}

Once again, we can relax the set of allowed superchannels to the set of two-extendible superchannels and define the probabilistic two-extendible distillable key of a quantum state-channel pair as follows:
\begin{multline}
    \widetilde{K}_{\operatorname{2-EXT}}\!\left(\rho_{CD},\mathcal{N}_{AB\to A'B'}\right) \coloneqq \\
     \sup_{\substack{\gamma^K_{C'D'C''D''}\in \mathcal{K} ,\\ p \in [0,1],\Theta\in \operatorname{2-EXT}}}\left\{
    \begin{array}{c}
         p\log_2 K:\\
         \left(\Theta\!\left(\mathcal{N}\right)\right)\!\left(\rho\right) = p|1\rangle\!\langle 1|_X \otimes \gamma^K_{C'D'C''D''} \\
         + (1-p)|0\rangle\!\langle 0|_X \otimes \sigma_{C'D'C''D''} ,\\
         \sigma\in \mathcal{S}(C'D'C''D'')
    \end{array}
    \right\}.
\end{multline}
\begin{proposition}\label{theo:distill_key_ch_st_ub}
    The expected rate at which secret bits can be probablistically distilled between Alice and Bob from a bipartite state $\rho_{CD}$ and $n$ instances of a bipartite channel $\mathcal{N}_{AB\to A'B'}$, assisted by one-way LOCC or two-extendible superchannels, is bounded from above as follows:
    \begin{align}
        \frac{1}{n}\widetilde{K}_{\operatorname{1WL}}\!\left(\rho,\mathcal{N}^{\otimes n}\right) &\le \frac{1}{n}\widetilde{K}_{\operatorname{2-EXT}}\!\left(\rho,\mathcal{N}^{\otimes n}\right)\notag \\
        &\le \frac{1}{n}\widehat{E}^u\!\left(\rho_{CD}\right) + \widehat{E}^u\!\left(\mathcal{N}_{AB\to A'B'}\right).
    \end{align}
\end{proposition}
\begin{proof}
    The proof is similar to the proof of Proposition~\ref{theo:distill_ent_ch_st_ub}. The only difference is that we use the fact that the $\alpha$-geometric unextendible entanglement of a bipartite private state is no less than the number of secret bits held by the private state, for all $\alpha \in (0,2]$~\cite[Corollary 22]{WWW19}.
\end{proof}

\subsubsection{Exact one-way distillable key }

We now look at a deterministic protocol to distill a secret key from a bipartite state using  a quantum channel assisted by local operations and forward classical communication. Alice and Bob use $n$ instances of a bipartite quantum channel $\mathcal{N}_{AB\to A'B'}$ assisted by a one-way LOCC superchannel $\Theta_{(A^nB^n\to A'^nB'^n)\to (CD\to C'D'C''D'')}$ to transform the existing bipartite quantum state $\rho_{CD}$ to a bipartite private state $\gamma^K_{C'D'C''D''}$ holding $\log_2 K$ secret bits:
\begin{equation}
    \left(\Theta\!\left(\mathcal{N}^{\otimes n}_{AB\to A'B'}\right)\right)\!\left(\rho_{CD}\right) = \gamma^K_{C'D'C''D''}.
\end{equation}

The exact one-way distillable key of a quantum state-channel pair $\left(\rho_{CD},\mathcal{N}_{AB\to A'B'}\right)$ is defined as
\begin{multline}
    \widetilde{K}^{e}_{\operatorname{1WL}}\!\left(\rho_{CD},\mathcal{N}_{AB\to A'B'}\right) \\
    \coloneqq \sup_{\substack{\gamma^K_{C'D'C''D''}\in \mathcal{K},\\ \Theta\in \operatorname{1WL}}}\left\{
    \begin{array}{c}
         \log_2 K:\\
         \left(\Theta\!\left(\mathcal{N}\right)\right)\!\left(\rho\right) = \gamma^K_{C'D'C''D''}
    \end{array}
    \right\}.
\end{multline}
Relaxing the set of allowed superchannels to the set of two-extendible superchannels, we define the the exact two-extendible distillable key of a quantum state-channel pair as follows:
\begin{multline}
    \widetilde{K}^{e}_{\operatorname{2-EXT}}\!\left(\rho_{CD},\mathcal{N}_{AB\to A'B'}\right) \\
    \coloneqq \sup_{\substack{\gamma^K_{C'D'C''D''} \in \mathcal{K},\\ \Theta\in \operatorname{2-EXT}}}\left\{
    \begin{array}{c}
         \log_2 K:\\
         \left(\Theta\!\left(\mathcal{N}\right)\right)\!\left(\rho\right) = \gamma^K_{C'D'C''D''}
    \end{array}
    \right\}.
\end{multline}
\begin{proposition}\label{theo:distill_key_ch_st_ub_ex}
    The rate at which exact secret bits can be exactly distilled from a quantum state-channel pair, assisted by one-way LOCC or two-extendible superchannels, is bounded from above by the following quantity:
    \begin{align}
        &\frac{1}{n}\widetilde{K}^{e}_{\operatorname{1WL}}\!\left(\rho,\mathcal{N}^{\otimes n}\right)\notag\\
        &\le \frac{1}{n}\widetilde{K}^{e}_{\operatorname{2-EXT}}\!\left(\rho,\mathcal{N}^{\otimes n}\right) \\
        &\le \frac{1}{n}\widehat{E}^u_{\min}\!\left(\rho_{CD}\right) + \widehat{E}^u_{\min}\!\left(\mathcal{N}_{AB\to A'B'}\right).
    \end{align}
\end{proposition}
\begin{proof}
    The proof follows from the same arguments used in the proof of Proposition~\ref{theo:distill_ent_ch_st_ub_ex}, and using the fact that the $\alpha$-geometric unextendible entanglement of a bipartite state holding $\log_2 K$ secret bits is no less than $\log_2 K$~\cite[Corollary 22]{WWW19} for all $\alpha \in (0,2]$.
\end{proof}

\begin{remark}
    Since the $\alpha$-geometric unextendible entanglement of a quantum state is finite, the expected rate of distilling secret key bits from a quantum state $\rho_{CD}$ probabilistically, using an arbitrarily large number of instances of a quantum channel $\mathcal{N}_{AB\to A'B'}$, assisted by one-way LOCC or two-extendible superchannels, is bounded from above by the unextendible entanglement of the quantum channel $\mathcal{N}_{AB\to A'B'}$ induced by the Belavkin--Staszewski relative entropy. That is,
    \begin{equation}
        \liminf_{n\to \infty}\frac{1}{n}\widetilde{K}_{\operatorname{2-EXT}}\!\left(\rho_{CD},\mathcal{N}^{\otimes n}\right) \le \widehat{E}^u\!\left(\mathcal{N}_{AB\to A'B'}\right).
    \end{equation}
    Similarly, the rate of distilling secret bits exactly from a quantum state $\rho_{CD}$, using an arbitrarily large number of instances of a quantum channel $\mathcal{N}_{AB\to A'B'}$, assisted by one-way LOCC or two-extendible superchannels, is bounded from above by the min-geometric unextendible entanglement of the quantum channel $\mathcal{N}_{AB\to A'B'}$. That is,
    \begin{equation}
        \liminf_{n\to \infty}\frac{1}{n}\widetilde{K}^{e}_{\operatorname{2-EXT}}\!\left(\rho_{CD},\mathcal{N}^{\otimes n}\right) \le \widehat{E}^u_{\operatorname{min}}\!\left(\mathcal{N}_{AB\to A'B'}\right).
    \end{equation}
\end{remark}

\section{Numerical calculations}\label{sec:numerical_calc}

In this section we present some calculations for the upper bounds proposed in Sections~\ref{sec:applications} and~\ref{sec:applications_bip}. It has been shown in~\cite{Fang_2021} that the $\alpha$-geometric R\'enyi relative entropy of channels can be calculated using a semidefinite program for $\alpha = 1 + 2^{-\ell}$ with $\ell \in \mathbb{N}$. This allows us to compute an upper bound on the unextendible entanglement of a quantum channel induced by the Belavkin--Staszewski relative entropy since the optimization is over a set of channels expressible by semidefinite constraints.

While a semidefinite program is not known for min-geometric unextendible entanglement of channels, the quantity can be calculated for some special channels. The $\alpha$-geometric R\'enyi relative entropy between two quantum channels $\mathcal{N}_{A\to B}$ and $\mathcal{M}_{A\to B}$, with Choi operators $\Gamma^{\mathcal{N}}_{AB}$ and $\Gamma^{\mathcal{M}}_{AB}$, respectively, can be calculated for all $\alpha \in (0,1)$ using the following equality from~\cite[Proposition 44]{Katariya2021}: 
\begin{multline}\label{eq:geo_rel_ent_ch_choi_op_(0,1)}
    \widehat{D}_{\alpha}\!\left(\mathcal{N}_{A\to B}\Vert\mathcal{M}_{A\to B}\right) \\= \frac{1}{\alpha-1}\log_2 \lambda_{\min}\!\left(\operatorname{Tr}_{B}\!\left[G_{\alpha}\!\left(\Gamma^{\mathcal{M}}_{AB},\widetilde{\Gamma^{\mathcal{N}}_{AB}}\right)\right]\right),
\end{multline}
where $\lambda_{\min}$ denotes the minimum eigenvalue of its argument,
\begin{align}
    G_{\alpha}\!\left(X,Y\right) &\coloneqq X^{\frac{1}{2}}\left(X^{-\frac{1}{2}}YX^{-\frac{1}{2}}\right)^{\alpha}X^{\frac{1}{2}},\\
    \widetilde{\Gamma^{\mathcal{N}}_{AB}} &\coloneqq \Gamma^{\mathcal{N}}_{0,0} - \Gamma^{\mathcal{N}}_{0,1}\left(\Gamma^{\mathcal{N}}_{1,1}\right)^{-1}\left(\Gamma^{\mathcal{N}}_{0,1}\right)^{\dagger},\label{eq:N_tilde_defn}\\
    \Gamma^{\mathcal{N}}_{0,0} &\coloneqq \Pi_{\Gamma^{\mathcal{M}}}\Gamma^{\mathcal{N}}\Pi_{\Gamma^{\mathcal{M}}}, \\ \Gamma^{\mathcal{N}}_{0,1} & \coloneqq \Pi_{\Gamma^{\mathcal{M}}}\Gamma^{\mathcal{N}}\Pi^{\perp}_{\Gamma^{\mathcal{M}}}, \\ \Gamma^{\mathcal{N}}_{1,1}  & \coloneqq \Pi^{\perp}_{\Gamma^{\mathcal{M}}}\Gamma^{\mathcal{N}}\Pi^{\perp}_{\Gamma^{\mathcal{M}}}, 
\end{align}
$\Pi_{\Gamma^{\mathcal{M}}}$ is the projection onto the support of the $\Gamma^{\mathcal{M}}$, $\Pi^{\perp}_{\Gamma^{\mathcal{M}}}$ is the projection onto its kernel, and all inverses are taken on the respective support. The min-geometric R\'enyi relative entropy of channels can then be calculated as follows:
\begin{align}
    \widehat{D}_0\!\left(\mathcal{N}\Vert \mathcal{M}\right) &\coloneqq \sup_{\psi_{RA}} \lim_{\alpha \to 0^+} \widehat{D}_{\alpha}\!\left(\mathcal{N}\!\left(\psi_{RA}\right)\Vert \mathcal{M}\!\left(\psi_{RA}\right)\right)\\
    &= \sup_{\psi_{RA}} \inf_{\alpha \in (0,1)} \widehat{D}_{\alpha}\!\left(\mathcal{N}\!\left(\psi_{RA}\right)\Vert \mathcal{M}\!\left(\psi_{RA}\right)\right)\\
    &= \inf_{\alpha \in (0,1)} \sup_{\psi_{RA}} \widehat{D}_{\alpha}\!\left(\mathcal{N}\!\left(\psi_{RA}\right)\Vert \mathcal{M}\!\left(\psi_{RA}\right)\right)\\
    &= \inf_{\alpha \in (0,1)} \widehat{D}_{\alpha}\!\left(\mathcal{N}\Vert \mathcal{M}\right)\\
    &=  \lim_{\alpha \to 0^+} \widehat{D}_{\alpha}\!\left(\mathcal{N}\Vert \mathcal{M}\right),
\end{align}
where the first equality is simply the definition of the min-geometric R\'enyi relative entropy between channels, the second equality follows from the monotonicity of the $\alpha$-geometric R\'enyi relative entropy in $\alpha$, the third equality follows from the Mosonyi--Hiai minimax theorem given in~\cite[Corollary~A.2]{MH11}, and the final equality follows by using the monotonicity of the $\alpha$-geometric R\'enyi relative entropy in $\alpha$ once again. Now using the expression of $\alpha$-geometric R\'enyi relative entropy from~\eqref{eq:geo_rel_ent_ch_choi_op_(0,1)}, we conclude the following equality:
\begin{equation}
    \widehat{D}_{0}\!\left(\mathcal{N}\Vert\mathcal{M}\right) = -\log_2 \lambda_{\min}\!\left(\operatorname{Tr}_B\left[\Gamma^{\mathcal{M}}\Pi_{\zeta}\right]\right),
\end{equation}
where $\Pi_{\zeta}$ is the projection onto the support of $\left(\Gamma^{\mathcal{M}}\right)^{-\frac{1}{2}}\widetilde{\Gamma^{\mathcal{N}}}\!\left(\Gamma^{\mathcal{M}}\right)^{-\frac{1}{2}}$ and $\widetilde{\Gamma^{\mathcal{N}}}$ is defined in~\eqref{eq:N_tilde_defn}.
\begin{proposition}
    The min-geometric unextendible entanglement of quantum channels with full rank Choi operators is equal to zero.
\end{proposition}

\begin{proof}
    Consider a quantum channel $\mathcal{N}_{A\to B}$ that has a full-rank Choi operator $\Gamma^{\mathcal{N}}_{AB}$. The quantum channel $\mathcal{N}_{A\to B_1}\otimes\mathcal{A}^{\pi}_{B_2}$ is a valid extension of the channel that lies in the set $\operatorname{Ext}\!\left(\mathcal{N}_{A\to B}\right)$, where $\mathcal{A}^{\pi}_{B_2}$ is a quantum channel that appends the maximally mixed state on the system $B_2$, and the system $B_2$ is isomorphic to the system $B$. The Choi operator of this channel is $\Gamma^{\mathcal{N}}_{AB_1}\otimes \pi_{B_2}$, where $\pi$ is the maximally mixed state. 
    
    By definition, the min-geometric unextendible entanglement of the channel $\mathcal{N}_{A\to B}$ is bounded from above as follows:
    \begin{equation}\label{eq:unext_full_rank_le_rel_ent}
        \widehat{E}^u_{\min}\!\left(\mathcal{N}_{A\to B}\right) 
        \le \frac{1}{2}\widehat{D}_{0}(\mathcal{N}_{A\to B}\Vert\operatorname{Tr}_{B_1}\circ\mathcal{N}_{A\to B_1}\otimes \mathcal{A}^{\pi}_{B_2}).
    \end{equation}
    The quantum channel $\operatorname{Tr}_{B_1}\circ\mathcal{N}_{A\to B_1}\otimes \mathcal{A}^{\pi}_{B_2}$ is essentially the replacer channel $\mathcal{R}^{\pi}_{A\to B_2}$ that traces out the input and replaces with the maximally mixed state. The Choi operator of this channel is
    \begin{equation}\label{eq:replacer_full_rank_choi}
        \Gamma^{\mathcal{R}}_{AB_2} = \frac{I_{AB_2}}{\left|B\right|},
    \end{equation}
    where $\left|B\right|$ is the dimension of the system $B$.
    Since $\Gamma^{\mathcal{R}}_{AB}$ is also full rank, $\operatorname{supp}\!\left(\Gamma^{\mathcal{N}}\right)\subseteq\operatorname{supp}\!\left(\Gamma^{\mathcal{M}}\right)$, which further implies that $\widetilde{\Gamma^{\mathcal{N}}} = \Gamma^{\mathcal{N}}$. The positive semidefinite operators $\Gamma^{\mathcal{N}}_{AB}$ and $\Gamma^{\mathcal{R}}_{AB}$ are both full-rank operators. Therefore, $\left(\Gamma^{\mathcal{R}}\right)^{-\frac{1}{2}}\Gamma^{\mathcal{N}}\left(\Gamma^{\mathcal{R}}\right)^{-\frac{1}{2}}$ is also a full-rank operator, and the projection onto the support of $\left(\Gamma^{\mathcal{R}}\right)^{-\frac{1}{2}}\Gamma^{\mathcal{N}}\left(\Gamma^{\mathcal{R}}\right)^{-\frac{1}{2}}$ is the identity operator. That is,
    \begin{equation}\label{eq:geo_proj_full_rank}
        \Pi_{\left(\Gamma^{\mathcal{R}}\right)^{-\frac{1}{2}}\Gamma^{\mathcal{N}}\left(\Gamma^{\mathcal{R}}\right)^{-\frac{1}{2}}} = I_{AB}.
    \end{equation}
    
    The min-geometric R\'enyi relative entropy between $\mathcal{N}_{A\to B}$ and $\mathcal{R}_{A\to B}$ can be evaluated as follows:
    \begin{align}
        &\frac{1}{2}\widehat{D}_{0}\!\left(\mathcal{N}\Vert\mathcal{R}^{\pi}\right)\notag\\ 
        &= -\log_2 \lambda_{\min}\!\left(\operatorname{Tr}_B\!\left[\Gamma^{\mathcal{R}}_{AB}\Pi_{\left(\Gamma^{\mathcal{R}}\right)^{-\frac{1}{2}}\Gamma^{\mathcal{N}}\left(\Gamma^{\mathcal{R}}\right)^{-\frac{1}{2}}}\right]\right)\\
        &= -\log_2 \lambda_{\min}\!\left(\operatorname{Tr}_{B}\!\left[\frac{I_{AB}}{\left|B\right|}I_{AB}\right]\right)\\
        &= -\log_2 \lambda_{\min}\!\left(I_A\right)\\
        &= 0.\label{eq:min_geo_rel_ent_le_0_full_rank}
    \end{align}
    where the second equality follows from~\eqref{eq:replacer_full_rank_choi} and~\eqref{eq:geo_proj_full_rank}. The non-negativity of the min-geometric unextendible entanglement combined with~\eqref{eq:unext_full_rank_le_rel_ent} and~\eqref{eq:min_geo_rel_ent_le_0_full_rank} implies that the min-geometric unextendible entanglement of a quantum channel that has a full-rank Choi operator is equal to zero.
\end{proof}

\medskip

Since the min-geometric unextendible entanglement of a quantum channel serves as an upper bound on the zero-error quantum capacity (Corollary~\ref{cor:0_err_cap_ub}) and the zero-error private capacity of the channel (Corollary~\ref{cor:0_err_priv_cap_ub}), we arrive at the following statement:
\begin{corollary}\label{cor:q_cap_p_cap_0_full_rank}
    The zero-error quantum capacity and the zero-error private capacity of a quantum channel with a full-rank Choi operator is equal to zero in the one-shot as well as the asymptotic setting.
\end{corollary}

Let us now look at the $\alpha$-geometric unextendible entanglement of a channel for $\alpha > 1$, which can be computed using a semidefinite program.

\subsection{Semidefinite program for \texorpdfstring{$\alpha$}{alpha}-geometric unextendible entanglement of point-to-point channels}

In Section~\ref{sec:applications} we showed that the min-geometric unextendible entanglement of a channel is an upper bound on several operationally relevant quantities corresponding to the channel. However, a semidefinite program to compute the min-geometric unextendible entanglement of an arbitrary channel is not known. The $\alpha$-geometric unextendible entanglement of the channel serves as a weaker upper bound on all the quantities mentioned in Section~\ref{sec:applications} due to the monotonicity of the $\alpha$-geometric unextendible entanglement of channels in $\alpha$ (see Proposition~\ref{prop:geo_unext_ent_monotonic}). The $\alpha$-geometric unextendible entanglement of a channel can be calculated for some values of $\alpha\in (1,2]$ using a semidefinite program, which we describe in this section. 

We make use of the semidefinite program given in~\cite[Lemma 9]{Fang_2021} that calculates the following quantity:
\begin{equation}
	f_{\mathcal{V}}\!\left(\mathcal{N}\right) \coloneqq \min_{\mathcal{M}\in \mathcal{V}}\widehat{D}_{\alpha}\!\left(\mathcal{N}\Vert\mathcal{M}\right), 
\end{equation}
where $\mathcal{V}$ is a set of channels described by semidefinite constraints. Recall that lower semi-continuity of $\alpha$-geometric R\'enyi relative entropy allows us to replace the infimum with a minimum, providing a way to calculate the $\alpha$-geometric unextendible entanglement of a channel.

We first propose a semidefinite program for calculating the $\alpha$-geometric unextendible entanglement of a point-to-point quantum channel.

\begin{proposition}\label{prop:SDP_for_geo_unext_ent}
	The $\alpha$-geometric unextendible entanglement of a channel $\mathcal{N}_{A\to B}$ can be calculated, for $\alpha = 1 + 2^{-\ell}$ where $\ell\in \mathbb{N}$, using the following semidefinite program:
	\begin{equation}
		\widehat{E}^u_{\alpha}\!\left(\mathcal{N}_{A\to B}\right) = 2^\ell \min_{\substack{y\in \mathbb{R}, \Gamma^{\mathcal{P}}_{AB_1B_2} \ge 0\\M_{AB}, \left\{N^i_{AB}\right\}_{i=0}^{\ell},\in \operatorname{Herm}}} \log_2 y,
	\end{equation}
	subject to the constraints,
	\begin{align}
		\operatorname{Tr}_{B_2}\!\left[\Gamma^{\mathcal{P}}_{AB_1B_2}\right] &= \Gamma^{\mathcal{N}}_{AB}\label{eq:non_sig_SDP_geo_unext_ent},\\
		\operatorname{Tr}_{B} \left[M_{AB}\right] &\le yI_A\label{eq:geo_unext_ent_cond_1},\\
		\operatorname{Tr}_{B_1} \left[\Gamma^{\mathcal{P}}_{AB_1B_2}\right] &= N^0_{AB}, \label{eq:geo_unext_ent_cond_2}\\
		\begin{bmatrix}
			M_{AB} & \Gamma^{\mathcal{N}}_{AB}\\
			\Gamma^{\mathcal{N}}_{AB} & N^{\ell}_{AB}
		\end{bmatrix}
		&\ge 0,\label{eq:geo_unext_ent_cond_3}\\
		\begin{bmatrix}
			\Gamma^{\mathcal{N}}_{AB} & N^i_{AB}\\
			N^i_{AB} & N^{i-1}_{AB}
		\end{bmatrix}
		&\ge 0 \quad \forall i\in \{1,2,\ldots,\ell\},\label{eq:geo_unext_ent_cond_4}
	\end{align}
	where $\Gamma^{\mathcal{N}}_{AB}$ is the Choi operator of the channel $\mathcal{N}_{A\to B}$, and $B$, $B_1$, and $B_2$ are isomorphic systems.
\end{proposition}

The constraints in~\eqref{eq:non_sig_SDP_geo_unext_ent} and~\eqref{eq:geo_unext_ent_cond_2} ensure that $\Gamma^{\mathcal{N}}_{AB}$ and $N^0_{AB}$ are the Choi operators of the marginal channels $\operatorname{Tr}_{B_2}\circ\mathcal{P}_{A\to B_1B_2}$ and $\operatorname{Tr}_{B_1}\circ\mathcal{P}_{A\to B_1B_2}$. The constraints in~\eqref{eq:geo_unext_ent_cond_1}--\eqref{eq:geo_unext_ent_cond_4} are the semidefinite constraints to calculate the quantity $\min_{\mathcal{P}_{A\to B_1B_2}\in \mathcal{V}}\widehat{D}_{\alpha}\!\left(\mathcal{N}_{A\to B}\Vert\operatorname{Tr}_{B_2}\circ\mathcal{P}_{A\to B_1B_2}\right)$ for a set of channels $\mathcal{V}$ defined by semidefinite conditions, given in~\cite[Lemma 9]{Fang_2021}. 

Since $\ell$ dictates the number of variables in the semidefinite program, we cannot practically use this SDP to calculate $\widehat{E}^u_{\alpha}$ for $\alpha$ arbitrarily close to one. We settle for $\ell=10$ in our calculations to get an upper bound on the unextendible entanglement of a channel induced by the Belavkin--Staszewski relative entropy.

\subsection{Semidefinite program for \texorpdfstring{$\alpha$}{alpha}-geometric unextendible entanglement of bipartite channels} 

The semidefinite program for the $\alpha$-geometric unextendible entanglement of a channel in Proposition~\ref{prop:SDP_for_geo_unext_ent} can be generalized for bipartite channels $\mathcal{N}_{AB\to A'B'}$ by making the identifications $AB\leftrightarrow A$ and $A'B'\leftrightarrow B$, and modifying the constraints for marginal channels.

\begin{proposition}
	The $\alpha$-geometric unextendible entanglement of a channel $\mathcal{N}_{A\to B}$ can be calculated, for $\alpha = 1 + 2^{-\ell}$ where $\ell\in \mathbb{N}$, using the following semidefinite program:
	\begin{equation}
		\widehat{E}^u_{\alpha}\!\left(\mathcal{N}_{AB\to A'B'}\right) = 2^{\ell} \min_{\substack{y\in \mathbb{R}, \Gamma^{\mathcal{P}}_{AB_1B_2A'B'_1B'_2} \ge 0\\M_{ABA'B}, \left\{N^i_{ABA'B'}\right\}_{i=0}^{\ell},\in \operatorname{Herm}}} \log_2 y,
	\end{equation}
	subject to the constraints,
	\begin{align}
		\operatorname{Tr}_{B'_2}\!\left[\Gamma^{\mathcal{P}}_{AB_1B_2A'B'_1B'_2}\right] &= \Gamma^{\mathcal{N}}_{ABA'B'}\otimes I_{B_2}, \label{eq:non_sig_SDP_geo_unext_ent_bip}\\
		\operatorname{Tr}_{A'B'} \left[M_{ABA'B'}\right] &\le yI_{AB} \label{eq:geo_unext_ent_cond_bip_1},\\
		\operatorname{Tr}_{B_1B'_1}\!\left[\Gamma^{\mathcal{P}}_{AB_1B_2A'B'_1B'_2}\right] &= N^0_{ABA'B'}, \label{eq:geo_unext_ent_cond_bip_2}\\
		\begin{bmatrix}
			M_{ABA'B'} & \Gamma^{\mathcal{N}}_{ABA'B'}\\
			\Gamma^{\mathcal{N}}_{ABA'B'} & N^{\ell}_{ABA'B'}
		\end{bmatrix}
		&\ge 0,\label{eq:geo_unext_ent_cond_bip_3}\\
		\begin{bmatrix}
			\Gamma^{\mathcal{N}}_{ABA'B'} & N^i_{ABA'B'}\\
			N^i_{ABA'B'} & N^{i-1}_{ABA'B'}
		\end{bmatrix}
		&\ge 0 \quad\forall i\in \{1,2,\ldots,\ell\},\label{eq:geo_unext_ent_cond_bip_4}
	\end{align}
	where $\Gamma^{\mathcal{N}}_{ABA'B'}$ is the Choi operator of the channel $\mathcal{N}_{A\to B}$ and $d_B$ is the dimension of system $B$. Systems $B_1$ and $B_2$ are isomorphic to the system $B$, and systems $B'_1$ and $B'_2$ are isomorphic to the system $B'$.
\end{proposition}

\begin{remark}
    We have used the subadditivity of the $\alpha$-geometric unextendible entanglement (Proposition~\ref{prop:unext_ent_subadditive}) to obtain a single-letter upper bound on all the operational quantities of a quantum channel discussed in Sections~\ref{sec:applications} and \ref{sec:applications_bip}, either in terms of the min-geometric unextendible entanglement of the channel or the unextendible entanglement of the channel induced by the Belavkin--Staszewski relative entropy. However, one can obtain a tighter bound by using the regularized unextendible entanglement of the channel because, for all $n\in \mathbb{N}$,
    \begin{align}
        \widehat{E}^u\!\left(\mathcal{N}_{A\to B}\right) &\ge \frac{1}{n}\widehat{E}^u\!\left(\mathcal{N}^{\otimes n}_{A\to B}\right),\\
        \widehat{E}^u_{\min}\!\left(\mathcal{N}_{A\to B}\right) &\ge \frac{1}{n}\widehat{E}^u_{\min}\!\left(\mathcal{N}^{\otimes n}_{A\to B}\right).
    \end{align}
    In practice, direct implementations of the optimizations for these quantities are much harder to calculate with increasing $n$ for arbitrary quantum channels, as the dimensions of the corresponding semidefinite programs increase exponentially with $n$ (however, it could be the case that the approach from~\cite{FST22}, which incorporates permutational symmetry, could make the computational difficulty of these optimizations grow only polynomially with $n$). By taking the asymptotic limit of these quantities, an arbitrarily large number of channel uses might be required to estimate it, restricting its usefulness from a practical perspective. Nonetheless, we remark that the regularized unextendible entanglement quantities give a generally tighter bound on the respective operational quantities than the single-letter upper bounds discussed in this section.
\end{remark}

\subsection{Semicausal channel for probabilistic distillation of resource}

The channel mentioned in~\eqref{eq:eras_ch_full_ext} is a specific case of a semicausal channel where Bob's input system is a trivial system. Alice sends some quantum information to Bob through an erasure channel, but instead of the quantum state being lost to the environment upon erasure, the state is returned to Alice. If the state is successfully transmitted to Bob, Alice receives back the erasure symbol instead. 

There exists a simple protocol to probabilistically distill entanglement from this channel using one-way LOCC. Alice sends one share of a locally prepared maximally entangled state to Bob using the channel. If the state is successfully transmitted to Bob, Alice receives an erasure symbol. Alice can determine if her state is erased or not by performing the POVM $\{\Pi,|e\rangle\!\langle e|\}$ on the system she received, where $\Pi$ is the projection onto the entire Hilbert space of her system. If she measures her system to be erased, she knows that a maximally entangled state has been established between herself and Bob, and she can indicate the results of her measurement to Bob, hence, distilling a maximally entangled state probabilistically. If the probability of erasure is $p$ and the channel allows Alice to send a $d$-dimensional state to Bob, then the probabilistic one-way distillable entanglement of the channel is no less than $(1-p)\log_2 d$. We find that the unextendible entanglement of this channel induced by the Belavkin--Staszewski relative entropy, which is an upper bound on the probabilitic distillable entanglement of the channel, is equal to $(1-p)\log_2 d$ as well, as we state formally in Proposition~\ref{prop:full_eras_unext_ent}.

\begin{proposition}\label{prop:full_eras_unext_ent}
    Consider a quantum channel that acts on an arbitrary state $\rho_{RA}$ as follows:
    \begin{multline}
        \mathcal{N}_{A\to A'B}\!\left(\rho_{RA}\right) = p\rho_{RA'}\otimes |e\rangle\!\langle e|_B \\+ (1-p)\rho_{RB}\otimes |e\rangle\!\langle e|_{A'}.
    \end{multline}
The unextendible entanglement of the channel mentioned above, induced by the Belavkin--Staszewski relative entropy, is equal to $(1-p)\log_2 d$.
\end{proposition}
\begin{proof}
    See Appendix~\ref{app:full_eras_unext_ent}.
\end{proof}

\medskip

One can consider a less idealistic channel, where the state that Alice sends to Bob is truly erased but Alice receives some information about the erasure process. Let us consider a channel that acts on an arbitrary state $\rho_{RA}$ as follows:
\begin{multline}
    \mathcal{N}_{A\to A'B}\!\left(\rho_{RA}\right) \coloneqq (1-p)\rho_{RB}\otimes\sigma_{A'} \\+ p\pi_R\otimes |e\rangle\!\langle e|_B\otimes \tau_{A'}, 
\end{multline}
where $\pi$ is a maximally mixed state, and $\sigma_{A'}$ and $\tau_{A'}$ are some quantum states. 

The probabilistic one-way distillable entanglement of this channel depends on Alice's ability to distinguish between the states $\sigma$ and $\tau$. Naturally, if $\sigma$ and $\tau$ are orthogonal, then Alice can perfectly distinguish between the two states, and the probabilistic one-way distillable entanglement of the channel will be equal to $(1-p)\log_2 d$. However, if the states are not orthogonal Alice will not be able to perfectly distinguish between the two states, and the number of ebits that can be probabilistically distilled from the channel using one-way LOCC would be smaller than $(1-p)\log_2 d$. 

We consider a simple case where Alice receives a single classical bit indicating if the state she sent to Bob was erased or not, where the classical bit also undergoes depolarizing noise. As such,
\begin{align}
    \sigma_{A'} &= \mathcal{D}_q\!\left(|1\rangle\!\langle 1|_{A'}\right),\\
    \tau_{A'} &= \mathcal{D}_q\!\left(|0\rangle\!\langle 0|_{A'}\right),
\end{align}
where
\begin{equation}
    \mathcal{D}_q\!\left(\rho_{A'}\right) = (1-q)\rho_{A'} + q\pi_{A'}.
\end{equation}
In Figure~\ref{fig:spec_ch_unext_ent}, we plot the $\alpha$-geometric unextendible entanglement of the following channel:
\begin{multline}\label{eq:Spec_eras_ch}
    \widetilde{\mathcal{E}}^{p,q}_{A\to A'B}\!\left(\rho_{RA}\right) \coloneqq (1-p)\rho_{RB}\otimes\mathcal{D}_q\!\left(|1\rangle\!\langle 1|_{A'}\right) \\+ p\pi_R\otimes |e\rangle\!\langle e|_B\otimes \mathcal{D}_q\!\left(|0\rangle\!\langle 0|_{A'}\right), 
\end{multline}
for different values of $p$ and $q$, and $\alpha = 1+2^{-10}$. The code for generating the figures in this paper is available with the arXiv posting.

\begin{figure}
    \centering
    \includegraphics[width = \linewidth]{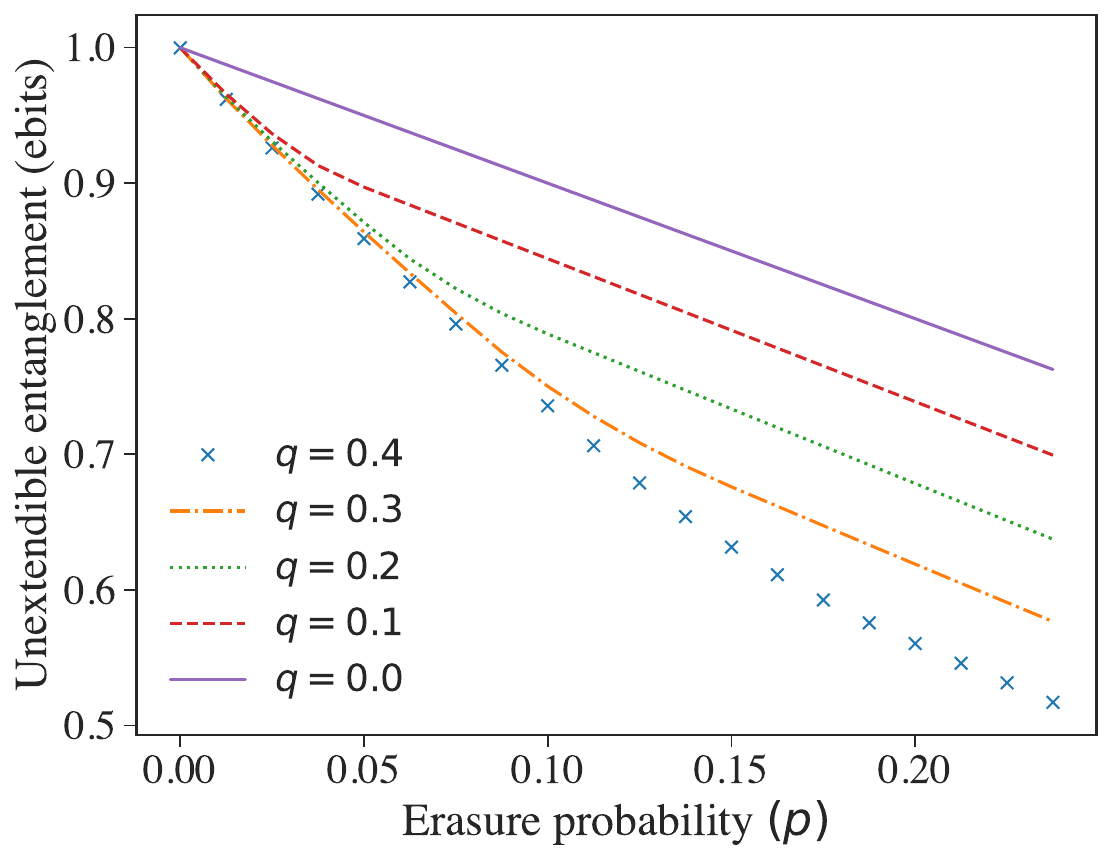}
    \caption{Here we plot the $\alpha$-geometric unextendible entanglement of the channel mentioned in~\eqref{eq:Spec_eras_ch}, for $\alpha = 1+2^{-10}$. The channel takes a two-dimensional state as input. The $\alpha$-geometric unextendible entanglement of a channel is an upper bound on the probabilistic one-way distillable entanglement as well as the probabilitic one-way distillable key of the channel for all $\alpha \in (1,2]$.}
    \label{fig:spec_ch_unext_ent}
\end{figure}

In the examples considered above, Alice can deduce if the state she sent to Bob was erased by analyzing the state she received from the semicausal channel, which allows Alice and Bob to probabilistically distill resource from such channels using only one-way LOCC. While the erasure channel and depolarizing channels cannot be used for exact or probabilistic distillation of entanglement or secret keys using only one-way LOCC, it is unclear if these channels cannot be used to boost the probabilistic one-way distillable entanglement or probabilistic one-way distillable key of a bipartite state in the presence of one-way LOCC (see Section~\ref{sec:applications_bip}). Nonetheless, we present an analytical upper bound on the $\alpha$-geometric unextendible entanglement of the erasure channel for all $\alpha \in (1,2]$, and an analytical expression for the unextendible entanglement of the depolarizing channel induced by the Belavkin--Staszewski relative entropy, in Appendices~\ref{app:unext_ent_eras_ch} and~\ref{app:proof_unext_ent_dep}, respectively.

\section{Conclusion}

\subsection{Summary}

In this work we defined a class of entanglement measures for quantum channels called generalized unextendible entanglement of quantum channels, based on the resource theory of unextendibility. We showed that this quantity does not increase under two-extendible superchannels, and consequently, decreases monotonically under one-way LOCC superchannels as well. This makes the unextendible entanglement of quantum channels a useful quantity for analyzing information-processing tasks that involve transformations of quantum channels assisted by local operations and one-way classical communication.

We found some immediate applications of the unextendible entanglement of quantum channels. The unextendible entanglement of a point-to-point quantum channel, induced by the $\alpha$-geometric R\'enyi relative entropy as $\alpha\to 0$, was shown to be an upper bound on the one-way distillable key, one-way distillable entanglement, forward-assisted zero-error quantum capacity, and the forward-assisted zero-error private capacity of the quantum channel. We found this quantity, which we call the min-geometric unextendible entanglement of a channel, to be equal to zero for several important channels such as the erasure channel and all channels with a full-rank Choi operator, indicating that these channels are useless for zero-error quantum and private communication. 

The formalism of $k$-extendibility was extended to bipartite superchannels, and we defined the unextendible entanglement of bipartite semicausal quantum channels. This quantity allowed us to bound the change in unextendibility of a bipartite quantum state when acted upon by an arbitrary bipartite semicausal quantum channel. Using this formalism we gave upper bounds on the probabilistic one-way distillable entanglement and probabilistic one-way distillable key of a bipartite quantum state when a bipartite quantum channel, not necessarily simulable by local operations and one-way classical communication, is also available.

Finally, we gave a semidefinite program to calculate the unextendible entanglement of a quantum channel induced by the $\alpha$-geometric R\'enyi relative entropy for $\alpha = 1+2^{-\ell}$, where $\ell$ is a positive integer, providing a computationally feasible method to calculate the upper bounds on the probabilistic one-way distillable entanglement and probabilistic one-way distillable key of a state-channel pair. We showed some example calculations for this technique, evaluating the $\alpha$-geometric unextendible entanglement of special erasure channels where one party sends quantum data to a distant party through an erasure channel but also receives some information about the erasure process.

\subsection{Future directions}

There are several future directions to be explored. As was the case in~\cite{WWW19}, we have restricted all of our measures of unextendibility and applications to two-extendibility.  An extension of the measures to $k$-unextendibility could possibly give tighter bounds on several quantities of interest when using channels assisted by one-way LOCC superchannels. Our formalism also restricts us to the case of zero-error capacities. It will be an interesting direction to allow for arbitrary error in our formalism in order to explore a more general and practical setting. 

We have obtained several bounds in terms of the min-geometric R\'enyi relative entropy of quantum channels. To the best of our knowledge, no prior work has used this measure as an upper bound on information-theoretic quantities. We believe our approach with min-geometric R\'enyi relative entropy can be extended to other dynamical resource theories to obtain tighter bounds on several quantities of interest. Given that we have shown the relevance of the min-geometric R\'enyi relative entropy in our work, we think this motivates developing efficient methods to optimize this quantity with respect to semidefinite constraints.

When multiple uses of a quantum channel are allowed for resource distillation, we restricted the discussion to the independent and identically distributed scenario in which all the quantum channels are used in parallel. A more general case of sequential distillation protocols can be considered. Similar ideas have been pursued in~\cite{kaur2017amortized,berta2018amortization,GS20,Fang_2021} when positive partial transpose (PPT) channels are allowed for free, and a general treatment for arbitrary resource theories has been considered in~\cite{GS20_resources}.

Extending the formalism of unextendibility to bipartite semicausal channels, we gave upper bounds on the probabilistic distillable entanglement and secret key from a bipartite state using one-way LOCC and an unextendible bipartite semicausal channel. It is known that an erasure channel can be used to boost the approximate distillable entanglement of shared bipartite state~\cite{WH10}; however, it is still unclear if an erasure channel could boost the one-way exact distillable entanglement or the probabilitic one-way distillable entanglement of the state. In general, it would be interesting to know if channels that cannot be used for zero-error private communication can be used to boost the probabilistic distillable key of a shared bipartite state.

\begin{acknowledgments}
We thank Nilanjana Datta, Tharon Holdsworth, Kaiyuan Ji, Hemant Mishra, Theshani Nuradha, Dhrumil Patel, Aby Philip, and Xin Wang for insightful discussions. We also acknowledge support from the School of Electrical and Computer Engineering at Cornell University and the National Science Foundation under Grant No.~2315398.
\end{acknowledgments}

\bibliography{Ref}

\appendix

\section{Proof of Proposition~\ref{prop:geo_unext_ent_st_conv_Bel_Stas}}\label{app:geo_unext_ent_st_conv_Bel_Stas}

In this appendix, we show that the $\alpha$-geometric unextendible entanglement of quantum states converges to the unextendible entanglement of quantum states induced by the Belavkin--Staszewski relative entropy as $\alpha \to 1$.

Let us first evaluate the $\alpha$-geometric unextendible entanglement of states when $\alpha$ approaches 1 from above. The $\alpha$-geometric unextendible entanglement is monotonic in $\alpha$ as is evident from the monotonicity of the underlying divergence, the $\alpha$-geometric R\'enyi relative entropy, in $\alpha$. Therefore,
\begin{align}
    & \lim_{\alpha\to 1^+}\widehat{E}^u_{\alpha}\!\left(\rho_{AB}\right) \notag  \\
    &= \inf_{\alpha\in (1,2]}\widehat{E}^u_{\alpha}\!\left(\rho_{AB}\right)\\
    &= \inf_{\alpha\in (1,2]} \frac{1}{2}\inf_{\sigma\in\operatorname{Ext}\left(\rho\right)}\widehat{D}_{\alpha}\!\left(\rho\Vert\operatorname{Tr}_{B_1}\!\left[\sigma_{AB_1B_2}\right]\right)\\
    &=\frac{1}{2}\inf_{\sigma\in\operatorname{Ext}\left(\rho\right)}\inf_{\alpha\in (1,2]}\widehat{D}_{\alpha}\!\left(\rho\Vert\operatorname{Tr}_{B_1}\!\left[\sigma_{AB_1B_2}\right]\right)\\
    &= \frac{1}{2}\inf_{\sigma\in\operatorname{Ext}\left(\rho\right)}\widehat{D}_{1}\!\left(\rho\Vert\operatorname{Tr}_{B_1}\!\left[\sigma_{AB_1B_2}\right]\right)\\
    &= \widehat{E}^u\!\left(\rho_{AB}\right)\label{eq:geo_unext_ent_st_conv_Bel_stas_up} .
\end{align}

Now let us evaluate the $\alpha$-geometric unextendible entanglement when $\alpha$ approaches 1 from below. Once again, using the monotonicity of $\alpha$-geometric unextendible entanglement in $\alpha$, we find the following equality:
\begin{align}
    & \lim_{\alpha\to 1^-}\widehat{E}^u_{\alpha}\!\left(\rho_{AB}\right) \notag \\ &= \sup_{\alpha\in (0,1)}\widehat{E}^u_{\alpha}\!\left(\rho_{AB}\right)\\
    &= \sup_{\alpha\in (0,1)}\frac{1}{2}\inf_{\sigma\in\operatorname{Ext}\left(\rho\right)}\widehat{D}_{\alpha}\!\left(\rho\Vert\operatorname{Tr}_{B_1}\!\left[\sigma_{AB_1B_2}\right]\right).
\end{align}
Since the $\alpha$-geometric R\'enyi relative entropy is lower semi-continuous in $\left(\rho,\sigma\right)$~\cite[Lemma~A.3]{Fawzi2021}, and it increases monotonically in $\alpha$ in the range $(0,2]$, we can employ the Mosonyi--Hiai minimax theorem from~\cite[Corollary~A.2]{MH11} to switch the order of supremum and infimum, obtaining the following equality:
\begin{align}
    & \lim_{\alpha\to 1^-}\widehat{E}^u_{\alpha}\!\left(\rho_{AB}\right) \notag \\ 
    &=\frac{1}{2}\inf_{\sigma\in\operatorname{Ext}\left(\rho\right)}\sup_{\alpha\in (0,1)}\widehat{D}_{\alpha}\!\left(\rho\Vert\operatorname{Tr}_{B_1}\!\left[\sigma_{AB_1B_2}\right]\right)\\
    &= \frac{1}{2}\inf_{\sigma\in\operatorname{Ext}\left(\rho\right)}\widehat{D}_{1}\!\left(\rho\Vert\operatorname{Tr}_{B_1}\!\left[\sigma_{AB_1B_2}\right]\right)\\
    &= \widehat{E}^u\left(\rho_{AB}\right).\label{eq:geo_unext_ent_st_conv_Bel_stas_low}
\end{align}
Combining~\eqref{eq:geo_unext_ent_st_conv_Bel_stas_up} and~\eqref{eq:geo_unext_ent_st_conv_Bel_stas_low}, we conclude~\eqref{eq:geo_unext_ent_st_conv_Bel_Stas}.

\section{Proof of Proposition~\ref{prop:geo_unext_ent_converge_Bel_Stas}}\label{app:geo_unext_ent_converg_Bel_Stas}

In this appendix, we show that the $\alpha$-geometric unextendible entanglement of quantum channels converges to the unextendible entanglement of point-to-point quantum channels induced by the Belavkin--Staszewski relative entropy as $\alpha \to 1$.

Let us first evaluate the limit when $\alpha$ approaches 1 from above. For a given quantum channel $\mathcal{N}_{A\to B}$, let us define the following set of quantum channels:
 \begin{equation}
        \operatorname{Ext}\!\left(\mathcal{N}\right) \coloneqq \left\{\mathcal{P}_{A\to B_1B_2}: \operatorname{Tr}_{B_2}\circ\mathcal{P}_{A\to B_1B_2} = \mathcal{N}_{A\to B}\right\}.
 \end{equation}
 We know that the unextendible entanglement of channels induced by the $\alpha$-geometric R\'enyi relative entropy increases monotonically with $\alpha$ for $\alpha > 0$~\cite{Katariya2021}. Therefore, we can write,
 \begin{align}
     &\lim_{\alpha\to 1^+} \widehat{E}^u_{\alpha}\!\left(\mathcal{N}_{A\to B}\right) \notag\\
     &= \inf_{\alpha\in (1,2]}\widehat{E}^u_{\alpha}\!\left(\mathcal{N}_{A\to B}\right)\\
     &= \inf_{\alpha\in (1,2]}\inf_{\mathcal{P}_{A\to B_1B_2} \in \operatorname{Ext}\left(\mathcal{N}\right)} \widehat{D}_{\alpha}\!\left(\mathcal{N}_{A\to B}\Vert\operatorname{Tr}_{B_1}\circ\mathcal{P}_{A\to B_1B_2}\right)\\
     &= \inf_{\mathcal{P}_{A\to B_1B_2} \in \operatorname{Ext}\left(\mathcal{N}\right)} \inf_{\alpha \in (1,2]} \widehat{D}_{\alpha}\!\left(\mathcal{N}_{A\to B}\Vert \operatorname{Tr}_{B_1}\circ\mathcal{P}_{A\to B_1B_2}\right)\\
     &= \inf_{\mathcal{P}_{A\to B_1B_2} \in \operatorname{Ext}\left(\mathcal{N}\right)} \widehat{D}\!\left(\mathcal{N}_{A\to B}\Vert \operatorname{Tr}_{B_1}\circ\mathcal{P}_{A\to B_1B_2}\right)\\
     &= \widehat{E}^u\!\left(\mathcal{N}_{A\to B}\right), \label{eq:geo_unext_ent_up_lim}
 \end{align}
 where the first equality is a consequence of monotonicity of $\alpha$-geometric unextendible entanglement when $\alpha \in (1,2]$ (Proposition \ref{prop:geo_unext_ent_monotonic}) and the penultimate equality is a consequence of the fact that the $\alpha$-geometric R\'enyi relative entropy of channels converges to the Belavkin--Staszewski relative entropy as $\alpha \to 1$~\cite[Lemma~35]{DKQSWW23}.

 Now let us evaluate the limit when $\alpha$ approaches 1 from below. By Proposition~\ref{prop:geo_unext_ent_monotonic}, we know that the $\alpha$-geometric unextendible entanglement increases monotonically for $\alpha\in (0,1)$. Therefore,
 \begin{align}
     &\lim_{\alpha \to 1^-}\widehat{E}^u_{\alpha}\!\left(\mathcal{N}_{A\to B}\right) \notag\\
     &= \sup_{\alpha \in (0,1)}\widehat{E}^u_{\alpha}\!\left(\mathcal{N}_{A\to B}\right)\\
     &= \sup_{\alpha \in (0,1)}\inf_{\mathcal{P}_{A\to B_1B_2}\in \operatorname{Ext}\left(\mathcal{N}\right)} \widehat{D}_{\alpha}\!\left(\mathcal{N}_{A\to B}\Vert \operatorname{Tr}_{B_1}\circ\mathcal{P}_{A\to B_1B_2}\right). \label{eq:geo_unext_ent_low_lim_1}
 \end{align}
Since the $\alpha$-geometric R\'enyi relative entropy of channels $\widehat{D}_{\alpha}\!\left(\mathcal{N}_{A\to B}\Vert\mathcal{M}_{A\to B}\right)$ is lower semi-continuous in $\mathcal{M}_{A\to B}$~\cite[Lemma~37]{DKQSWW23} and increases monotonically in $\alpha$  in the range $(0,1)$, we can employ the Mosonyi--Hiai minimax theorem from~\cite[Corollary~A.2]{MH11} and establish that
\begin{align}
 &\sup_{\alpha \in (0,1)}\inf_{\mathcal{P}_{A\to B_1B_2}\in \operatorname{Ext}\left(\mathcal{N}\right)} \widehat{D}_{\alpha}\!\left(\mathcal{N}_{A\to B}\Vert \operatorname{Tr}_{B_1}\circ\mathcal{P}_{A\to B_1B_2}\right)\notag\\
 &= \inf_{\mathcal{P}_{A\to B_1B_2}\in \operatorname{Ext}\left(\mathcal{N}\right)}\sup_{\alpha \in (0,1)} \widehat{D}_{\alpha}\!\left(\mathcal{N}_{A\to B}\Vert \operatorname{Tr}_{B_1}\circ\mathcal{P}_{A\to B_1B_2}\right)\\
 &= \inf_{\mathcal{P}_{A\to B_1B_2}\in \operatorname{Ext}\left(\mathcal{N}\right)} \widehat{D}\!\left(\mathcal{N}_{A\to B}\Vert\operatorname{Tr}_{B_1}\circ\mathcal{P}_{A\to B_1B_2}\right)\\
 &= \widehat{E}^u\!\left(\mathcal{N}_{A\to B}\right), \label{eq:geo_unext_ent_low_lim_2}
\end{align}
where the first equality follows from the Mosonyi--Hiai minimax theorem in~\cite[Corollary~A.2]{MH11}, the second equality follows from the fact that $\alpha$-geometric R\'enyi relative entropy converges to the Belavkin--Staszewski relative entropy as $\alpha\to 1$, and the final equality follows from the definition of unextendibe entanglement induced by the Belavkin--Staszewski relative entropy. Hence, combining~\eqref{eq:geo_unext_ent_up_lim},~\eqref{eq:geo_unext_ent_low_lim_1}, and~\eqref{eq:geo_unext_ent_low_lim_2}, we conclude~\eqref{eq:geo_unext_ent_converge_Bel_Stas}.

\section{Proof of Corollary~\ref{cor:0_err_cap_ub}}\label{app:0_err_cap_ub_proof}

In this appendix, we give an alternate proof of Corollary~\ref{cor:0_err_cap_ub}. Let us begin by evaluating the $\alpha$-geometric unextendible entanglement of a $d$-dimensional identity channel.

\begin{proposition}\label{prop:id_geo_unext_ent}
    The $\alpha$-geometric unextendible entanglement of a $d$-dimensional identity channel is equal to $\log_2 d$ for all $\alpha \in (0,2]$. That is,
    \begin{equation}
        \widehat{E}^u_{\alpha}\!\left(\operatorname{id}^d_{A\to B}\right) = \log_2 d \qquad \forall \alpha \in (1,2].
    \end{equation}
\end{proposition}

\begin{proof}
    Let $\mathcal{P}_{A \to B_1B_2}$ be an extension of the $d$-dimensional identity channel, i.e.,
\begin{equation}\label{eqn:id_extension}
	\operatorname{Tr}_{B_2}\circ\mathcal{P}_{A\to B_1B_2} = \operatorname{id}^d_{A\to B}.
\end{equation}
Let $\Gamma^{\mathcal{P}}_{AB_1B_2}$ be the Choi operator of the quantum channel $\mathcal{P}_{A\to B_1B_2}$. Since the Choi operator of the identity channel is the unnormalized maximally entangled state, the Choi operator of the channel $\mathcal{P}_{A\to B_1B_2}$ has the following form:
\begin{equation}
    \Gamma^{\mathcal{P}}_{AB_1B_2} = \Gamma_{AB_1}\otimes\sigma_{B_2},
\end{equation}
where $\Gamma_{AB_1}$ is the unnormalized maximally entangled state, and $\sigma_{B_2}$ is an arbitrary quantum state. Thus, an arbitrary extension, $\mathcal{P}_{A\to B_1B_2}$, of the identity channel can be expressed as
\begin{equation}
	\mathcal{P}_{A\to B_1B_2} = \operatorname{id}^d_{A\to B_1}\otimes\mathcal{A}^{\sigma}_{B_2},
\end{equation}
where $\mathcal{A}^{\sigma}_{B_2}$ is a channel that prepares the state $\sigma_{B_2}$.

The two marginals of $\mathcal{P}_{A\to B_1B_2}$ act on an arbitrary quantum state $\psi_{RA}$ as
\begin{equation}
    \operatorname{Tr}_{B_2}\circ\mathcal{P}_{A\to B_1B_2}\!\left(\psi_{RA}\right) = \operatorname{id}^d_{A\to B_1}\!\left(\psi_{RA}\right) = \psi_{RB_1} ,
\end{equation}
and
\begin{align}
    \operatorname{Tr}_{B_1}\circ\mathcal{P}_{A\to B_1B_2}\!\left(\psi_{RA}\right) & = \operatorname{Tr}_A\otimes\mathcal{A}^{\sigma}_{B_2}\!\left(\psi_{RA}\right)\notag\\
    & = \psi_R\otimes \sigma_{B_2}.
\end{align}
Now consider that
\begin{align}
    &\widehat{E}^u_{\alpha}\!\left(\operatorname{id}^d_{A\to B}\right)\notag \\
    & = \frac{1}{2} \inf_{\mathcal{P}_{A\to B_1 B_2} \in \operatorname{Ext}(\operatorname{id})} \widehat{D}_{\alpha}(\operatorname{id} \Vert \operatorname{Tr}_{B_1} \circ \mathcal{P}_{A\to B_1 B_2}) \\
    & = \frac{1}{2} \inf_{\sigma_{B_2}} \sup_{\psi_{RA}} \widehat{D}_\alpha\left(\psi_{RB_1}\Vert \psi_{R}\otimes\sigma_{B_2}\right)\\
 & = \frac{1}{2} \inf_{\sigma_{B_2} \in \mathcal{S}_+} \sup_{\psi_{RA}} \widehat{D}_\alpha\left(\psi_{RB_1}\Vert \psi_{R}\otimes\sigma_{B_2}\right),
\end{align}
where $\mathcal{S}_+$ denotes the set of positive definite states, and here we have used the fact that the geometric R\'enyi relative entropy is lower semi-continuous in the last equality.
As shown in~\cite[Appendix B]{WWW19}, 
\begin{equation}\label{eq:geo_min_sigma_inv}
	\inf_{\sigma_{B_2}\in \mathcal{S}_+} \widehat{D}_\alpha(\psi_{RB_1}\Vert \psi_{R}\otimes\sigma_{B_2}) = \inf_{\sigma_{B_2}\in \mathcal{S}_+} \log_2\operatorname{Tr}[\sigma_{B_2}^{-1}].
\end{equation}
The right hand side of the~\eqref{eq:geo_min_sigma_inv} is independent of $\psi_{RA}$. Therefore,
\begin{equation}
	\widehat{E}^u_{\alpha}\!\left(\operatorname{id}^d_{A\to B}\right) = \frac{1}{2}\inf_{\sigma_{B_2}\in \mathcal{S}_+} \log_2\operatorname{Tr}[\sigma_{B_2}^{-1}].
\end{equation}

Let $\{\lambda_i\}_i$ be the eigenvalues of $\sigma_{B_2}$, so that this state can be written as
\begin{equation}
	\sigma_{B_2} = \sum_{i=1}^d \lambda_i |i\rangle\!\langle i|,
\end{equation}
where $\{|i\rangle\}_{i=1}^d$ is an eigenbasis of  $\sigma_{B_2}$. The inverse of this state is
\begin{equation}
	\sigma_{B_2}^{-1} = \sum_{i=1}^d \lambda_i^{-1}|i\rangle\!\langle i|.
\end{equation}
The $\alpha$-geometric unextendible entanglement of the identity channel becomes
\begin{equation}
	\widehat{E}^u_{\alpha}\!\left(\operatorname{id}^d_{A\to B}\right) = \frac{1}{2}\log_2\inf_{\substack{\{\lambda_i\}_i\\ \sum_i \lambda_i = 1}} \sum_{i=1}^d \lambda^{-1}_i,
\end{equation}
where the infimum is over every probability distribution $\{\lambda_i\}_i$ with full support.
Using the well known arithmetic mean-harmonic mean inequality
\begin{equation}
	\frac{d}{\sum_{i}\lambda_i^{-1}} \le \frac{\sum_i \lambda_i}{d}
	= \frac{1}{d}.
\end{equation}
Rearranging this inequality then implies that
\begin{equation}
    \frac{1}{2}\log_2\inf_{\substack{\{\lambda_i\}_i\\ \sum_i \lambda_i = 1}} \sum_{i=1}^d \lambda^{-1}_i \geq \frac{1}{2}\log_2 d^2.
\end{equation}
The inequality above is saturated when all $\lambda_i$ are equal. Therefore,
\begin{equation}
    \widehat{E}^u_{\alpha}\!\left(\operatorname{id}^d_{A\to B}\right) = \frac{1}{2}\log_2 d^2 = \log_2 d.
\end{equation}
This concludes the proof of Proposition~\ref{prop:id_geo_unext_ent}.
\end{proof}

\medskip

Now we show that the zero-error quantum capacity of a channel, assisted by one-way LOCC superchannels or two-extendible superchannels, is bounded from above by the min-geometric unextendible entanglement of the channel.

\medskip

\begin{proof}[Proof of Corollary~\ref{cor:0_err_cap_ub}]
    Consider a two-extendible superchannel $\Theta_{(A^n\to B^n)\to (C\to D)}$ that acts on $n$ instances of the channel $\mathcal{N}_{A\to B}$ to exactly simulate the $d$-dimensional identity channel:
\begin{equation}
    \operatorname{id}^d_{C\to D} = \Theta_{(A^n\to B^n)\to (C\to D)}\!\left(\mathcal{N}_{A\to B}^{\otimes n}\right).
\end{equation}
Equating the $\alpha$-geometric unextendible entanglement of the identity channel and the simulated channel, the following inequality holds for all $\alpha \in (0,2]$:
\begin{align}
    \widehat{E}^u_{\alpha}\!\left(\operatorname{id}^d_{C\to D}\right) &= \widehat{E}^u_{\alpha}\!\left(\Theta_{(A^n\to B^n)\to (C\to D)}\!\left(\mathcal{N}_{A\to B}^{\otimes n}\right)\right)\\
    &\le \widehat{E}^u_{\alpha}\!\left(\mathcal{N}^{\otimes n}_{A\to B}\right)\\
    &\le n\widehat{E}^u_{\alpha}\!\left(\mathcal{N}_{A\to B}\right),\label{eq:unext_ent_id_le_n_unext_ent_ch}
\end{align}
where the first inequality comes from the monotonicity of the unextendible entanglement under the action of two-extendible superchannels (Theorem~\ref{theo:two_ext_monotonic_p2p}) and the second inequality comes from the subadditivity of the $\alpha$-geometric unextendible entanglement (Proposition~\ref{prop:unext_ent_subadditive}).

Using~\eqref{eq:unext_ent_id_le_n_unext_ent_ch} and Proposition~\ref{prop:id_geo_unext_ent}, we arrive at the following inequality: 
\begin{equation}\label{eq:rate_le_alpha_unext_ent_ch}
    \widehat{E}^u_{\alpha}\!\left(\mathcal{N}_{A\to B}\right) \ge \frac{\log_2d}{n} = \frac{1}{n}Q^{(1)}_{0,\operatorname{2-EXT}}\!\left(\mathcal{N}^{\otimes n}_{A\to B}\right),
\end{equation}
where we arrive at the equality after recalling the definition of $Q^{(1)}_{0,\operatorname{2-EXT}}\!\left(\mathcal{N}_{A\to B}\right)$ from~\eqref{eq:1shot_q_cap_2_ext_defn}.
Since the inequality holds for all $\alpha \in (0,2]$, we can take the limit $\alpha \to 0$ to get the tightest inequality in~\eqref{eq:rate_le_alpha_unext_ent_ch}, owing to the monotonicity of the $\alpha$-geometric unextendible entanglement in $\alpha$. Moreover,~\eqref{eq:rate_le_alpha_unext_ent_ch} holds for every positive integer $n$, which leads to the following inequality:
\begin{equation}\label{eq:0_err_cap_non_asymp_ub}
    Q_{0,\operatorname{2-EXT}}\!\left(\mathcal{N}_{A\to B}\right) = \frac{1}{n}Q^{(1)}_{0,\operatorname{2-EXT}}\!\left(\mathcal{N}^{\otimes n}_{A\to B}\right) \le \widehat{E}^u_{\operatorname{min}}\!\left(\mathcal{N}_{A\to B}\right).
\end{equation}
The zero-error quantum capacity of a channel assisted by two-extendible superchannels is never less than the zero-error quantum capacity of the channel assisted by one-way LOCC superchannels, which concludes the proof.
\end{proof}

\section{Proof of Proposition~\ref{prop:bip_such_ch_marg_theo}}\label{app:bip_supch_ch_marg_theo_proof}

Consider an arbitrary quantum channel $\mathcal{P}_{AB_{[k]}\to A'B'_{[k]}}$ where all systems in the set $\left\{B_i\right\}_{i=1}^k$ are isomorphic to each other, and all systems in the set $\left\{B'_i\right\}_{i=1}^{k}$ are isomorphic to each other. Let $\Theta_{(AB\to A'B')\to (CD\to C'D')}$ be a $k$-extendible superchannel with the $k$-extension $\Upsilon_{(AB_{[k]}\to A'B'_{[k]})\to (CD_{[k]}\to C'D'_{[k]})}$. Let $\mathcal{N}^i_{AB_i\to A'B'_i}$ be a marginal of the channel $\mathcal{P}_{AB_{[k]}\to A'B'_{[k]}}$ such that,
    \begin{equation}
        \mathcal{N}^i_{AB_i\to A'B'_i} = \operatorname{Tr}_{B'_{[k]\setminus i}}\circ\mathcal{P}_{AB_{[k]}\to A'B'_{[k]}}\circ\mathcal{A}_{B_{[k]\setminus i}},
    \end{equation}
    
    The Choi operator of the channel $\Upsilon\left(\mathcal{P}\right)$, using the propagation rule stated in~\eqref{eq:prop_rule}, is,
    \begin{equation}
        \Gamma^{\Upsilon\left(\mathcal{P}\right)} = \operatorname{Tr}_{AB_{[k]}A'B'_{[k]}}\!\left[\left(\Gamma^{\mathcal{P}}\right)^T\Gamma^{\Upsilon}\right],
    \end{equation}
    where $\Gamma^{\mathcal{P}}_{AB_{[k]}A'B'_{[k]}}$ is the Choi operator of the quantum channel $\mathcal{P}_{AB_{[k]}\to A'B'_{[k]}}$, and $\Gamma^{\Upsilon}_{AB_{[k]}A'B'_{[k]}CD_{[k]}C'D'_{[k]}}$ is the Choi operator of the superchannel $\Upsilon_{(AB_{[k]}\to A'B'_{[k]})\to (CD_{[k]}\to C'D'_{[k]})}$. The Choi operator of the channel $\mathcal{N}^i_{AB_i\to A'B'_i}$ is related to the Choi operator of the channel $\mathcal{P}_{AB_{[k]}\to A'B'_{[k]}}$ as
    \begin{equation}\label{eq:ch_marg_choi_proof}
        \Gamma^{\mathcal{N}_i}_{AB_iA'B'_i}\otimes I_{B_{[k]\setminus i}} = \operatorname{Tr}_{B'_{[k]\setminus i}}\!\left[\Gamma^{\mathcal{P}}_{AB_{[k]}A'B'_{[k]}}\right], 
    \end{equation}
    since the former is a marginal of the latter.
    Now consider the Choi operator of the channel $\operatorname{Tr}_{D'_{[k]\setminus i}}\circ\left(\Upsilon\left(\mathcal{P}\right)\right)$,
    \begin{align}
        &\operatorname{Tr}_{D'_{[k]\setminus i}}\!\left[\Gamma^{\Upsilon\left(\mathcal{P}\right)}\right]\notag\\
        &= \operatorname{Tr}_{AB_{[k]}A'B'_{[k]}D'_{[k]\setminus i}}\!\left[\left(\Gamma^{\mathcal{P}}\right)^T\Gamma^{\Upsilon}\right]\\
        &= \operatorname{Tr}_{AB_{[k]}A'B'_{[k]}}\!\left[\left(\Gamma^{\mathcal{P}}\right)^T\operatorname{Tr}_{D'_{[k]\setminus i}}\!\left[\Gamma^{\Upsilon}\right]\right]\\
        &= \frac{1}{\left|B'\right|^{k-1}}\operatorname{Tr}_{AB_{[k]}A'B'_{[k]}}\!\left[\left(\Gamma^{\mathcal{P}}\right)^T\operatorname{Tr}_{D'_{[k]\setminus i}B'_{[k]\setminus i}}\!\left[\Gamma^{\Upsilon}\right] \otimes I_{B'_{[k]\setminus i}} \right]\\
        &= \frac{1}{\left|B'\right|^{k-1}}\operatorname{Tr}_{AB_{[k]}A'B'_i}\!\left[\left(\operatorname{Tr}_{B'_{[k]\setminus i}}\!\left[\Gamma^{\mathcal{P}}\right]\right)^T\operatorname{Tr}_{D'_{[k]\setminus i}B'_{[k]\setminus i}}\!\left[\Gamma^{\Upsilon}\right]\right]\\
        &= \frac{1}{\left|B'\right|^{k-1}}\operatorname{Tr}_{AB_{[k]}A'B'_i}\!\left[\left(\Gamma^{\mathcal{N}_i}\otimes I_{B_{[k]\setminus i}}\right)^T\operatorname{Tr}_{D'_{[k]\setminus i}B'_{[k]\setminus i}}\!\left[\Gamma^{\Upsilon}\right]\right]\\
        &= \frac{1}{\left|B'\right|^{k-1}}\operatorname{Tr}_{AB_{i}A'B'_i}\!\left[\left(\Gamma^{\mathcal{N}_i}\right)^T\operatorname{Tr}_{B_{[k]\setminus i}D'_{[k]\setminus i}B'_{[k]\setminus i}}\!\left[\Gamma^{\Upsilon}\right]\right]\\
        &= \frac{1}{\left|B'\right|^{k-1}}\operatorname{Tr}_{AB_iA'B'_i}\!\left[\left(\Gamma^{\mathcal{N}_i}\right)^T\operatorname{Tr}_{B'_{[k]\setminus i}}\!\left[\Gamma^{\Theta}\otimes I_{B'_{[k]\setminus i}D_{[k]\setminus i}}\right]\right]\\
        &= \operatorname{Tr}_{AB_iA'B'_i}\!\left[\left(\Gamma^{\mathcal{N}_i}\right)^T\Gamma^{\Theta}\right]\otimes I_{D_{[k]\setminus i}}\\
        &= \Gamma^{\Theta\left(\mathcal{N}_i\right)}\otimes I_{D_{[k]\setminus i}},
    \end{align}
    where the third equality is a consequence of~\eqref{eq:k_ext_bip_supch_non_sig_choi}, the fifth equality is a consequence of~\eqref{eq:ch_marg_choi_proof}, the seventh equality is a consequence of~\eqref{eq:k_ext_bip_supch_marg_choi} and the final equality is arrived at by using the propagation rule again. Since the above equalities are true for all $i \in [k]$, we conclude~\eqref{eq:bip_supch_ch_marg_theo}.

\section{Proof of Proposition~\ref{prop:bip_supch_preserve_ext}}\label{app:bip_supch_preserve_ext}

In this appendix, we show that the action of a bipartite $k$-extendible superchannel on a bipartite $k$-extendible quantum channel results in a $k$-extendible quantum channel. 

Let $\mathcal{N}_{AB\to A'B'}$ be a $k$-extendible quantum channel with a $k$-extension $\mathcal{P}_{AB_{[k]}\to A'B'_{[k]}}$, and $\Theta_{(AB\to A'B')\to (CD\to C'D')}$ be a $k$-extendible superchannel with a $k$-extension $\Upsilon_{(AB_{[k]}\to A'B'_{[k]})\to (CD_{[k]}\to C'D'_{[k]})}$. Let $\Gamma^{\mathcal{N}}_{ABA'B'}$ and $\Gamma^{\mathcal{P}}_{AB_{[k]}A'B'_{[k]}}$ be the respective Choi operators of the channels $\mathcal{N}_{AB\to A'B'}$ and $\mathcal{P}_{AB_{[k]}\to A'B'_{[k]}}$, and let $\Gamma^{\Theta}_{ABA'B'CDC'D'}$ and $\Gamma^{\Upsilon}_{AB_{[k]}A'B'_{[k]}CD_{[k]}C'D'_{[k]}}$ be the respective Choi operators of the superchannels $\Theta_{(AB\to A'B')\to (CD\to C'D')}$ and $\Upsilon_{(AB_{[k]}\to A'B'_{[k]})\to (CD_{[k]}\to C'D'_{[k]})}$.

    The Choi operator of the channel $\Theta\left(\mathcal{N}\right)$ can be evaluated using the propagation rule, stated in~\eqref{eq:prop_rule}, as
    \begin{equation}
        \Gamma^{\Theta\left(\mathcal{N}\right)}_{CDC'D'} = \operatorname{Tr}_{ABA'B'}\!\left[\left(\Gamma^{\mathcal{N}}_{ABA'B'}\right)^T\Gamma^{\Theta}_{ABA'B'CDC'D'}\right].
    \end{equation}
    Consider the Choi operator of the quantum channel $\Upsilon\left(\mathcal{P}\right)$,
    \begin{equation}
        \Gamma^{\Upsilon\left(\mathcal{P}\right)}_{CD_{[k]}C'D'_{[k]}} = \operatorname{Tr}_{AB_{[k]}A'B'_{[k]}}\!\left[\left(\Gamma^{\mathcal{P}}\right)^T\Gamma^{\Upsilon}\right].
    \end{equation}
    Let us first show that $\Theta\left(\mathcal{N}\right)$ is a marginal of the channel $\Upsilon\left(\mathcal{P}\right)$. The non-signaling condition for $\mathcal{P}_{AB_{[k]}\to A'B'_{[k]}}$ and $\Upsilon_{(AB_{[k]}\to A'B'_{[k]})\to (CD_{[k]}\to C'D'_{[k]})}$ from~\eqref{eq:no_signal_bip_choi} and~\eqref{eq:k_ext_bip_supch_non_sig_choi}, respectively, implies
    \begin{align}
        &\operatorname{Tr}_{D'_{[k]\setminus 1}}\!\left[\Gamma^{\Upsilon\left(\mathcal{P}\right)}\right]\notag \\
        &= \operatorname{Tr}_{AB_{[k]}A'B'_{[k]}D'_{[k]\setminus 1}}\!\left[\left(\Gamma^{\mathcal{P}}\right)^T\Gamma^{\Upsilon}\right]\\
        &= \operatorname{Tr}_{AB_{[k]}A'B'_{[k]}}\!\left[\left(\Gamma^{\mathcal{P}}\right)^T\operatorname{Tr}_{D'_{[k]\setminus 1}}\!\left[\Gamma^{\Upsilon}\right]\right]\\
        &= \operatorname{Tr}_{AB_{[k]}A'B'_{[k]}}\!\left[\left(\Gamma^{\mathcal{P}}\right)^T\left(\operatorname{Tr}_{D'_{[k]\setminus 1}B'_{[k]\setminus 1}}\!\left[\Gamma^{\Upsilon}\right]\otimes\frac{I_{B'_{[k]\setminus 1}}}{\left|B'\right|^{k-1}}\right) \right]\\
        &= \frac{1}{\left|B'\right|^{k-1}}\operatorname{Tr}_{AB_{[k]}A'B'_{1}}\!\left[\left(\operatorname{Tr}_{B'_{[k]\setminus 1}}\Gamma^{\mathcal{P}}\right)^T\operatorname{Tr}_{D'_{[k]\setminus 1}B'_{[k]\setminus 1}}\!\left[\Gamma^{\Upsilon}\right]\right]\\
        &= \frac{1}{\left|B'\right|^{k-1}}\operatorname{Tr}_{AB_{[k]}A'B'_{1}}\!\left[\left(\Gamma^{\mathcal{N}}\otimes I_{B_{[k]\setminus 1}}\right)^T\operatorname{Tr}_{D'_{[k]\setminus 1}B'_{[k]\setminus 1}}\!\left[\Gamma^{\Upsilon}\right]\right]\\
        &= \frac{1}{\left|B'\right|^{k-1}}\operatorname{Tr}_{AB_{1}A'B'_{1}}\!\left[\left(\Gamma^{\mathcal{N}}\right)^T\operatorname{Tr}_{D'_{[k]\setminus 1}B'_{[k]\setminus 1}B_{[k]\setminus 1}}\!\left[\Gamma^{\Upsilon}\right]\right]\\
        &= \operatorname{Tr}_{AB_{1}A'B'_{1}}\!\left[\left(\Gamma^{\mathcal{N}}\right)^T\left(\Gamma^{\Theta}\otimes I_{D_{[k]\setminus 1}}\right)\right]\\
        &= \Gamma^{\Theta\left(\mathcal{N}\right)}_{CD_1C'D'_1}\otimes I_{D_{[k]\setminus 1}},\label{eq:mod_k_ext_ch_marg_eq}
    \end{align}
    where the first equality follows from the propagation rule stated in~\eqref{eq:prop_rule}, the third equality follows from the non-signaling condition for $k$-extensions of sueprchannels given in~\eqref{eq:k_ext_bip_supch_non_sig_choi}, the fifth equality follows from the non-signaling condition for $k$-extensions of quantum channels given in~\eqref{eq:no_signal_bip_choi}, the penultimate equality follows from the marginality condition for bipartite $k$-extendible superchannels given in~\eqref{eq:k_ext_bip_supch_marg_choi}, and the final equality follows once again by using the propagation rule. Equation~\eqref{eq:mod_k_ext_ch_marg_eq} implies that $\Theta\left(\mathcal{N}\right)$ is a marginal of the channel $\Upsilon\left(\mathcal{P}\right)$ which follows the non-signaling condition given in~\eqref{eq:no_signal_bip_choi}.

    Now let us test the quantum channel $\Upsilon\left(\mathcal{P}\right)$ for permutation covariance. Since $\mathcal{P}_{AB_{[k]}A'B'_{[k]}}$ is a $k$-extension of the quantum channel $\mathcal{N}_{AB\to A'B'}$, it obeys the permutation covariance condition given in~\eqref{eq:perm_cov_bipartite}; that is,
    \begin{equation}
        \left(W^{\pi}_{B_{[k]}}\otimes W^{\pi}_{B'_{[k]}}\right)\Gamma^{\mathcal{P}}\!\left(W^{\pi\dagger}_{B_{[k]}}\otimes W^{\pi\dagger}_{B'_{[k]}}\right) = \Gamma^{\mathcal{P}},
    \end{equation}
    where $W^{\pi}$ is the unitary corresponding to the permutation $\pi$ in the symmetric group $S_k$. Let us define the following unitary:
    \begin{equation}
        U^{\pi}_{B_{[k]}B'_{[k]}} \coloneqq W^{\pi}_{B_{[k]}}\otimes W^{\pi}_{B'_{[k]}},     
    \end{equation}
    so that the permutation condition for $\mathcal{P}_{AB_{[k]}\to A'B'_{[k]}}$ can be written as,
    \begin{equation}
        \Gamma^{\mathcal{P}} = U^{\pi}_{B_{[k]}B'_{[k]}}\Gamma^{\mathcal{P}}U^{\pi\dagger}_{B_{[k]}B'_{[k]}}.
    \end{equation}
    Let us also define a unitary $V^{\pi}$ in a similar fashion,
    \begin{equation}
        V^{\pi}_{D'_{[k]}B_{[k]}D_{[k]}B'_{[k]}} \coloneqq W^{\pi}_{D'_{[k]}}\otimes W^{\pi}_{B_{[k]}}\otimes W^{\pi}_{D_{[k]}}\otimes W^{\pi}_{B'_{[k]}}.
    \end{equation}
    The permutation covariance condition for $k$-extendible superchannels given in~\eqref{eq:perm_cov_bip_supch_choi} can then be written as,
    \begin{equation}
        \Gamma^{\Upsilon} = V^{\pi}_{D'_{[k]}B_{[k]}D_{[k]}B'_{[k]}}\Gamma^{\Upsilon}V^{\pi\dagger}_{D'_{[k]}B_{[k]}D_{[k]}B'_{[k]}}.
    \end{equation}
Note that $\left(W^{\pi}\right)^T = W^{\pi\dagger}$ since the permutation unitary is completely real. This implies the following equality:
\begin{align}
    \left(\Gamma^{\mathcal{P}}\right)^T\Gamma^{\Upsilon} &=  \left(U^{\pi}_{B_{[k]}B'_{[k]}}\Gamma^{\mathcal{P}}U^{\pi\dagger}_{B_{[k]}B'_{[k]}}\right)^TV^{\pi}\Gamma^{\Upsilon}V^{\pi\dagger}\\
    &= U^{\pi}_{B_{[k]}B'_{[k]}}\!\left(\Gamma^{\mathcal{P}}\right)^TU^{\pi\dagger}_{B_{[k]}B'_{[k]}}V^{\pi}\Gamma^{\Upsilon}V^{\pi\dagger}\\
    &= U^{\pi}_{B_{[k]}B'_{[k]}}\!\left(\Gamma^{\mathcal{P}}\right)^TU^{\pi\dagger}_{D_{[k]}D'_{[k]}}\Gamma^{\Upsilon}V^{\pi\dagger}.
\end{align}
Using the above equality in the expression for the Choi operator of the channel $\Upsilon\left(\mathcal{P}\right)$,
\begin{align}
    &\Gamma^{\Upsilon\left(\mathcal{N}\right)}_{CD_{[k]}C'D'_{[k]}}\notag\\
    &= \operatorname{Tr}_{AB_{[k]}A'B'_{[k]}}\!\left[\left(\Gamma^{\mathcal{P}}\right)^T\Gamma^{\Upsilon}\right] \\
    &= \operatorname{Tr}_{AB_{[k]}A'B'_{[k]}}\!\left[U^{\pi}_{B_{[k]}B'_{[k]}}\!\left(\Gamma^{\mathcal{P}}\right)^TU^{\pi\dagger}_{D_{[k]}D'_{[k]}}\Gamma^{\Upsilon}V^{\pi\dagger}\right]\\
    &= \operatorname{Tr}_{AB_{[k]}A'B'_{[k]}}\!\left[\left(\Gamma^{\mathcal{P}}\right)^TU^{\pi\dagger}_{D_{[k]}D'_{[k]}}\Gamma^{\Upsilon}V^{\pi\dagger}U^{\pi}_{B_{[k]}B'_{[k]}}\right]\\
    &= \operatorname{Tr}_{AB_{[k]}A'B'_{[k]}}\!\left[\left(\Gamma^{\mathcal{P}}\right)^TU^{\pi\dagger}_{D_{[k]}D'_{[k]}}\Gamma^{\Upsilon}U^{\pi}_{D_{[k]}D'_{[k]}}\right]\\
    &= U^{\pi\dagger}_{D_{[k]}D'_{[k]}}\!\left(\operatorname{Tr}_{AB_{[k]}A'B'_{[k]}}\!\left[\left(\Gamma^{\mathcal{P}}\right)^T\Gamma^{\Upsilon}\right]\right)U^{\pi}_{D_{[k]}D'_{[k]}}\\
    &= U^{\pi\dagger}_{D_{[k]}D'_{[k]}}\!\left(\Gamma^{\Upsilon\left(\mathcal{N}\right)}_{CD_{[k]}C'D'_{[k]}}\right)U^{\pi}_{D_{[k]}D'_{[k]}},
\end{align}
where we have used the cyclicity of trace to arrive at the third equality. Thus, the quantum channel $\Upsilon\!\left(\mathcal{P}\right)$ follows the permutation covariance condition given in~\eqref{eq:perm_cov_bipartite}.

Since $\Upsilon\!\left(\mathcal{P}\right)$ is an extension of the quantum channel $\Theta\!\left(\mathcal{N}\right)$ that follows the permutation covariance condition given in~\eqref{eq:perm_cov_bipartite}, and the non-signaling condition given in~\eqref{eq:no_signal_bipartite}, we conclude that $\Theta\!\left(\mathcal{N}\right)$ is a $k$-extendible channel.

\section{Proof of Proposition~\ref{prop:1WL_supch_k_ext}}\label{app:1WL_supch_k_ext}

In this section, we show that every bipartite one-way LOCC superchannel is $k$-extendible for all $k\ge 2$ by constructing an explicit extension of an arbitrary bipartite one-way LOCC superchannel.

In a bipartite one-way LOCC superchannel $\Theta_{(AB\to A'B')\to (CD\to C'D')}$, both the pre-processing and post-processing channels are bipartite one-way LOCC channels. Therefore, we can write the pre-processing channel as follows:
    \begin{equation}
        \mathcal{E}^{\Theta}_{CD\to AM_ABM_B} \coloneqq \sum_x \mathcal{E}^{A,x}_{C\to AM_A}\otimes \mathcal{E}^{B,x}_{D\to BM_B},
    \end{equation}
    where systems $C$, $A$ and $M_A$ are held by Alice, and systems $D$, $B$ and $M_B$ are held by Bob. In the above, $\left\{\mathcal{E}^{A,x}_{C\to AM_A}\right\}_x$ is a set of CP maps and $\left\{\mathcal{E}^{B,x}_{D\to BM_B}\right\}_x$ is a set of quantum channels such that $\mathcal{E}^{\Theta}_{CD\to AM_ABM_B}$ is a quantum channel. Similarly, we can write the post-processing channel as follows:
    \begin{equation}
        \mathcal{D}^{\Theta}_{A'M_AB'M_B\to C'D'} \coloneqq \sum_y \mathcal{D}^{A,y}_{A'M_A\to C'}\otimes \mathcal{D}^{B,y}_{B'M_B\to D'},
    \end{equation}
    where systems $A'$, $C'$ and $M_A$ are held by Alice, and systems $D'$, $B'$ and $M_B$ are held by Bob. The set $\left\{\mathcal{D}^{A,y}_{A'M_A\to C'}\right\}_y$ is a set of CP maps and $\left\{\mathcal{D}^{B,y}_{B'M_B\to D'}\right\}_y$ is a set of quantum channels such that $\mathcal{D}^{\Theta}_{A'M_AB'M_B\to C'D'}$ is a quantum channel.

    One can define a superchannel $\Upsilon_{(AB_{[k]}\to A'B'_{[k]})\to (CD_{[k]}\to C'D'_{[k]})}$ with the following pre-processing channel:
    \begin{multline}
        \mathcal{E}^{\Upsilon}_{CD_{[k]}\to AM_AB_{[k]}M_{B_{[k]}}} \coloneqq\\ \sum_x \mathcal{E}^{A,x}_{C\to AM_A}\otimes \mathcal{E}^{B,x}_{D_1\to B_1M_{B_1}}\otimes\cdots\otimes\mathcal{E}^{B,x}_{D_k\to B_kM_{B_k}}.
    \end{multline}
    The post-processing channel associated with the superchannel $\Upsilon$ can be defined as follows:
    \begin{multline}
        \mathcal{D}^{\Upsilon}_{A'M_AB'_{[k]}M_{B_{[k]}}\to C'D'_{[k]}} \coloneqq\\ \sum_y \mathcal{D}^{A,y}_{A'M_A\to C'}\otimes \mathcal{D}^{B,y}_{B'_1M_{B_1}\to D'_1}\otimes\cdots\otimes\mathcal{D}^{B,y}_{B'_kM_{B_k}\to D'_k}.
    \end{multline}
    It is straightforward to verify that the superchannel $\Theta$ is a marginal of the superchannel $\Upsilon$, and the quantum channels $\mathcal{Q}^{\Theta}$ and $\mathcal{Q}^{\Upsilon}$ unique to the superchannels $\Theta$ and $\Upsilon$ respectively, follow the conditions given in~\eqref{eq:two_ext_bip_supch_perm_cov},~\eqref{eq:k_ext_bip_supch_non_sig}, and~\eqref{eq:k_ext_bip_supch_marg_cond}. Therefore, $\Upsilon$ is a valid $k$-extension of the superchannel $\Theta$. Since such a two-extension can be constructed for every one-way LOCC superchannel $\Theta$, we conclude that every one-way LOCC superchannel is $k$-extendible for all $k\ge 2$.

\section{Proof of Theorem~\ref{theo:two_ext_monotonic_bip}}\label{app:two_ext_monotonic_bip}

In this appendix, we show that the unextendible entanglement of bipartite channels decreases under the action of two-extendible superchannels.

Let $\Theta_{(AB\to A'B')\to (CD\to C'D')}$ be an arbitrary two-extendible superchannel. Let $\mathcal{P}_{AB_1B_2\to A'B'_1B'_2}$ be an arbitrary extension of the channel $\mathcal{N}_{AB\to A'B'}$ such that
\begin{equation}
	\mathcal{N}_{AB\to A'B'} = \operatorname{Tr}_{B'_2}\circ\mathcal{P}_{AB_1B_2\to A'B'_1B'_2}\circ\mathcal{A}_{B_2},
\end{equation}
where $\mathcal{A}_{B_2}$ is a quantum channel that appends an arbitrary quantum state to the system $B_2$. The generalized divergernce between the channels $\mathcal{N}_{AB\to A'B'}$ and $\operatorname{Tr}_{B'_1}\circ\mathcal{P}_{AB_1B_2\to A'B'_1B'_2}\circ\mathcal{A}_{B_1}$ obeys the following inequality:
\begin{align}
    &\mathbf{D}\!\left(\mathcal{N}\big\Vert\operatorname{Tr}_{B'_1}\circ\mathcal{P}\circ\mathcal{A}_{B_1}\right)\notag \\
    &= \mathbf{D}\!\left(\operatorname{Tr}_{B'_2}\circ\mathcal{P}\circ\mathcal{A}_{B_2}\big\Vert\operatorname{Tr}_{B'_1}\circ\mathcal{P}\circ\mathcal{A}_{B_1}\right)\\
    &\ge \mathbf{D}\!\left(\Theta\left(\operatorname{Tr}_{B'_2}\circ\mathcal{P}\circ\mathcal{A}_{B_2}\right)\big\Vert\Theta\!\left(\operatorname{Tr}_{B'_1}\circ\mathcal{P}\circ\mathcal{A}_{B_1}\right)\right)\\
    &= \mathbf{D}\!\left(\operatorname{Tr}_{D'_2}\circ\left(\Upsilon\!\left(\mathcal{P}\right)\right)\circ\mathcal{A}_{D_2}\big\Vert\operatorname{Tr}_{D'_1}\circ\left(\Upsilon\!\left(\mathcal{P}\right)\right)\circ\mathcal{A}_{D_1}\right)\\
    &\ge 2\mathbf{E}^u\!\left(\operatorname{Tr}_{D'_2}\circ\left(\Upsilon\!\left(\mathcal{P}\right)\right)\circ\mathcal{A}_{D_2}\right)\\
    &= 2\mathbf{E}^u\!\left(\Theta\left(\operatorname{Tr}_{B'_2}\circ\mathcal{P}\circ\mathcal{A}_{B_2}\right)\right)\\
    &= 2\mathbf{E}^u\!\left(\Theta\left(\mathcal{N}\right)\right),
\end{align}
where the first inequality follows from the data-processing inequality for generalized channel divergence of quantum channels (Theorem~\ref{theo:gen_div_channel_data_proc}), the second equality follows from Proposition~\ref{prop:bip_such_ch_marg_theo}, and the second inequality follows from the definition of the generalized unextendible entanglement of quantum channels. Since the above inequality holds for all quantum channels $\mathcal{P}_{AB_1B_2\to A'B'_1B'_2}$ that lie in the set $\operatorname{Ext}\!\left(\mathcal{N}\right)$ defined in~\eqref{eq:ext_set_bip_ch}, we conclude the statement of Theorem~\ref{theo:two_ext_monotonic_bip}.

\section{Proof of Theorem~\ref{theo:unext_ent_state_le_unext_ent_ch_bip}}\label{app:unext_ent_state_le_unext_ent_ch_bip}

In this appendix we present the proof of Theorem~\ref{theo:unext_ent_state_le_unext_ent_ch_bip}. 

Let $\rho_{R_CCR_DD}$ be an arbitrary two-extendible state. This means that there exists an extension $\tau_{R_CCR_{D_1}D_1R_{D_2}D_2}$ of the state $\rho_{R_CCR_DD}$ with the following marginals:
\begin{align}
    \operatorname{Tr}_{R_{D_1}D_1}\!\left[\tau\right] = \operatorname{Tr}_{R_{D_2}D_2}\!\left[\tau\right]
    = \rho_{R_CCR_DD},
\end{align}
where system $R_{D_1}$ is isomorphic to $R_{D_2}$ and system $D_1$ is isomorphic to $D_2$.

Let $\Theta_{(AB\to A'B')\to (CD\to C'C''D'D'')}$ be a two-extendible superchannel, and let $\mathcal{N}_{AB\to A'B'}$ be an arbitrary semicausal channel. Let us define the following quantum state:
\begin{equation}
    \sigma_{R_CC'C''R_DD'D''} \coloneqq \Theta\!\left(\mathcal{N}_{AB\to A'B'}\right)\!\left(\rho_{R_CCR_DD}\right).
\end{equation}
Let $\mathcal{P}_{AB_1B_2\to A'B'_1B'_2}$ be an arbitrary extension of the channel $\mathcal{N}_{AB\to A'B'}$; that is,
\begin{equation}
    \operatorname{Tr}_{B'_2}\circ\mathcal{P}_{AB_1B_2\to A'B'_1B'_2} = \mathcal{N}_{AB_1\to A'B'_1}\otimes\operatorname{Tr}_{B_2}.
\end{equation}
Proposition~\ref{prop:bip_such_ch_marg_theo} implies that there exists a superchannel $\Upsilon_{(AB_1B_2\to A'B'_1B'_2)\to (CD_1D_2\to C'D'_1D'_2)}$ such that the following equalities hold:
\begin{align}
    \operatorname{Tr}_{D'_2}\circ\left(\Upsilon\!\left(\mathcal{P}\right)\right) &= \Theta\!\left(\mathcal{N}\right)\otimes\operatorname{Tr}_{D_2},\\
    \operatorname{Tr}_{D'_1}\circ\left(\Upsilon\!\left(\mathcal{P}\right)\right) &= \Theta\!\left(\operatorname{Tr}_{B'_1}\otimes\mathcal{P}\right)\circ\operatorname{Tr}_{D_1}.\label{eq:marg_ups_P_to_theta_marg_P}
\end{align}

Consider the following state:
\begin{align}
    \operatorname{Tr}_{R_{D_2}D'_2}\circ\left(\Upsilon\!\left(\mathcal{P}\right)\right)\left(\tau\right) &= \left(\Theta\!\left(\mathcal{N}\right)\otimes\operatorname{Tr}_{R_{D_2}D_2}\right)\!\left(\tau\right)\\
    &= \Theta\!\left(\mathcal{N}\right)\!\left(\rho_{R_CCR_DD}\right)\\
    &= \sigma_{R_CC'C''R_DD'D''},
\end{align}
where the first equality follows from Proposition~\ref{prop:bip_such_ch_marg_theo}, the second equality follows from the fact that $\tau$ is a two-extension of $\rho$, and the final equality follows from the definition of the state~$\sigma$. The above equality implies that $\Upsilon\!\left(\mathcal{P}\right)\!\left(\tau\right)$ is an extension of the state $\sigma$. By definition of the unextendible entanglement of states, the following inequality holds:
\begin{equation}
    \mathbf{E}^u\!\left(\sigma\right) \le \inf_{\mathcal{P} \in \operatorname{Ext}\left(\mathcal{N}\right)} \frac{1}{2}\mathbf{D}\!\left(\sigma\big\Vert \operatorname{Tr}_{R_{D_1}D'_1}\circ\left(\Upsilon\left(\mathcal{P}\right)\right)\!\left(\tau\right)\right),
\end{equation}
where $\operatorname{Ext}\!\left(\mathcal{N}\right)$ is the set of all extensions of the channel $\mathcal{N}$ as defined in~\eqref{eq:ext_set_bip_ch}. Using~\eqref{eq:marg_ups_P_to_theta_marg_P} we arrive at the following equality:
\begin{align}
    \operatorname{Tr}_{R_{D_1}D'_1}\circ\left(\Upsilon\left(\mathcal{P}\right)\right)\!\left(\tau\right) &= \Theta\!\left(\operatorname{Tr}_{B'_1}\circ\mathcal{P}\right)\left(\operatorname{Tr}_{R_{D_1}D_1}\!\left[\tau\right]\right)\\
    &= \Theta\!\left(\operatorname{Tr}_{B'_1}\circ\mathcal{P}\right)\!\Big(\rho_{R_CCR_DD}\Big),
\end{align}
where the second equality follows from the fact that $\tau$ is a two-extension of $\rho$. As such, the following inequality holds for all two-extendible states $\rho_{R_CCR_DD}$:
\begin{equation}
    \mathbf{E}^u\!\left(\sigma\right) \le \inf_{\mathcal{P}\in \operatorname{Ext}(\mathcal{N})}\frac{1}{2}\mathbf{D}\!\left(\Theta\!\left(\mathcal{N}\right)\!\left(\rho\right)\Big\Vert \Theta\!\left(\operatorname{Tr}_{B'_1}\circ\mathcal{P}\right)\!\Big(\rho\Big)\right).
\end{equation}

For brevity, let us denote the set of two-extendible states with respect to the partition $R_CC:R_DD$ as $\operatorname{2-EXT_{RCD}}$ and the set of all states on systems $R_CCR_DD$ as $\mathcal{S}_{\operatorname{RCD}}$. Supremizing over all two-extendible states in $\operatorname{2-EXT_{RCD}}$,
\begin{align}
    &\sup_{\rho \in \operatorname{2-EXT_{RCD}}}\mathbf{E}^u\!\left(\sigma_{R_CC'C'':R_DD'D''}\right) \notag \\
    &\le \sup_{\rho \in \operatorname{2-EXT_{RCD}}}\inf_{\mathcal{P}\in \operatorname{Ext}(\mathcal{N})}\frac{1}{2}\mathbf{D}\!\left(\Theta\!\left(\mathcal{N}\right)\!\left(\rho\right)\Big\Vert \Theta\!\left(\operatorname{Tr}_{B'_1}\circ\mathcal{P}\right)\!\Big(\rho\Big)\right)\\
    &\le \sup_{\rho \in \mathcal{S}_{\operatorname{RCD}}}\inf_{\mathcal{P}\in \operatorname{Ext}(\mathcal{N})}\frac{1}{2}\mathbf{D}\!\left(\Theta\!\left(\mathcal{N}\right)\!\left(\rho\right)\Big\Vert \Theta\!\left(\operatorname{Tr}_{B'_1}\circ\mathcal{P}\right)\!\Big(\rho\Big)\right)\\
    &\le \inf_{\mathcal{P}\in \operatorname{Ext}(\mathcal{N})}\sup_{\rho \in \mathcal{S}_{\operatorname{RCD}}}\frac{1}{2}\mathbf{D}\!\left(\Theta\!\left(\mathcal{N}\right)\!\left(\rho\right)\Big\Vert \Theta\!\left(\operatorname{Tr}_{B'_1}\circ\mathcal{P}\right)\!\Big(\rho\Big)\right)\\
    &= \inf_{\mathcal{P}\in \operatorname{Ext}(\mathcal{N})} \frac{1}{2}\mathbf{D}\!\left(\Theta\!\left(\mathcal{N}\right) \big\Vert \Theta\!\left(\operatorname{Tr}_{B'_1}\circ\mathcal{P}\right)\right)\\
    &\le \inf_{\mathcal{P}\in \operatorname{Ext}(\mathcal{N})} \frac{1}{2}\mathbf{D}\!\left(\mathcal{N} \big\Vert \operatorname{Tr}_{B'_1}\circ\mathcal{P}\right)\\
    &= \mathbf{E}^u\!\left(\mathcal{N}\right),
\end{align}
where the second inequality follows from the fact that the set $\operatorname{2-EXT_{RCD}}$ is contained inside the set $\mathcal{S}_{\operatorname{RCD}}$. The third inequality is a consequence of the max-min inequality, the first equality follows from the definition of generalized divergence of channels, the last inequality follows from the data-processing inequality for generalized divergence of channels, and the last equality follows from the definition of the generalized unextendible entanglement of channels.

Therefore, we conclude the statement of Theorem~\ref{theo:unext_ent_state_le_unext_ent_ch_bip}.

\section{Proof of Proposition~\ref{prop:full_eras_unext_ent}}\label{app:full_eras_unext_ent}

In this appendix, we compute the Belavkin--Staszewski induced unextendible entanglement of the channel whose action on an arbitrary state $\rho_{RA}$ is defined as follows:
\begin{equation}\label{eq:full_eras_ch_output_st}
    \mathcal{N}_{A\to A'B}\!\left(\rho_{RA}\right) = p\rho_{RA'}\otimes |e\rangle\!\langle e|_B + (1-p)\rho_{RB}\otimes |e\rangle \!\langle e|_{A'},
\end{equation}
for some $p\in [0,1]$.

We first establish a lower bound on the unextendible entanglement of the channel induced by the Belavkin--Staszewski relative entropy.

Note that Alice can find out if she received the erased state or not by performing the POVM $\{\Pi_{A'},|e\rangle\!\langle e|_{A'}\}$ on her system. She can convey the results of the measurement to Bob via one-way classical communication, which is assumed available for free. Alice and Bob can then replace the erased state with a maximally mixed state, and Alice holds the flag indicating the result from the POVM. As such, Alice and Bob can transform the output state given in~\eqref{eq:full_eras_ch_output_st} into the following state using a one-way LOCC $\mathcal{L}^{\to}_{A'B\to A'BX_A}$:
\begin{multline}
    (\mathcal{L}^{\to}_{A'B\to A'BX_A}\circ\mathcal{N})\left(\rho_{RA}\right)
    = p\rho_{RA'}\otimes \pi_{B}\otimes |0\rangle\!\langle 0|_{X_A}\\ + (1-p)\rho_{RB}\otimes \pi_{A'} \otimes |1\rangle\!\langle 1|_{X_A},
\end{multline}
where $\pi$ is the maximally mixed state.

Let us choose the input state $\rho_{RA}$ to be $\Phi^d_{RA}$, a maximally entangled state of Schmidt rank $d$. Here we assume that system $R$ is held by Alice. Using Proposition~\ref{prop:geo_unext_direct_sum}, we have the following equality:
\begin{multline}
    \widehat{E}^u\!\left((\mathcal{L}^{\to}\circ\mathcal{N})\left(\Phi^d_{RA}\right)\right) = p\widehat{E}^u\!\left(\Phi^d_{RA'}\otimes\pi_{B}\right)\\ + (1-p)\widehat{E}^u\!\left(\Phi^d_{RB}\otimes \pi_{A}\right).
\end{multline}
The state $\Phi^d_{RA'}\otimes \pi_{B}$ is a separable state with respect to the partition $RA':B$; therefore, 
\begin{equation}
    \widehat{E}^u\!\left(\Phi^d_{RA'}\otimes\pi_{B}\right) = 0.
\end{equation}
It is easy to see that the unextendible entanglement induced by the Belavkin--Staszewski relative entropy of the state $\Phi^d_{RB}\otimes \pi_{A}$ is equal to $\log_2 d$, by means of the following reasoning:
\begin{align}
    \log_2 d &= \widehat{E}^u\!\left(\Phi^d_{R:B}\right)\\
    &= \widehat{E}^u\!\left(\operatorname{Tr}_{A'}\!\left[\Phi^d_{RB}\otimes \pi_{A'}\right]\right)\\
    &\le \widehat{E}^u\!\left(\Phi^d_{RB}\otimes \pi_{A'}\right)\\
    &\le \widehat{E}^u\!\left(\Phi^d_{RB}\right) + \widehat{E}^u\!\left(\pi_{A'}\right)\\
    &= \widehat{E}^u\!\left(\Phi^d_{R:B}\right),
\end{align}
where the first equality follows from~\eqref{eq:geo_unext_ent_edit_logd}, the first follows from the monotonicity of unextendible entanglement under local operations, the second inequality follows from the subadditivity of the uenxtendible entanglement induced by the Belavkin--Staszewski relative entropy (see~\eqref{eq:alpha_geo_unext_ent_subadditive}), and the last equality follows from the fact that the unextendible entanglement induced by the Belavkin--Staszewski relative entropy is equal to zero for two-extendible states. Therefore, 
\begin{equation}
    \widehat{E}^u\!\left((\mathcal{L}^{\to}\circ\mathcal{N})\left(\Phi^d_{RA}\right)\right) = (1-p)\log_2 d.
\end{equation}
Since $\mathcal{L}^{\to}_{A'B\to A'BX_A}$ is an instance of a two-extendible superchannel, Theorem~\ref{theo:unext_ent_state_le_unext_ent_ch_bip} implies the following inequality:
\begin{equation}\label{eq:spec_ch_unext_ent_lb}
    \widehat{E}^u\!\left(\mathcal{N}_{A\to A'B}\right) \ge (1-p)\log_2 d.
\end{equation}

Now let us establish an upper bound on the unextendible entanglement of the channel induced by the Belavkin--Staszewski relative entropy of the channel.

Consider the following extension of the channel $\mathcal{N}_{A\to A'B}$:
\begin{equation}
    \mathcal{P}_{A\to A'B_1B_2} = \mathcal{N}_{A\to A'B_1}\otimes\mathcal{A}^{\pi}_{B_2},
\end{equation}
where $\mathcal{A}^{\pi}_{B_2}$ is a channel that appends a maximally mixed state on system $B_2$, and systems $B_1$ and $B_2$ are isomorphic to each other. Let us define the following channel:
\begin{equation}
    \mathcal{M}_{A\to A'B_2} = \operatorname{Tr}_{B_1}\circ\mathcal{P}_{A\to A'B_1B_2}.
\end{equation}
The Choi operators of the two relevant marginals of the channel $\mathcal{P}_{A\to A'B_1B_2}$ are as follows:
\begin{align}
    \Gamma^{\mathcal{N}}_{AA'B_1} &= p\Gamma_{AA'}\otimes |e\rangle\!\langle e|_{B_1} + (1-p)\Gamma_{AB_1}\otimes |e\rangle\!\langle e|_{A'},\\
    \Gamma^{\mathcal{M}}_{AA'B_2} &= p\Gamma_{AA'}\otimes \pi_{B_2} + (1-p)I_A\otimes\pi_{B_2}\otimes |e\rangle\!\langle e|_{A'},
\end{align}
where $\Gamma$ is the unnormalized maximally entangled operator. Note that the erasure symbol is orthogonal to every quantum state in the Hilbert space of the system. Hence, the identity operator, $I_A$, acts on the subspace orthogonal to $|e\rangle\!\langle e|_A$. 

The Belavkin--Staszewski relative entropy between the channels $\mathcal{N}_{A\to A'B}$ and $\mathcal{M}_{A\to A'B}$ can be calculated using their Choi operators as follows~\cite{Fang_2021}:
\begin{multline}\label{eq:BS_entropy_spec_ch}
    \widehat{D}\!\left(\mathcal{N}_{A\to A'B} \Vert\mathcal{M}_{A\to A'B}\right) \\= \left\Vert\operatorname{Tr}_{A'B}\!\left[\left(\Gamma^{\mathcal{N}}_{AA'B}\right)^{1/2} \left(\log_2 Q_{AA'B}\right) \!\left(\Gamma^{\mathcal{N}}_{AA'B}\right)^{1/2}\right]\right\Vert_{\infty},
\end{multline}
where
\begin{equation}
    Q_{AA'B} = \left(\Gamma^{\mathcal{N}}_{AA'B}\right)^{1/2}\!\left(\Gamma^{\mathcal{M}}_{AA'B}\right)^{-1}\!\left(\Gamma^{\mathcal{N}}_{AA'B}\right)^{1/2}.
\end{equation}
We can write $\Gamma^{\mathcal{M}}_{AA'B}$ as a linear combination of orthogonal projectors as follows:
\begin{equation}
    \Gamma^{\mathcal{M}}_{AA'B} = p\Phi^d_{RA'}\otimes I_{B} + \frac{1-p}{d} I_{RB}\otimes |e\rangle\!\langle e|_{A'}.
\end{equation}
Therefore,
\begin{equation}
    \left(\Gamma^{\mathcal{M}}_{AA'B}\right)^{-1} = \frac{1}{p}\Phi^d_{RA'}\otimes I_{B} + \frac{d}{1-p} I_{RB}\otimes |e\rangle\!\langle e|_{A'}. 
\end{equation}
Similarly, $\Gamma^{\mathcal{N}}_{AA'B}$ can be written as a linear combination of orthogonal projectors as follows:
\begin{equation}
    \Gamma^{\mathcal{N}}_{AA'B_1} = pd\Phi^d_{RA'}\otimes |e\rangle\!\langle e|_{B} + (1-p)d\Phi^d_{RB}\otimes |e\rangle\!\langle e|_{A'}.
\end{equation}
Following simple linear algebra, we find that
\begin{equation}
    Q_{AA'B} = d^2 \Phi^d_{RB}\otimes |e\rangle\!\langle e|_{A'},
\end{equation}
and consequently,
\begin{equation}
    \log_2 Q_{AA'B} = 2\log_2 d ~\Phi^d_{RB}\otimes |e\rangle\!\langle e|_{A'}.
\end{equation}
Substituting this value into~\eqref{eq:BS_entropy_spec_ch}, we evaluate the Belavkin--Staszewski relative entropy between $\mathcal{N}_{A\to A'B}$ and $\mathcal{M}_{A\to A'B}$ to be equal to $(1-p)\log_2 d$; that is,
\begin{equation}
    \widehat{D}\!\left(\mathcal{N}_{A\to A'B}\Vert \mathcal{M}_{A\to A'B}\right) = 2(1-p)\log_2 d.
\end{equation}
By the definition of the unextendible entanglement of channels,
\begin{align}
    \widehat{E}^u\!\left(\mathcal{N}_{A\to A'B}\right) &\le \frac{1}{2}\widehat{D}\!\left(\mathcal{N}_{A\to A'B}\Vert \mathcal{M}_{A\to A'B}\right)\\
    &= (1-p)\log_2 d.\label{eq:spec_ch_unext_ent_ub}
\end{align}
Combining~\eqref{eq:spec_ch_unext_ent_lb} and~\eqref{eq:spec_ch_unext_ent_ub}, we conclude that
\begin{equation}
    \widehat{E}^u\!\left(\mathcal{N}_{A\to A'B}\right) = (1-p)\log_2 d.
\end{equation}

\section{Unextendible entanglement of erasure channels}\label{app:unext_ent_eras_ch}

In this section we find analytical and numerical upper bounds on the unextendible entanglement of the erasure channel. An erasure channel erases the input state with some probability $p$ and is defined as follows~\cite{GBP97}:
\begin{equation}
    \mathcal{E}^p_{A\to B}\!\left(Y\right) = (1-p)Y + p|e\rangle\!\langle e| \operatorname{Tr}[Y],
\end{equation}
where $|e\rangle\!\langle e|$ is the erasure symbol, orthogonal to all input states. The Choi operator of this channel is as follows:
\begin{equation}
    \Gamma^{\mathcal{E}}_{AB} = \left(1-p\right)\Gamma_{AB} + p I_A\otimes|e\rangle\!\langle e|_{B}.
\end{equation}

Erasure channels are of special interest in the context of quantum communication because there exists a well known protocol to distill a maximally entangled state using this channel, assisted by local operations and two-way classical communication. Alice sends one share of a locally prepared maximally entangled state of Schmidt rank $d$ to Bob through an erasure channel. Bob performs the projective measurement $\left\{\sum_{i=0}^{d-1}|i\rangle\!\langle i|, |e\rangle\!\langle e|\right\}$ on the state he received, thus, finding out if the quantum state sent by Alice was erased or not. Bob can convey this classical information back to Alice using a classical channel, hence, establishing a maximally entangled state between the two parties which can be used for quantum communication or private communication tasks.

\begin{figure}
    \centering
    \includegraphics[width=\linewidth]{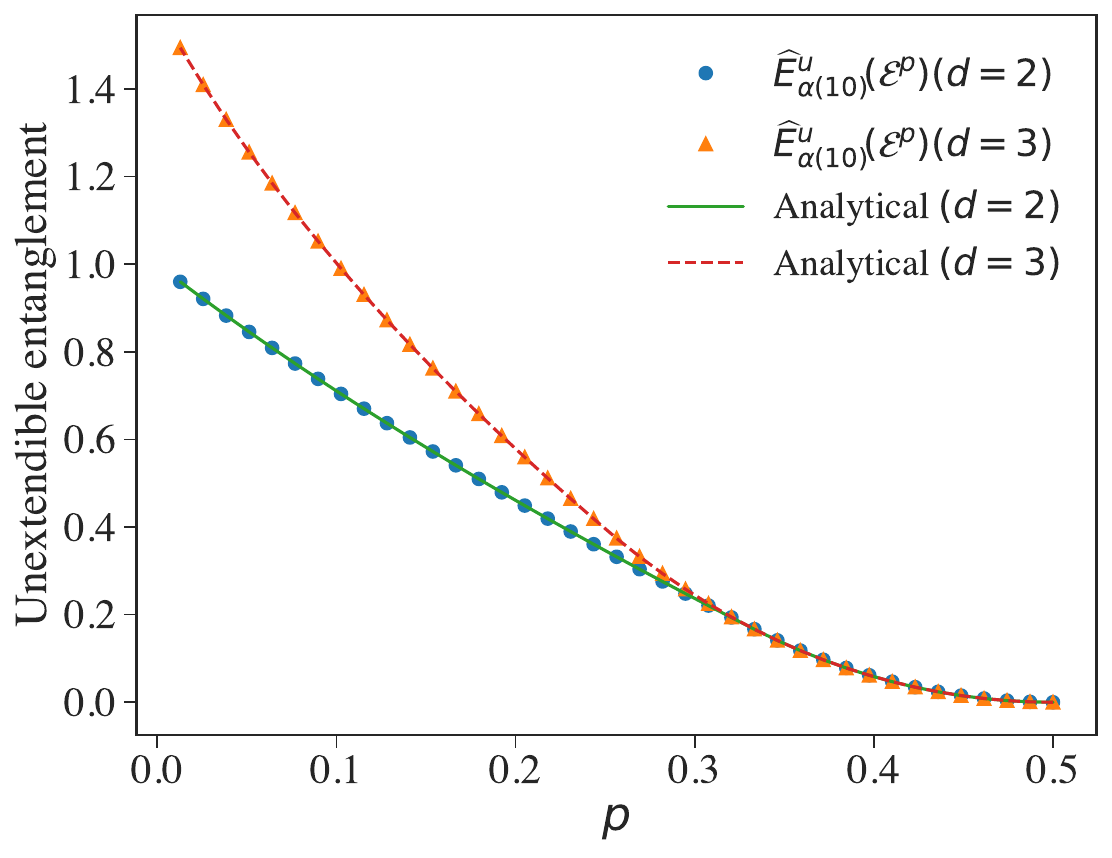}
    \caption{Here we plot the upper bounds on the unextendible entanglement of the two-dimensional and the three-dimensional erasure channel induced by the Belavkin--Staszewski relative entropy using the analytical expression given in Proposition~\ref{theo:eras_unext_ent_ub}. We also plot the numerical values of the $\alpha$-geometric unextendible entanglement calculated for $\alpha = 1+2^{-10}$ using the semidefinite program given in Proposition~\ref{prop:SDP_for_geo_unext_ent}.}
    \label{fig:eras_distill_ent}
\end{figure}

The protocol mentioned above can be used to probabilistically distill $\log_2 d$ ebits from a $d$-dimensional erasure channel with probability $1-p$, where $p$ is the erasure probability of the channel. However, it requires Bob to send back classical data to Alice. We will see that one cannot distill any entanglement from the erasure channel with the assistance of one-way LOCC superchannels only, as the min-geometric unextendible entanglement of the erasure channel is equal to zero (Proposition~\ref{prop:proof_unext_eras_alpha}). See also~\cite{LLS09} for a further study of the assisted quantum capacities of the erasure channel. 

The proposition below provides an upper bound on the unextendible entanglement of the erasure channel induced by the Belavkin--Staszewski relative entropy.

\begin{proposition}\label{theo:eras_unext_ent_ub}
    The unextendible entanglement of a $d$-dimensional erasure channel with erasure probability $p\le~1/2$, induced by the Belavkin--Staszewski relative entropy, is bounded from above by
    \begin{multline}\label{eq:eras_unext_ent_ub}
        \widehat{E}^u\!\left(\mathcal{E}^p_{A\to B}\right) 
            \le \left(1-p\right)\log_2 d- \frac{1}{2}\log_2\!\left((d^2-1)p+1\right)
    \end{multline}
    for all $p\in \left[0,\frac{1}{d+1}\right]$,
    and by
    \begin{multline}
        \widehat{E}^u\!\left(\mathcal{E}^p_{A\to B}\right) \le \frac{1}{2}\!\left(1-p\right)\log_2\!\left(\frac{1-p}{p}\right) + \frac{1}{2}p\log_2\!\left(\frac{p}{1-p}\right)
    \end{multline}
    for all $p \in \left(\frac{1}{d+1},\frac{1}{2}\right]$.
\end{proposition}

\begin{proof}
    First, we note that an erasure channel is two-extendible if the erasure probability is greater than $1/2$. Hence, the unextendible entanglement of such erasure channels, induced by the Belavkin--Staszewski relative entropy, is equal to zero.
    
Now consider the extension $\mathcal{P}_{A\to B_1B_2}$ of the erasure channel $\mathcal{E}^p_{A\to B}$ with the Choi operator,
\begin{multline}
    \Gamma^{\mathcal{P}}_{AB_1B_2} \coloneqq p\Gamma_{AB_2}\otimes|e\rangle\!\langle e|_{B_1} \\+ \left(1-p-dx\right)\Gamma_{AB_1}\otimes|e\rangle\!\langle e|_{B_2} + x\Gamma_{AB_1}\otimes \Pi_{B_2},
\end{multline}
where
\begin{equation}
\Pi\coloneqq |0\rangle\!\langle 0| + \cdots +  |d-1\rangle\!\langle d-1|    
\end{equation}
is the projection onto all possible states of the input space.
This operator is positive semidefinite for all $ x\in\left[0,(1-p)/d\right]$. The two marginals of this channel are described by the Choi operators,
\begin{equation}
    \Gamma^{\mathcal{N}}_{AB} = \operatorname{Tr}_{B_2}\!\left[\Gamma^{\mathcal{P}}_{AB_1B_2}\right] = \left(1-p\right)\Gamma_{AB} + p \Pi_A\otimes |e\rangle\!\langle e|_B,
    \label{eq:N-marginal-eras-app}
\end{equation}
and
\begin{align}
    \Gamma^{\mathcal{M}}_{AB} & = \operatorname{Tr}_{B_1}\!\left[\Gamma^{\mathcal{P}}_{AB_1B_2}\right] \\
    & = p\Gamma_{AB} + \left(1-p-dx\right) \Pi_A\otimes |e\rangle\!\langle e|_B \notag \\
    & \qquad + x\Pi_{A} \otimes \Pi_B.
    \label{eq:M-marginal-eras-app}
\end{align}
Note that $\Gamma^{\mathcal{N}}_{AB}$ is the Choi operator of the erasure channel with erasure probability $p$, justifying the claim that $\mathcal{P}_{A\to B_1B_2}$ is an extension of the said erasure channel.

By definition, 
\begin{align}
    \widehat{E}^u\!\left(\mathcal{E}^p_{A\to B}\right) \le \frac{1}{2}\widehat{D}\!\left(\mathcal{E}^p_{A\to B}\big\Vert \operatorname{Tr}_{B_2}\circ\mathcal{P}_{A\to B_1B_2}\right),
\end{align}
where $\widehat{D}\!\left(\cdot\Vert\cdot\right)$ is the Belavkin--Staszewski relative entropy between channels. This quantity has a closed-form expression in terms of the Choi operators of the two channels~\cite{Fang_2021},
\begin{multline}\label{eq:Bel_Stas_eras_ext}
    \widehat{D}\!\left(\mathcal{N}_{A\to B} \Vert\mathcal{M}_{A\to B}\right) \\= \left\Vert\operatorname{Tr}_B\!\left[\left(\Gamma^{\mathcal{N}}_{AB}\right)^{1/2} \left(\log_2 Q_{AB}\right) \!\left(\Gamma^{\mathcal{N}}_{AB}\right)^{1/2}\right]\right\Vert_{\infty},
\end{multline}
where
\begin{equation}
    Q_{AB} = \left(\Gamma^{\mathcal{N}}_{AB}\right)^{1/2}\!\left(\Gamma^{\mathcal{M}}_{AB}\right)^{-1}\!\left(\Gamma^{\mathcal{N}}_{AB}\right)^{1/2}.
\end{equation}
The Choi operator $\Gamma^{\mathcal{M}}_{AB}$ can be written as
\begin{multline}
    \Gamma^{\mathcal{M}}_{AB} =  \left(pd + x\right)\Phi^d_{AB} + x\left(\Pi_{A} \otimes \Pi_B - \Phi^d_{AB}\right) \\
    + \left(1-p-dx\right)\Pi_A\otimes |e\rangle\!\langle e|_{B},
\end{multline}
where $\Phi^d_{AB}$ is the maximally entangled state of Schmidt rank~$d$. Since we have written $\Gamma^{\mathcal{M}}_{AB}$ as a linear combination of orthogonal projections, we can conclude
\begin{multline}
    \left(\Gamma^{\mathcal{M}}_{AB}\right)^{-1} = \frac{1}{pd + x}\Phi^d_{AB} + \frac{1}{x}\!\left(\Pi_{A} \otimes \Pi_B - \Phi^d_{AB}\right) \\
     + \frac{1}{1-p-dx}\Pi_A\otimes |e\rangle\!\langle e|_{B}.
\end{multline}
From the above, we conclude that
\begin{multline}
    \left(\Gamma^{\mathcal{N}}_{AB}\right)^{1/2} \left(\log_2 Q_{AB}\right) \!\left(\Gamma^{\mathcal{N}}_{AB}\right)^{1/2} = \\
    \left(1-p\right)d\log_2\!\left(\frac{d\left(1-p\right)}  {pd+x}\right)\Phi_{AB}^d
\\         + p\log_2\!\left(\frac{p}{1-p-dx}\right) \Pi_A \otimes |e\rangle\!\langle e|_B.
\end{multline}
This allows us to evaluate the quantity in~\eqref{eq:Bel_Stas_eras_ext} to be
\begin{multline}\label{eq:geo_div_eras_in_x}
    \widehat{D}\!\left(\mathcal{N}_{A\to B}\Vert\mathcal{M}_{A\to B}\right) = \left(1-p\right)\log_2\!\left(\frac{d\left(1-p\right)}{pd+x}\right)
\\         + p\log_2\!\left(\frac{p}{1-p-dx}\right).
\end{multline}
This quantity is minimized for
\begin{equation}\label{eq:x_min_val}
    x = \frac{\left(1-p\right)^2-p^2d^2}{d}.
\end{equation}
The Choi operator $\Gamma^{\mathcal{M}}_{AB}$ is required to be a positive semidefinite operator. This in turn requires $x$ to be non-negative. Therefore, we choose
\begin{equation}
    x = \frac{\left(1-p\right)^2-p^2d^2}{d},
\end{equation}
if $p\le \frac{1}{d+1}$, and $x=0$ otherwise. Using these values of $x$ in~\eqref{eq:geo_div_eras_in_x} and rearranging, we arrive at the upper bound given in Proposition~\ref{theo:eras_unext_ent_ub}.
\end{proof}

\medskip 

In Figure~\ref{fig:eras_distill_ent}, we plot the $\alpha$-geometric unextendible entanglement of the channel, for $\alpha = 1+2^{-10}$, against the analytical upper bound on the unextendible entanglement of the channel induced by the Belavkin--Staszewski relative entropy.

We also evaluate an upper bound on the $\alpha$-geometric unextendible entanglement of the erasure channel and find an analytical expression given in the proposition below. In the limit $\alpha \to 0^+$, this quantity is equal to zero. This finding, combined with Corollaries~\ref{cor:0_err_cap_ub} and \ref{cor:0_err_priv_cap_ub}, implies that both the zero-error quantum and private capacities of the erasure channel are equal to zero. 

\begin{proposition}
\label{prop:proof_unext_eras_alpha}
    For all $\alpha \in (0,1)\cup(1,2]$, the $\alpha$-geometric unextendible entanglement of a $d$-dimensional erasure channel with erasure probability $p\le~1/2$ is bounded from above by
    \begin{multline}
    \label{eq:eras_unext_ent_ub_alpha}
        \widehat{E}^u_{\alpha}\!\left(\mathcal{E}^p_{A\to B}\right) 
            \le \\
            \frac{1}{\alpha-1}\log_{2}\left(
\begin{array}
[c]{c}%
\frac{1}{d}\left[  d\left(  1-p\right)  \right]  ^{\alpha}\left(  pd+x\right)
^{1-\alpha}\\
+p^{\alpha}\left(  1-p-dx\right)  ^{1-\alpha}%
\end{array}
\right)  .
    \end{multline}
    for all $p\in \left[0,\frac{1}{d^{1/\alpha}+1}\right]$,
    and where
    \begin{align}
    x & =\frac{1-p-pdk}{k+d} , \\
    k & =\frac{d^{2/\alpha}p}{d\left(  1-p\right)  }.
    \end{align}
For all $p \in \left(\frac{1}{d^{1/\alpha}+1},\frac{1}{2}\right]$,
    \begin{multline}
    \label{eq:eras_unext_ent_ub_alpha_other}
        \widehat{E}_{\alpha}^u\!\left(\mathcal{E}^p_{A\to B}\right) \le \\
        \frac{1}{\alpha-1}\log_{2}\left(  \left(  1-p\right)  ^{\alpha}%
p^{1-\alpha}+p^{\alpha}\left(  1-p\right)  ^{1-\alpha}\right).
    \end{multline}
    As such, for all $p\in (0,1/2]$,
    \begin{equation}
         \widehat{E}^u_{\min}\!\left(\mathcal{E}^p_{A\to B}\right) = 0.
    \end{equation}
\end{proposition}

\begin{proof}
Here we follow the same approach used in the proof of Proposition~\ref{theo:eras_unext_ent_ub}. Let us first recall
from~\cite[Proposition~44]{Katariya2021} that the geometric R\'enyi relative entropy of channels can be written
explicitly as 
\begin{equation}
\widehat{D}_{\alpha}(\mathcal{N}\Vert\mathcal{M})=\frac{1}{\alpha-1}\log
_{2}\widehat{Q}_{\alpha}(\mathcal{N}\Vert\mathcal{M}),
\end{equation}
where
\begin{multline}
\widehat{Q}_{\alpha}(\mathcal{N}\Vert\mathcal{M})\coloneqq \\
\left\Vert \operatorname{Tr}_{B}\!\left[  \left(  \Gamma^{\mathcal{M}}\right)
^{1/2}\left[  \left(  \Gamma^{\mathcal{M}}\right)  ^{-1/2}\Gamma^{\mathcal{N}%
}\left(  \Gamma^{\mathcal{M}}\right)  ^{-1/2}\right]  ^{\alpha}\left(
\Gamma^{\mathcal{M}}\right)  ^{1/2}\right]  \right\Vert _{\infty}%
\end{multline}
when $\alpha \in (1,2]$ and 
\begin{multline}
\widehat{Q}_{\alpha}(\mathcal{N}\Vert\mathcal{M})\coloneqq \\
\lambda_{\min}\!\left( \operatorname{Tr}_{B}\!\left[  \left(  \Gamma^{\mathcal{M}}\right)
^{1/2}\left[  \left(  \Gamma^{\mathcal{M}}\right)  ^{-1/2}\Gamma^{\mathcal{N}%
}\left(  \Gamma^{\mathcal{M}}\right)  ^{-1/2}\right]  ^{\alpha}\left(
\Gamma^{\mathcal{M}}\right)  ^{1/2}\right]  \right)%
\end{multline}
when $\alpha \in (0,1)$, with both expressions above holding under the assumption that supp$(\Gamma^{\mathcal{N}}) \subseteq $supp$(\Gamma^{\mathcal{M}})$.
Now using the expressions in~\eqref{eq:N-marginal-eras-app} and~\eqref{eq:M-marginal-eras-app}, we find that
\begin{multline}
\left(  \Gamma^{\mathcal{M}}\right)  ^{1/2}\left[  \left(  \Gamma
^{\mathcal{M}}\right)  ^{-1/2}\Gamma^{\mathcal{N}}\left(  \Gamma^{\mathcal{M}%
}\right)  ^{-1/2}\right]  ^{\alpha}\left(  \Gamma^{\mathcal{M}}\right)
^{1/2}\\
=\left[  d\left(  1-p\right)  \right]  ^{\alpha}\left(  pd+x\right)
^{1-\alpha}\Phi_{AB}^{d}\\
+p^{\alpha}\left(  1-p-dx\right)  ^{1-\alpha}\Pi_{A}\otimes|e\rangle\!\langle
e|_{B}.
\end{multline}
This in turn implies that%
\begin{multline}
\widehat{D}_{\alpha}(\mathcal{N}\Vert\mathcal{M})=\\
\frac{1}{\alpha-1}\log_{2}\left(
\begin{array}
[c]{c}%
\frac{1}{d}\left[  d\left(  1-p\right)  \right]  ^{\alpha}\left(  pd+x\right)
^{1-\alpha}\\
+p^{\alpha}\left(  1-p-dx\right)  ^{1-\alpha}%
\end{array}
\right)  .
\end{multline}
This quantity is minimized for the choice%
\begin{equation}
x=\frac{1-p-pdk}{k+d},
\end{equation}
where $k=\frac{d^{2/\alpha}p}{d\left(  1-p\right)  }$. In order for
$\Gamma^{\mathcal{M}}$ to be positive semidefinite, it is required that
$x\geq0$, which is the same as $p\leq\frac{1}{d^{1/\alpha}+1}$. So when this
condition holds, we choose $x$ as above, and otherwise choose $x=0$. In the
latter case, we find that%
\begin{align}
&  \frac{1}{\alpha-1}\log_{2}\left(
\begin{array}
[c]{c}%
\frac{1}{d}\left[  d\left(  1-p\right)  \right]  ^{\alpha}\left(  pd+x\right)
^{1-\alpha}\\
+p^{\alpha}\left(  1-p-dx\right)  ^{1-\alpha}%
\end{array}
\right)  \nonumber\\
&  =\frac{1}{\alpha-1}\log_{2}\left(  \frac{1}{d}\left[  d\left(  1-p\right)
\right]  ^{\alpha}\left(  pd\right)  ^{1-\alpha}+p^{\alpha}\left(  1-p\right)
^{1-\alpha}\right)  \\
&  =\frac{1}{\alpha-1}\log_{2}\left(  \left(  1-p\right)  ^{\alpha}%
p^{1-\alpha}+p^{\alpha}\left(  1-p\right)  ^{1-\alpha}\right)  .
\end{align}
This leads to the inequalities in~\eqref{eq:eras_unext_ent_ub_alpha} and~\eqref{eq:eras_unext_ent_ub_alpha_other}.

To establish the limit when $\alpha\rightarrow0^{+}$, we simply set $x=0$ and
then take the limit as $\alpha\rightarrow0^{+}$, leading to%
\begin{align}
&  \lim_{\alpha\rightarrow0^{+}}\frac{1}{\alpha-1}\log_{2}\left(  \left(
1-p\right)  ^{\alpha}p^{1-\alpha}+p^{\alpha}\left(  1-p\right)  ^{1-\alpha
}\right) \nonumber\\
&  =-\log_{2}\left(  \left(  1-p\right)  ^{0}p^{1}+p^{0}\left(  1-p\right)
^{1}\right) \\
&  =0.
\end{align}
This completes the proof.
\end{proof}

\section{Unextendible entanglement of depolarizing channels}\label{app:proof_unext_ent_dep}

Depolarizing channels are commonly used to model noise in quantum circuits. The $d$-dimensional depolarizing channel~$\mathcal{D}_p$ is a completely positive trace-preserving map when the parameter $p\in \left[0,\frac{d^2}{d^2-1}\right]$, and it acts on a quantum state $\rho$ as
\begin{equation}\label{eq:dep_ch_defn}
	\mathcal{D}_p\!\left(\rho\right) = (1-p)~\rho + p\pi,
\end{equation} 
where $\pi\coloneqq I/d$ is the $d$-dimensional maximally mixed state. The Choi operator of a depolarizing channel $\mathcal{D}_p$ is,
\begin{equation}
	\Gamma^{\mathcal{D}_p}_{AB} = (1-p)\Gamma_{AB} + p I_A\otimes \pi_B.
\end{equation}
The Choi operator can be written as a linear combination of orthogonal projectors as
\begin{equation}
    \Gamma^{\mathcal{D}_p}_{AB} = d\left(F\Phi^d_{AB} + (1-F)\frac{I_{AB}-\Phi^d_{AB}}{d^2-1}\right),
\end{equation}
where
\begin{equation}
    F \coloneqq  1-p + \frac{p}{d^2}.
\end{equation}

\begin{figure}
	\centering
		\includegraphics[width=\linewidth]{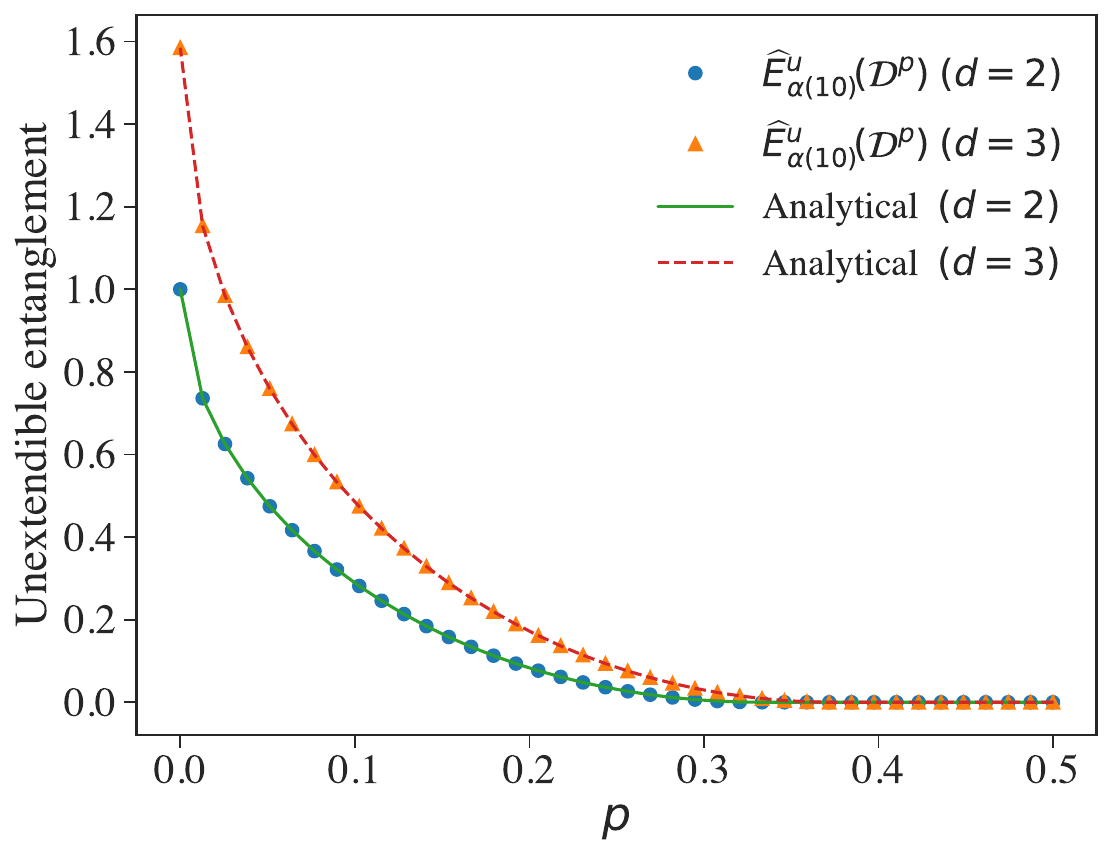}
		\caption{Here we plot the unextendible entanglement of the two-dimensional and the three-dimensional depolarizing channel induced by the Belavkin--Staszewski relative entropy using the analytical expression given in Proposition~\ref{theo:unext_ent_dep}. We also plot the numerical values of the $\alpha$-geometric unextendible entanglement calculated for $\alpha = 1+2^{-10}$ using the semidefinite program given in Proposition~\ref{prop:SDP_for_geo_unext_ent}.}
		\label{fig:Dep_2d}
\end{figure}

Since the Choi operator of the depolarizing channel is a full-rank operator for $p>0$, the min-geometric unextendible entanglement of this channel is equal to zero, which implies that the zero-error quantum capacity and the zero-error private capacity of this channel, assisted by one-way LOCC or two-extendible superchannels, are also equal to zero (see Corollary~\ref{cor:q_cap_p_cap_0_full_rank}). 

The extendibility of isotropic states has been studied in~\cite{Johnson_2013}, and since the Choi operator of the depolarizing channel is a scaled isotropic state, we can find an analytical expression for the $\alpha$-geometric unextendible entanglement of the channel. Let us first look at the extremities. Since a point-to-point quantum channel is two-extendible if and only if its Choi state is two-extendible, the $d$-dimensional depolarizing channel is two-extendible for $p\ge \frac{d}{2(d+1)}$~\cite[Theorem III.8]{Johnson_2013} (also see~\cite[Lemma 3]{Kaur_2021}). This implies that the $\alpha$-geometric unextendible entanglement of a depolarizing channel with $p\ge \frac{d}{2(d+1)}$ is equal to zero. For $p=0$, the depolarizing channel is the same as the identity channel; hence, the $\alpha$-geometric unextendible entanglement in this case is equal to one.
\begin{proposition}\label{theo:unext_ent_dep}
    The unextendible entanglement of the $d$-dimensional depolarizing channel, with parameter $p < \frac{d}{2(d+1)}$, induced by the Belavkin--Staszewski relative entropy is
    \begin{equation}
        \widehat{E}^u\!\left(\mathcal{D}^p\right)= \frac{1}{2}\!\left[F\log_2\!\left(\frac{F}{F'}\right) + (1-F)\log_2\!\left(\frac{1-F}{1-F'}\right)\right],
    \end{equation}
    where $F = 1-p + \frac{p}{d^2}$ 
    and 
    \begin{equation}
        F' \coloneqq  \max\left\{\frac{2 F-1}{d^2}+\frac{2  \sqrt{\left(d^2-1\right) (1-F) F}}{d^2}-F+1,F\right\}.
    \end{equation}
\end{proposition}
\begin{proof}
    Consider a $d$-dimensional depolarizing channel $\mathcal{D}^p_{A\to B}$ with parameter $p$ as defined in~\eqref{eq:dep_ch_defn}. The Choi operator of this channel is
\begin{equation}\label{eq:dep_ch_choi_op_F}
    \Gamma^{\mathcal{D}^p}_{AB} = d\left(F\Phi^d_{AB} + (1-F)\frac{I_{AB}-\Phi^d_{AB}}{d^2-1}\right),
\end{equation}
where
\begin{equation}
    F = 1-p + \frac{p}{d^2}.
\end{equation}
The depolarizing channel does not change under the action of a twirling superchannel; that is,
\begin{equation}
    \mathcal{T}_{AB}\!\left(\mathcal{D}^p_{A\to B}\right) = \int dU ~\mathcal{U}_{B}\circ\mathcal{D}^p_{A\to B}\circ\mathcal{U}^{\dagger}_A = \mathcal{D}^p_{A\to B},
\end{equation}
where $\mathcal{U}$ is the unitary channel corresponding to the unitary $U$ and acts as $\mathcal{U}\!\left(\cdot\right) = U\left(\cdot\right)U^{\dagger}$ and the integration is taken over the Haar measure. Twirling an arbitrary point-to-point quantum channel results in a depolarizing channel~\cite{HHH99,Nielsen_2002}.

Let $\mathcal{P}_{A\to B_1B_2}$ be a extension of $\mathcal{D}^p_{A\to B}$ lying in the set $\operatorname{Ext}\!\left(\mathcal{D}^p\right)$. Consider the following tripartite twirling superchannel:
\begin{multline}
    \mathcal{T}_{AB_1B_2}\!\left(\mathcal{P}_{A\to B_1B_2}\right) \\
    \coloneqq \int dU ~\mathcal{U}_{B_1}\circ\mathcal{U}_{B_2}\circ\mathcal{P}_{A\to B_1B_2}\circ\mathcal{U}^{\dagger}_A.
\end{multline}
The quantum channel $\mathcal{T}_{AB_1B_2}\!\left(\mathcal{P}_{A\to B_1B_2}\right)$ also lies in the set $\operatorname{Ext}\!\left(\mathcal{D}^p\right)$ since
\begin{equation}
    \operatorname{Tr}_{B_2}\circ\mathcal{T}_{AB_1B_2}\!\left(\mathcal{P}\right) = \mathcal{T}_{AB_1}\!\left(\operatorname{Tr}_{B_2}\circ\mathcal{P}\right) = \mathcal{D}^p_{A\to B_1},
\end{equation}
which follows from the trace-preserving nature of the channel~$\mathcal{U}_{B_2}$.
Moreover, the other marginal of this channel is also a depolarizing channel, with some parameter $p'$ as shown below:
\begin{equation}
    \operatorname{Tr}_{B_1}\circ\mathcal{T}_{AB_1B_2}\!\left(\mathcal{P}\right) = \mathcal{T}_{AB_2}\!\left(\operatorname{Tr}_{B_1}\circ\mathcal{P}\right) = \mathcal{D}^{p'}_{A\to B_2}.
\end{equation}

As stated above, twirling a quantum channel is a valid superchannel, and the generalized divergence between two quantum channels decreases upon twirling due to the data-processing inequality. This implies that
\begin{align}
    &\mathbf{D}\!\left(\operatorname{Tr}_{B_2}\circ\mathcal{T}\!\left(\mathcal{P}_{A\to B_1B_2}\right)\Vert\operatorname{Tr}_{B_1}\circ\mathcal{T}\!\left(\mathcal{P}_{A\to B_1B_2}\right)\right)\notag \\
    &= \mathbf{D}\!\left(\mathcal{T}\!\left(\mathcal{D}^p_{A\to B}\right)\Vert\mathcal{T}\!\left(\operatorname{Tr}_{B_1}\circ\mathcal{P}_{A\to B_1B_2}\right)\right) \\
    &\le  \mathbf{D}\!\left(\mathcal{D}^p_{A\to B}\Vert\operatorname{Tr}_{B_1}\circ\mathcal{P}_{A\to B_1B_2}\right).
\end{align}
This further implies that
\begin{multline}
    \inf_{\mathcal{P}_{A\to B_1B_2}\in \operatorname{Ext}\left(\mathcal{D}^p\right)}\mathbf{D}\!\left(\mathcal{D}^p_{A\to B}\Vert\operatorname{Tr}_{B_1}\circ\mathcal{T}\!\left(\mathcal{P}_{A\to B_1B_2}\right)\right) \le \\ \inf_{\mathcal{P}_{A\to B_1B_2}\in \operatorname{Ext}\left(\mathcal{D}^p\right)}\mathbf{D}\!\left(\mathcal{D}^p_{A\to B}\Vert\operatorname{Tr}_{B_1}\circ\mathcal{P}_{A\to B_1B_2}\right),
\end{multline}
and hence, we only need to consider the extensions $\mathcal{P}_{A\to B_1B_2}\in \operatorname{Ext}\!\left(\mathcal{D}^p\right)$ that are invariant under the tripartite twirl $\mathcal{T}_{AB_1B_2}$ when computing the unextendible entanglement of the channel. Since the marginals of such channels are always depolarizing channels, we can write,
\begin{multline}\label{eq:unext_ent_dep_twirl_opt}
    \mathbf{E}^u\!\left(\mathcal{D}^p_{A\to B}\right) \\= \inf_{\mathcal{P}\in \operatorname{Ext}\left(\mathcal{N}\right)}\frac{1}{2}\left\{\begin{array}{c}
         \mathbf{D}\!\left(\mathcal{D}^p_{A\to B}\Vert\mathcal{D}^{p'}_{A\to B}\right):\\
         \mathcal{P}_{A\to B_1B_2} = \mathcal{T}\!\left(\mathcal{P}_{A\to B_1B_2}\right),\\
         \mathcal{D}^{p'}_{A\to B} = \operatorname{Tr}_{B_1}\circ\mathcal{P}_{A\to B_1B_2}
    \end{array}\right\}.
\end{multline}
The Belavkin--Staszewski relative entropy between two depolarizing channels $\mathcal{D}^p_{A\to B}$ and $\mathcal{D}^{p'}_{A\to B}$, using the analytical expression given in~\cite[Theorem 3]{Fang_2021}, evaluates to the following quantity: 
\begin{equation}\label{eq:Bel_Stas_ch_div_dep}
    \widehat{D}\!\left(\mathcal{D}^p\Vert\mathcal{D}^{p'}\right) = F\log_2\!\left(\frac{F}{F'}\right) + (1-F)\log_2\!\left(\frac{1-F}{1-F'}\right),
\end{equation}
where
\begin{align}
    F &= 1-p + \frac{p}{d^2},\label{eq:F_defn}\\
    F' &= 1-p' + \frac{p'}{d^2}.\label{eq:F'_defn}
\end{align}

Let $\zeta^F_{AB}$ denote an isotropic state with parameter $F$:
\begin{equation}
   \zeta^F_{AB} \coloneqq F\Phi^d_{AB} + (1-F)\frac{I_{AB}-\Phi^d_{AB}}{d^2-1}. 
\end{equation}
Since the Choi operator of the depolarizing channel is an isotropic state with a scaling factor (see~\eqref{eq:dep_ch_choi_op_F}), the following two statements are equivalent:
\begin{enumerate}
    \item There exists a quantum channel $\mathcal{P}_{A\to B_1B_2}$ with the depolarizing channels, $\mathcal{D}^p_{A\to B_1}$ and $\mathcal{D}^{p'}_{A\to B_2}$, as its marginals.
    \item There exists a quantum state $\tau_{AB_1B_2}$ with the isotropic states, $\zeta^F_{AB_1}$ and $\zeta^{F'}_{AB_2}$, as its marginals, where $F$ and $F'$ are given in~\eqref{eq:F_defn} and~\eqref{eq:F'_defn}, respectively.
\end{enumerate} 
Therefore, we can compute the unextendible entanglement of a depolarizing channel, using the measure induced by the Belavkin--Staszewski relative entropy as follows:
\begin{multline}\label{eq:unext_ent_Bel_Stas_dep_red}
    \widehat{E}^u\!\left(\mathcal{D}^p_{A\to B}\right) = \\ \inf_{F'\in [0,1]}
    \left\{
    \begin{array}{c}
         \widehat{D}\!\left(\mathcal{D}^p\Vert\mathcal{D}^{p'}\right):  \\
         \tau_{AB_1B_2} \in \mathcal{S}\left(AB_1B_2\right)\\
         \operatorname{Tr}_{B_2}\!\left[\tau\right] = \zeta^F_{AB_1},
         \operatorname{Tr}_{B_1}\!\left[\tau\right] = \zeta^{F'}_{AB_2},\\
         F = 1 - p + \frac{p}{d^2},
         F' = 1 - p' + \frac{p'}{d^2}
    \end{array}
    \right\}.
\end{multline}
Moreover, the infimum can be replaced with a minimum due to the lower-semicontinuity of the Belavkin--Staszewski relative entropy.

It has been shown in~\cite[Corollary III.4]{Johnson_2013} that there exists a quantum state $\tau_{AB_1B_2}$ with marginals $\zeta^F_{AB}$ and $\zeta^{F'}_{AB}$ if and only if $F$ and $F'$ lie in the convex hull of the ellipse
\begin{equation}
    \frac{\left(F+F'-1\right)^2}{1/d^2} + \frac{\left(F-F'\right)^2}{\left(d^2-1\right)/d^2} = 1,
\end{equation}
and the point $(F,F') = (0,0)$.
Rewriting the equation of ellipse in the form of a quadratic equation in $F'$,
\begin{align}
    &\left(F+F'-1\right)^2 + \frac{(F-F')^2}{d^2-1} = \frac{1}{d^2}\\
    \Rightarrow \quad &(F')^2\left(1+\frac{1}{d^2-1}\right) + F'\left(2(F-1) - \frac{2F}{d^2-1}\right) \notag \\
    &\quad + (F-1)^2 + \frac{F^2}{d^2-1} - \frac{1}{d^2} = 0.\label{eq:ellipse_eq_quad_F'}
\end{align}
Let $F'_{\operatorname{hi}}$ and $F'_{\operatorname{lo}}$ be the two solutions of this quadratic equation such that $F'_{\operatorname{hi}} \ge F'_{\operatorname{lo}}$. Since $(F,F')$ can reside anywhere in the convex hull of the ellipse and the point $(0,0)$, the largest value $F'$ can take for a fixed value of $F$ lies on the boundary of the ellipse, and hence, is the larger of the two solutions of the quadratic equation in~\eqref{eq:ellipse_eq_quad_F'} which is $F'_{\operatorname{hi}}$. Solving~\eqref{eq:ellipse_eq_quad_F'}, we find that
\begin{equation}
    F'_{\operatorname{hi}} =\frac{2 F-1}{d^2}+\frac{2  \sqrt{\left(d^2-1\right) (1-F) F}}{d^2}-F+1.
\end{equation}

The Belavkin--Staszewski relative entropy between two depolarizing channels as given in~\eqref{eq:Bel_Stas_ch_div_dep} is minimized when $F'$ is the closest to $F$. Therefore, the optimal value of $F'$ is achieved by $\max\!\left\{F'_{\operatorname{hi}},F\right\}$. Hence, we can substitute this optimal value of $F'$ in~\eqref{eq:Bel_Stas_ch_div_dep} and~\eqref{eq:unext_ent_Bel_Stas_dep_red} to arrive at the following analytical expression for the unextendible entanglement of the depolarizing channel, using the measure induced by the Belavkin--Staszewski relative entropy:
\begin{multline}
    \widehat{E}^u\!\left(\mathcal{D}^p_{A\to B}\right) \\= \frac{1}{2}\left\{\begin{array}{c}
         F\log_2\!\left(\frac{F}{F'}\right) + (1-F)\log_2\!\left(\frac{1-F}{1-F'}\right):\\
         F' = \max\left\{\frac{2 F-1}{d^2}+\frac{2  \sqrt{\left(d^2-1\right) (1-F) F}}{d^2}-F+1,F\right\},\\
         F = 1-p + \frac{p}{d^2}
    \end{array}\right\}.
\end{multline}
This completes the proof.
\end{proof}

\medskip
 In Figure~\ref{fig:Dep_2d} we plot the $\alpha$-geometric unextendible entanglement of the two-dimensional and three-dimensional depolarizing channels for $\alpha = 1+ 2^{-10}$, with respect to the prameter $p$. We also plot the analytical expression for the unextendible entanglement of the depolarizing channel induced by the Belavkin--Staszewski relative entropy given in Proposition~\ref{theo:unext_ent_dep}.

\end{document}